%% file: main.tex
\newif\ificlr
\iclrfalse
\newif\ifmalgai
\malgaifalse

\documentclass[11pt]{article}
\usepackage{graphicx} 
\usepackage{fullpage}
\usepackage{subcaption}
\input{includes}

\title{Personalization Aids Pluralistic Alignment Under Competition}
\author{Natalie Collina, Surbhi Goel, Aaron Roth, Mirah Shi\\ \\ 
University of Pennsylvania}

\begin{document}

\maketitle

\begin{abstract}
   \input{files_arxiv/abstract}
\end{abstract}


\section{Introduction}
\input{files_arxiv/intro}

\subsection{Related Work}
\input{files_arxiv/related}

\section{Preliminaries}\label{sec:prelims}
\input{files_arxiv/prelims}

\section{The Personalized Game}\label{sec:personalized}
\input{files_arxiv/personalized}

\section{The Anonymous Game}\label{sec:anonymous}

\input{files_arxiv/anonymous}

\section{Robustness to New Users}\label{sec:robustness}
\input{files_arxiv/robustness}

\section{Experiments}\label{sec:experiments}

\input{files_arxiv/experiments}

\subsection*{Acknowledgments}
The authors gratefully acknowledge support from the UK AI Safety Institute (AISI), the NSF ENCoRE TRIPODS Institute, the Simons Collaboration on Algorithmic Fairness, a Schmidt Sciences AI2050 fellowship, an Amazon Research Award, and a Google PhD Fellowship.

{\color{purple}\textit{We dedicate this paper to the memory of Emily Ryu. We are grateful for her contributions to the ideas that inspired this work. We miss her.}}
\bibliographystyle{plainnat}
\bibliography{main}

\appendix
\input{files_arxiv/appendix_proofs}

\section{Additional Experiments}\label{app:experiments}
\input{files_arxiv/appendix_experiments}
\end{document}

%% file: includes.tex
\usepackage[utf8]{inputenc}
\usepackage{amsmath,amssymb,amsthm}
\usepackage{bbm}
\usepackage{natbib}
\usepackage{thmtools}

\DeclareMathOperator*{\argmax}{arg\,max}
\DeclareMathOperator*{\argmin}{arg\,min}
 
\usepackage{hyperref}
\usepackage{cleveref}
\newtheorem{theorem}{Theorem}[section]

\newtheorem{definition}[theorem]{Definition}

\newtheorem{proposition}[theorem]{Proposition}

\usepackage{xcolor}

\usepackage{comment}
\usepackage{mathtools}

\usepackage[shortlabels]{enumitem}

\newcommand{\E}{\mathbb{E}}

\usepackage[shortlabels]{enumitem}

\newcommand{\ignore}[1]{}

\usepackage[shortlabels]{enumitem}

\usepackage[shortlabels]{enumitem}

\newcommand{\M}{\mathcal{M}}

\usepackage{caption}
\usepackage{xcolor}

\let\Pr\undefined

\DeclareMathOperator*{\Pr}{Pr}

\newcommand{\cA}{\mathcal{A}}

\newcommand{\cI}{\mathcal{I}}

\newcommand{\cM}{\mathcal{M}}

\newcommand{\cX}{\mathcal{X}}
\newcommand{\cY}{\mathcal{Y}}

\newcommand{\eps}{{\varepsilon}}

\newcommand{\R}{\mathbb{R}}
\newcommand{\1}{\mathbbm{1}}

\usepackage{float}

\usepackage{nicematrix}

\usepackage{algorithm}
\usepackage[noend]{algpseudocode}

%% file: files_arxiv/abstract.tex
Can competition among misaligned AI providers yield aligned outcomes for a diverse population of users, and what role does model personalization play? We study a setting where multiple competing AI providers interact with multiple users who must make downstream decisions but differ in preferences. Providers have their own objectives over users’ actions and strategically deploy AI models to advance them.
We model the interaction as a Stackelberg game with multiple leaders (providers) and followers (users): providers commit to conversational policies, and users choose which model to use, how to converse, and how to act. With user-specific personalization, we show that under a Weak Market Alignment condition, every equilibrium gives each user outcomes comparable to those from a perfectly aligned common model—so personalization can induce pluralistically aligned outcomes, \textit{even when} providers are self-interested. In contrast, when providers must deploy a single anonymous policy, there exist equilibria with uninformative behavior under the same condition. We then give a stronger alignment condition that guarantees each user their optimal utility in the anonymous setting.

%% file: files_arxiv/intro.tex
AI alignment is often framed as a single-developer problem: how to encode a particular set of goals and values into a model. A core challenge is that users disagree about what ``alignment'' should mean. This has motivated work on \emph{pluralistic} alignment guarantees for a single trained model \citep{sorensen2024position,shirali2025direct,huang2024collective,klassen2024pluralistic}. We take a broader view and ask when the \emph{marketplace} for AI models itself can induce pluralistic alignment for a diverse population of users.
\ifmalgai 
\else

\fi
Crucially, providers may be misaligned with many of their users and cannot be relied on to align in isolation. For example, an AI system produced (or sponsored) by a drug manufacturer for diagnosis and treatment may be incentivized to recommend products sold by its creator, diverging from the goals of doctors using the system. Technical advances in alignment do not resolve this incentive mismatch if the provider has little reason to align to the user in the first place.
\ifmalgai 
\else

\fi
However, providers do not operate in a vacuum: they operate in a competitive marketplace, and their users have a choice about which model to use.  \citet{collina2025emergentalignmentcompetition} study provider competition in a market with a single downstream user and show that even if every provider is substantially misaligned, if providers are sufficiently \emph{differently} misaligned---in the sense that the user's utility lies in the non-negative span of providers' utilities---then competition forces perfect alignment in equilibrium. They explicitly leave pluralistic alignment with multiple downstream users as an open question.

We extend and generalize \citet{collina2025emergentalignmentcompetition} to many providers and many users, each with different objectives. Providers commit to AI systems modeled as arbitrary \emph{conversation rules} mapping private information and a conversation prefix to a next message. Users then choose which model to consult, how to interact, and what action to take. Because different users can select different models and use them differently, pluralistic alignment may emerge endogenously; we study when it does.
\ifmalgai 
\else

\fi
An important distinction emerges from our framework: whether or not AI providers can produce \emph{personalized} models or only \emph{anonymous} models. AI providers who can provide personalized models are able to commit to \emph{different} conversation rules for different downstream users. This models the ability of many modern consumer-facing chat-bots to maintain memory across conversations and profile their users in other ways (via e.g. browser history) which allows them to personalize to their users. We also study AI providers who must deploy anonymous models --- i.e. those that instantiate identical conversation rules across all users. This models e.g. open weight models (which any user can download and use) as well as anonymous API access. 
\ifmalgai 
\else

\fi
We give a simple market alignment condition---roughly, a decomposition across users and providers of the condition in \citep{collina2025emergentalignmentcompetition}---under which, in \emph{all} equilibria of the personalized commitment game, every user attains the benefits of a fully aligned model, despite heterogeneous user utilities and provider misalignment. In contrast, in the anonymous game, even under the same condition, alignment need not obtain: we construct an example where a Nash equilibrium yields \emph{no} non-trivial utility from using the deployed models. Thus, personalization can be \emph{helpful} for pluralistic alignment, even when no provider is aligned with any user. Finally, we introduce a stronger, in a sense dual, alignment condition---that there exists a subset of providers whose utilities lie in the non-negative span of their users' utilities---which is strictly stronger than our first condition. Under this stronger condition (and additional structural assumptions), pluralistic alignment is guaranteed in all equilibria even when providers must deploy anonymous models.
\ifmalgai
We defer discussion of related work to Section \ref{sec:related-work}.
\fi

\subsection{Our model and results}
We study a marketplace where multiple AI providers commit to conversational policies, and downstream users choose whom to consult and how to use the interaction to select an action. We formalize this as a multi-leader, multi-follower Stackelberg game: providers commit first; each user then best-responds by choosing a provider, interacting, and acting. Our key distinction is whether providers can condition on user identity. In the \emph{personalized game}, each provider commits to a potentially different policy for each user, and users observe only their own available policies. In the \emph{anonymous game}, each provider commits to a single publicly observable policy used for all users. Informally, these capture account-based personalization versus anonymous/open-access deployment.

\paragraph{A benchmark based on common information.}
Since providers may hold different private information, not all policies are feasible for all providers. We therefore benchmark welfare using only providers' \emph{common information}. We formalize this via \emph{shared conversation rules} which informally are the conversation rules that can be implemented by any of the providers, and define each user's optimal shared rule $C^*_S(i)$. Our theorems lower-bound equilibrium utility by $u_i(C^*_S(i))$, i.e., the best utility user $i$ could achieve from a perfectly aligned provider informed only by what all providers commonly know.

\paragraph{Market-level alignment conditions.}
We introduce two structural assumptions linking provider and user objectives. \emph{Weak Market Alignment} (Definition~\ref{def:separability}) requires providers' utilities to be (approximately) additively separable across users, with each user's utility (approximately) a nonnegative combination of the corresponding components. \emph{Strong Market Alignment} (Definition~\ref{def:provider-alignment}) requires each provider's utility to be (approximately) a nonnegative combination of users' utilities; we show this implies Weak Market Alignment.
In this framework, we prove several results: 

\begin{itemize}
\item In the personalized game, competition plus Weak Market Alignment suffices to guarantee that \emph{every} user receives near-benchmark utility in \emph{every} Nash equilibrium (Theorem~\ref{thm:private}).
In this sense, personalization ``decomposes'' competition across users: even if no provider is globally aligned, equilibrium behavior must still be approximately user-optimal for each user.
\item In contrast, under the same Weak Market Alignment condition, the anonymous game can have equilibria in which all users receive very low utility: we construct such an example where equilibrium reveals essentially no useful information (Theorem~\ref{thm:public-example}).
This captures a core tension of anonymous/public policies: a provider may be pushed to withhold information to avoid outcomes it dislikes for some users, even when that hurts all users.
\item We then study when anonymous competition can still deliver strong user guarantees.
Under Strong Market Alignment plus an additional assumption on the marginal value of an additional round of conversation, we prove anonymous-game utility guarantees (Theorem~\ref{thm:public-natalie-general}).
We also prove a clean special case: if each relevant provider has an \emph{user-dominant} conversation rule (e.g., when it is possible for the providers to fully reveal their private observations), then each user attains near-benchmark utility in every anonymous-game equilibrium (Theorem~\ref{thm:public-natalie}).
\item Finally, we study how guarantees change when a new user with arbitrary preferences --- possibly breaking the alignment conditions --- enters the market and perturbs providers' utilities.
In the personalized setting, the original users' guarantees persist (Theorem~\ref{thm:private-adding-users}); in the anonymous setting, we show that adding a user can break previously good guarantees (Theorem~\ref{thm:public-adding-users}).
\item We empirically test our alignment conditions (\Cref{sec:experiments}) using OpinionQA~\citep{pmlr-v202-santurkar23a}, with demographic groups as users and LLMs as providers. We find that Weak Market Alignment holds: user utilities are well-approximated by nonnegative combinations of provider utilities, improving with market size, though some groups are harder to align than others. We find that Strong Market Alignment is more difficult to satisfy: single-provider coverage is poor, but adding one to two more providers dramatically improves it. However, as theory predicts, approximation errors for Strong alignment are roughly 40\% larger than for Weak.
\end{itemize}

%% file: files_arxiv/related.tex
\paragraph{Bayesian Persuasion and Economic Competition.}
Bayesian Persuasion was introduced by \cite{kamenica2011bayesian} --- in the canonical model, there is a single informed ``sender'' and an uninformed ``receiver'' who share a common prior. The sender commits to a ``signaling scheme'', which is a mapping from observations to messages sent to the receiver, who conditions on the message and takes their best response action under their posterior. We adopt the basics of this model, but extend it by allowing that both parties be differently informed, and that interaction involve a multi-round conversation rather than a single message. Multi-sender Bayesian Persuasion was introduced by \cite{gentzkow2016competition} and studies the standard Bayesian Persuasion model with multiple senders who simultaneously commit to a signaling scheme. Subsequently a number of papers have studied multi-sender Bayesian Persuasion \citep{gentzkow2017bayesian,li2018bayesian,au2020competitive,wu2023sequential}. There is also work studying multiple \emph{receiver} Bayesian persuasion, but with a single sender. This work highlights a fundamental tension: when restricted to public signals \citep{alonso2016persuading}, a single sender must often degrade information quality to balance conflicting receiver preferences—a dynamic we recover in our `Anonymous' setting. Conversely, when private signals are allowed \citep{arieli2019private}, a single sender can tailor disclosures to exploit each receiver individually. Our work extends this to the multi-sender setting, showing that competition resolves these failures: in a personalized market, competition prevents senders from exploiting the privacy of the channel, forcing them instead to reveal full information to capture the user. Our model has multiple participants on both sides of the market, and is a direct extension of the model of \citet{collina2025emergentalignmentcompetition} from a single user to many users. Our finding that personalized competition leads to user-optimal outcomes parallels the classic result in economic theory by \citet{thisse1988strategic}, who showed that competitive price discrimination forces firms into a Prisoner's Dilemma that drives prices to marginal cost---improving outcomes for consumers.

\paragraph{AI Alignment.}
Our work fits broadly into the study of multi-agent AI systems \citep{guo2024large}. We present a game theoretic approach in which ``alignment'' emerges from the competitive interaction of many mis-aligned agents. Recent work has explored cooperative multi-agent approaches to AI safety, where multiple AI systems work together to improve alignment outcomes. Constitutional AI \citep{bai2022constitutional} uses AI feedback to train more helpful and harmless models, with one AI system providing critiques and revisions of another's outputs. Similarly, approaches using AI systems to evaluate and improve other AI systems \citep{leike2018scalable} rely on cooperative dynamics where the evaluating system is assumed to be sufficiently aligned to provide useful feedback. These approaches typically assume that at least some components of the multi-agent system are well-aligned or that the agents share compatible objectives. Our work differs by studying strategic rather than cooperative multi-agent settings.

Several recent papers with alignment motivations  \citep{collina2025tractable,collina2025collaborative,nayebi2025barriers} have studied \emph{agreement protocols} through which conversational agents can come to agreement about their beliefs through short interactions. These should be viewed as protocols for cooperative agents, as they are assumed to express their true beliefs at each iteration. We adopt the conversational framework of these papers but like \citep{collina2025emergentalignmentcompetition} study \emph{strategic} agents who do not have the same goals. Our work can likewise be viewed as a strategic generalization of the agreement literature \citep{aaronson2005complexity,frongillo2021agreementimpliesaccuracysubstitutable,collina2025collaborative,collina2025tractable,nayebi2025barriers,kearns2026networked}.

\paragraph{Pluralistic Alignment}
A recent literature on pluralistic alignment (see e.g. \citep{sorensen2024position,shirali2025direct,huang2024collective,klassen2024pluralistic}) focuses on the technical alignment problem of a single model provider when they have many downstream users with different preferences. These approaches pre-suppose that the model designer wants to do this. We instead study conditions under which market competition between many providers, none of which may have incentives to align their models with their users (pluralistically or otherwise) will nevertheless lead to pluralistic alignment emerging from equilibrium.

%% file: files_arxiv/prelims.tex
We study a strategic market of AI providers who interact with many downstream users. Each user is trying to gain information in order to make an optimal decision. Meanwhile, AI providers aim to steer users towards actions that they prefer. 

\paragraph{Players and information.}
There is a finite set of \emph{users} \(N=\{1,\dots,n\}\) and a finite set of \emph{AI providers} \(K=\{1,\dots,k\}\).
A common state \(y\in\mathcal Y\) is drawn along with private features
\((x^U_1,\ldots,x^U_n,x_{1}^P,\ldots,x_{k}^P)\) from a common prior\footnote{For simplicity of notation we assume that all users and providers have access to this joint distribution. However, for all results in this paper, it is sufficient for user $i$ to know only the joint distribution of $(x^{u}_{i}, x^{p}_{j},y)$ marginally for each provider $j$.} \(\phi \in \Delta(\cY \times \prod_{i\in N} \cX_{U_i} \times \prod_{j\in K} \cX_{P_j}) \). Each user $i$ observes $x_i^U$ and each provider $j$ observes $x_j^P$. 
\\
\\
The private information of each provider captures the vast amount of data used to train AI models and their ability to parse large numeric datasets that may be illegible to human users; meanwhile, the private information of each user captures the individual expertise and observations of users interacting with AI systems. 

\paragraph{Actions and utilities.} Each user \(i\) has a utility function
\(u^{U}_i:\mathcal A_i\times\mathcal Y\to[0,1]\) that depends on her chosen action $a_i$ and the state.
Each provider \(j\) has a utility function $u^{P}_j: \prod_{i\in N} \cA_i \times\cY \to [0,1]$ that depends on all users’ actions $a_{1:n} =(a_1,...,a_n)$ and the state. 

\paragraph{Conversations and decision rules.}
We model a setting where an AI model and a user interact over multiple rounds of conversation, exemplified by the usage of AI chatbots. In particular, providers deploy models in the form of \textit{conversation rules} that can be multi-round and adaptive, and users interact via their own conversation rules. The providers' conversation rules  naturally represent model weights (together with any accompanying guardrails), which dictate the next sequence of tokens given a context and conversation history. Meanwhile, users' conversation rules represent a strategy for querying an AI, similarly mapping context and a conversation history to a next message.

\begin{definition}[Player Strategies]

Let $\M$ be a message space. Let $\M^{<R}$ denote the space of conversation transcripts of length less than $R$ and $\M^{R}$ denote the space of conversation transcripts of length $R$.

\begin{itemize}
    \item A \emph{conversation rule}
for provider \(j\) is a mapping
\[
C_{P_j}:\ \mathcal X_{P_{j}}\times \M^{<R} \to\Delta(\mathcal M),
\]
from his private features and the transcript of his conversation with some user to a distribution over next messages.
    \item A \emph{conversation rule} for user \(i\) is a mapping
\[
C_{U_i}:\ \mathcal X_{U_i}\times \mathcal \M^{<R}\to\Delta(\mathcal M),
\]
from her private
features and the transcript of her conversation
with some provider to a distribution over next messages.
    \item After the conversation ends, user \(i\) applies a \emph{decision rule}
\[
D_{U_i}:\ \mathcal X_{U_i}\times \M^{R} \to \Delta(\mathcal A_i),
\]
mapping her private features and the full transcript of her conversation with some provider to a distribution over actions.
\end{itemize} 
\end{definition}

\paragraph{The game.} We model the interaction as a multi-leader, multi-follower Stackelberg game: providers commit first to conversation rules, then users select which provider to consult, how to converse, and how to act. Commitment models the fact that training a new LLM requires substantial time and investment and happens on a longer time-frame than user decisions. Thus, after these weights are deployed, a user can explore different AI models and learn how helpful each one is for the types of questions she cares about. After information-gathering and experimentation, each user will choose the model (provider conversation rule) to interact with which is best for her over her distribution of interests. This setup mirrors the ``Best-AI Selection Game" defined in \citet{collina2025emergentalignmentcompetition}).

A key modeling choice is the strategy space available to providers --- in particular, whether they can tailor rules to each user or must deploy a single rule used anonymously by all users. We consider two versions of the game that capture this distinction. In the \textit{personalized} game, each provider can commit to a separate conversation rule for each user, and each user observes only the conversation rules used in her own interaction. In the \textit{anonymous} game, each provider commits to a single conversation rule that is publicly observable and applied across all users. 

The personalized game captures account-based AI models, which can interact with different users differently based upon their profiles. Meanwhile, the anonymous game captures open-access AI models, which have no information about a user beyond their prompt.

\begin{definition}[The Personalized Game]
    The personalized game proceeds as follows:
    \begin{enumerate}
        \item Every provider $j$ simultaneously commits to a separate conversation rule $C_{P_j}(i)$ for each user $i$. Let $C_{\vec P_j} = (C_{P_j}(1),...,C_{P_j}(n))$ be the vector of rules chosen by provider $j$ and let $C_{\vec P}=(C_{\vec P_1},\ldots,C_{\vec P_k})$ be the entire collection of chosen rules. Let $C_{\vec P}(i)=(C_{P_1}(i),\ldots,C_{P_k}(i))$ be the collection of rules chosen for user $i$.  
        \item Every user $i$ observes $C_{\vec P}(i)$. She chooses a provider $j$ to interact with, as well as a conversation rule $C_{U_i}$ and a decision rule $D_{U_i}$. 
        \item The state $y$ and players' private information $(x^U_1,\ldots,x^U_n,x_{1}^P,\ldots,x_{k}^P)$ are drawn from $\phi$. Each user $i$ observes $x^U_i$, and each provider $j$ observes $x^P_j$.
        \item Every user $i$ engages in the conversation defined by $C_{U_i}$ and $C_{P_j}(i)$ of her chosen provider and selects an action $a_i$ according to $D_{U_i}$. 
        \item Each user $i$ receives her utility $u^{U}_i(a_i,y)$. Each provider $j$ receives his utility $u^{P}_j(a_{1:n}, y)$.
    \end{enumerate}
\end{definition}

\begin{definition}[The Anonymous Game]
    The anonymous game proceeds as follows:
    \begin{enumerate}
        \item Every provider $j$ simultaneously commits to a conversation rule $C_{P_j}$. Let $C_{\vec P}=(C_{P_1},\ldots,C_{P_k})$ be the collection of chosen rules.   
        \item Every user $i$ observes $C_{\vec P}$. She chooses a provider $j$ to interact with, as well as a conversation rule $C_{U_i}$ and a decision rule $D_{U_i}$. 
        \item The state $y$ and players' private information $(x^U_1,\ldots,x^U_n,x_{1}^P,\ldots,x_{k}^P)$ are drawn from $\phi$. Each user $i$ observes $x^U_i$, and each provider $j$ observes $x^P_j$.
        \item Every user $i$ engages in the conversation defined by $C_{U_i}$ and $C_{P_j}$ of her chosen provider and selects an action $a_i$ according to $D_{U_i}$. 
        \item Each user $i$ receives her utility $u^{U}_i(a_i,y)$. Each provider $j$ receives his utility $u^{P}_j(a_{1:n}, y)$.
    \end{enumerate}
\end{definition}

In both games, each user observes only her own conversational transcript when applying her decision rule.

The user's interaction with a provider, under chosen conversation and decision rules, induces a joint distribution over the user's action and the state of the world. 

\begin{definition}[Induced distribution $\mathcal{I}_i(C_{P}, C_{U},D_{U};j)$]
    We write $\mathcal{I}_i(C_{P}, C_{U},D_{U};j) \in \Delta (\mathcal{A}_i \times \mathcal{Y})$ for the distribution over user $i$'s actions and states induced by her interaction with provider $j$, given the user's conversation rule $C_U$, the user's decision rule $D_U$, and the provider's conversation rule $C_P$.
\end{definition}

In both of the above games, the user has three choices to make: (i) which provider she should converse with, (ii) how to converse with him, and (iii) how to map the learned information to a final decision. We assume that users solve these optimally, as formalized below.

\begin{definition}[User's Best Response]
User $i$'s best response consists of the following:

\begin{itemize}
    \item
    User $i$'s best response decision rule $D_{U_i}^*$ is the decision rule that maximizes her utility under the posterior distribution $\pi(y|x_i^U, M^R)$ given her observed features $x_i^U\in \cX_{U_i}$ and a conversation transcript $M^R \in \cM^R$ with provider $j$:
    \[
    D_{U_i}^* = \argmax_{a_i\in\cA_i} \E_{\pi}[u^{U}_i(a_i,y)]
    \]
    \item 
    User $i$'s best response conversation rule $C^*_{U_i}$ given a conversation rule $C_{P_j}$ is the conversation rule that maximizes her expected utility given that she uses her best response decision rule:
    \[
    C^*_{U_i} = \argmax_{C_{U_i}} \E_{\cI_i(C_{P_j}, C_{U_i}, D^*_{U_i};j)}[u_i(a_i,y)]
    \]

    \item 
    User $i$'s optimal choice of a provider to interact with is the provider $j^*$ that maximizes her expected utility given that she uses her best response conversation and decision rules:
    \[
    j^* = \argmax_{j\in K} \E_{\cI_i(C_{P_j}, C^*_{U_i}, D^*_{U_i};j)}[u_i(a_i,y)]
    \]
    \end{itemize}
    We assume a fixed and common tie-breaking rule across all users. 
\end{definition}

Since we assume that every user best responds, the induced distributions depend on the providers' conversation rules alone. For every user $i$, we write $\cI_i(C_{\vec P}) = \cI_i(C_{P_{j^*}}, C^*_{U_i}, D^*_{U_i};j)$ for the induced distribution over $(a_i,y)$. We write $\cI(C_{\vec P})$ for the induced joint distribution over $(a_{1:n}, y)$. 


Now, given that users best respond, this defines a game among providers. 
We will be interested in Nash equilibria of the game in both personalized and anonymous settings:

\begin{definition}[Nash Equilibrium in the Personalized Game]
    A vector of providers' conversation rules $C_{\vec P}$ is a Nash equilibrium in the personalized game if for every provider $j$ and every alternative vector of personalized rules $C_{\vec P_j}'$:
    \[
    \E_{\cI(C_{\vec P})}[u^{P}_j(a_{1:n}, y)] \geq \E_{\cI(C'_{\vec P_j}, C_{\vec P_{-j}})}[u^{P}_j(a_{1:n}, y)]
    \]
\end{definition}

\begin{definition}[Nash Equilibrium in the Anonymous Game]
    A vector of providers' conversation rules $C_{\vec P}$ is a Nash equilibrium in the anonymous game if for every provider $j$ and every alternative vector of rules $C_{P_j}'$:
    \[
    \E_{\cI(C_{\vec P})}[u^{P}_j(a_{1:n}, y)] \geq \E_{\cI(C'_{P_j}, C_{ P_{-j}})}[u^{P}_j(a_{1:n}, y)]
    \]
\end{definition}

\paragraph{Benchmark. }
Our goal is to lower-bound the individual utility of each user in \emph{every} possible Nash equilibrium. As each provider has potentially different information, there is a large and asymmetric set of defections each provider could make from any potential equilibrium. We will compare equilibrium outcomes to a benchmark that captures what \textit{any} provider can implement given his private information.  Accordingly, we are interested in the common information among all providers, a notion that we call a \emph{common garbling}.

Before formally defining this notion, we provide an illustrative setting where our notion of common information is straightforward. Suppose $y\in \R^d$ is a vector where each element is distributed independently, and each $x^P_{j}$ is a subset of these coordinates. Here, the information that every provider has about $y$ is captured precisely by the intersection of their sets. In other words, the ``common information" among all providers is $z = \bigcap_{j \in K} x^{P}_{j}$. Importantly, since every provider has access to $z$, all providers can implement any conversation rule $C_{P_j}$ depending only on $z$.

In general, $y$ and $x^{P}_{1:k}$ may be arbitrarily correlated random variables, and thus the ``common information" of $x^{P}_{1:k}$ may not be representable by a set intersection. Common garblings allow us to generalize this notion to random variables which can be simulated by any provider while preserving their correlation with $y$. 
In the special case where all providers observe the same features, the shared feature itself is a common garbling. This is the setting of \citet{collina2025emergentalignmentcompetition}, and so our model is a strict generalization.

\begin{definition}[Common garbling]
    A random variable $z \in \mathcal{Z}$ correlated with $y \in \mathcal{Y}$ is a common garbling of $(x^P_1, \ldots, x^P_k)$ with respect to $y$ if for every provider $j \in K$, there exists a (possibly randomized) mapping $f_{j}: \mathcal{X}_{P_j} \rightarrow \Delta\mathcal{Z}$ such that, for all $\hat{y}$ and $\hat{z}$:
    \[ 
    \Pr_{y \sim \phi}\left(z = \hat{z} | y = \hat{y}\right)  = \mathbb{E}_{x^{P}_j, y \sim \phi}[\Pr_{f_{j}}\left(f_{j}(x_{j}^{P}) = \hat{z}\right) | y = \hat{y}] 
    \]
In other words, conditional on a fixed value of $y$, the marginals of $z$ and $f_{j}(x^{P}_j)$ are equivalent.
    
\end{definition}

With this definition in hand, we can define the set of conversation rules that can be implemented by any provider in some set $T$.

\begin{definition}[$T$-Shared Conversation Rules $\mathcal{C}_{S}(x^{P}_{T})$]
    Given a set of random variables $x^P_{1:k}, y$, a \emph{$T$-Shared Conversation Rule} is a conversation rule of the form 
    \[C_S : \mathcal{Z} \times \M^{<R} \to \Delta(\mathcal M)\] 
    where $z\in \mathcal{Z}$ is a realization of a random variable that is a common garbling of $x^{P}_{T}$ with respect to $y$, where $T \subseteq K$. 
    We denote the set of all such $R$-round $T$-shared conversation rules $\mathcal{C}_{S,R}(x^{P}_{T})$. 
\end{definition}

We will frequently refer to the utility a user $i$ receives when she interacts with a particular provider who deploys some conversation rule. 

\begin{definition}[User $i$'s utility when interacting with provider $j$ who deploys $C_{P_j}$]
We write $u^{U}_i(C_{P_j}; j) = \E_{\cI_i(C_{P_j}, C^*_{U_i}, D^*_{U_i};j)}[u^{U}_i(a_i,y)]$ to denote user $i$'s expected utility when she interacts with the conversation rule $C_{P_j}$ deployed by provider $j$.     
\end{definition}
A special case is when a provider deploys a shared conversation rule $C_G$ taking as input a common garbling $z$. Since \textit{any} provider can generate $z$ and implement $C_G$, the distribution induced by the user's interaction with $C_G$ --- and hence the user's expected utility --- is the same no matter which provider deployed it. Thus, we will drop the provider's index when speaking about the user's utility when interacting with a shared conversation rule. 
\begin{definition}[User $i$'s utility when interacting with $C_S$]
We write $u^{U}_i(C_S) = \E_{\cI_i(C_S, C^*_{U_i}, D^*_{U_i})}[u^{U}_i(a_i,y)]$ to denote user $i$'s expected utility over the induced distribution of interacting with conversation rule $C_S$ which depends on random variable $z$.     
\end{definition}

This leads us to the following benchmark which we will use throughout the paper --- the outcome the user can obtain if she were interacting with a single perfectly aligned provider who can implement any shared conversation rule. As we will frequently consider subsets of providers who satisfy market alignment conditions, it will be useful to define the optimal shared conversation rule for any subset of providers $T$ in terms of the information shared only among all providers in $T$. 

\begin{definition}[Optimal $T$-Shared Conversation Rule $C_{S,T,R}^{*}(i)$]
For some user $i$, her optimal $R$-round $T$-shared conversation rule $C_{S,T,R}^{*}(i)$ is the $T$-shared conversation rule which maximizes her utility over $R$ rounds of conversation, assuming she best responds. Formally: 
\[
C_{S,T,R}^{*}(i) := \argmax_{C_S \in \mathcal{C}_{S,R}(x^{P}_{T})} u^{U}_{i}(C_S)
\]
We will drop the subscript $R$ when the number of rounds is clear from context. 
\end{definition}

Note that if $C$ is a $K$-shared conversation rule, it is also a $T$-shared conversation rule for any $T \subseteq K$. Thus, $\argmax_{C_S \in \mathcal{C}_{S,R}(x^{P}_{T})} u^{U}_{i}(C_S) \geq \argmax_{C_S \in \mathcal{C}_{S,R}(x^{P}_{1:k})} u^{U}_{i}(C_S)$.


\paragraph{Market Alignment Conditions.}
We now give alignment conditions formalizing how AI providers' preferences relate to users'. These are \textit{market-level} conditions: they do not require any single AI provide to be aligned with any user, but rather offer enough structure such that competition can lead to user-aligned outcomes. In our results, it will only be necessary for a \emph{subset} of the providers to satisfy our alignment conditions; thus, we write the below definitions for a set of providers $T \subseteq K$.

\begin{definition}[$(T, \eps_U,\eps_P)$-Weak Market Alignment]\label{def:separability}
We say a set $T\subseteq K$ of providers satisfy $(T,\eps_U,\eps_P)$-Weak Market Alignment if:
\begin{enumerate}
    \item (Provider separability) Every provider \(j\in T\) has an additively separable utility. That is, for every provider $j\in T$, there exist non-negative weights $(\lambda_{j,i})_{i\in N}$, a constant $c_j\in\R$, and an error $\eps_P \geq 0$ such that for all action profiles $a_{1:n}$ and states $y$:
\begin{equation}\label{eq:separableU}
\left| u^{P}_j\big(a_{1:n},y\big)\ - \left(\sum_{i\in N}\lambda_{j,i}\,F_{j,i}(a_i,y) + c_j\right)\right| \leq \eps_P
\end{equation}
where \(F_{j,i}:\mathcal A_i\times\mathcal Y\to\mathbb R^{\geq 0}\) depends only on user \(i\)'s action and the state.
\item (User alignment) For every user $i\in N$, there exist non-negative weights $(w_{j,i})_{j\in T}$, a constant $c_i\in \R$, and an error $\eps_U\geq 0$ such that for all actions $a_i$ and states $y$:
\begin{equation}\label{eq:peruserAlignment}
\left| u^{U}_i(a_i,y)\ - \left(\sum_{j\in T} w_{j,i}\,F_{j,i}(a_i,y) + c_i\right) \right| \leq \eps_U
\end{equation}
\end{enumerate}
We assume that for every user $i$, there is a provider $j^*$ where \(\lambda_{j^*,i}>0\) and $w_{j^*, i} > 0$. Additionally, if $w_{j,i} > 0$, then $\lambda_{j,i}>0$.
\end{definition}

Our Weak Market Alignment condition can, for example, naturally model settings where AI providers care about the number of users who take a particular downstream action (e.g. purchase a product, prescribe a drug, or vote for a political candidate) with different importance weights placed on different users or user groups (e.g. demographic groups of consumers or voters). For every provider $j$, $F_{j,i}(a_i,y)\in[0,1]$ can measure how desirable user $i$'s action $a_i$ is in a state $y$ for that provider (e.g. $F_{j,i}=\1[a_i = a^*(j;y)]$ where $a^*(y,j)$ is provider $j$'s preferred downstream action in state $y$). The weights $\lambda_{j,i}$ then capture how much provider $j$ values the resulting action taken by a particular user (or group), while $w_{j,i}$ capture how much relative utility the user receives from provider $j$'s preferred outcome. 

We also define a Strong Market Alignment condition:

\begin{definition}[$(T, \epsilon)$-Strong Market Alignment]\label{def:provider-alignment}
    We say a set $T\subseteq K$ of providers satisfy $(T, \epsilon)$-Strong Market Alignment if every provider $j \in T$ has a utility that can be approximately written as a non-negative linear combination of user utilities. That is, for every provider $j\in T$, there exist non-negative weights $(\lambda_{j,i})_{i\in N}$, a constant $c_j\in \R$, and an error $\eps \geq 0$ such that for all action profiles $a_{1:n}$ and states $y$:
    \begin{equation}
\left| u^{P}_j\big(a_{1:n},y\big)\ -  \left(\sum_{i\in N}\lambda_{j,i}\,u_{i}(a_i,y) + c_j\right) \right| \ \leq \epsilon
\end{equation}
    We assume coverage: for every user $i$, there exists a provider $j^*\in T$ with $\lambda_{j^*,i}> 0$.
\end{definition}
This condition can be viewed as the dual of the ``Weighted Alignment Assumption" in \citet{collina2025emergentalignmentcompetition}, which requires that, in a single-user setting, the user's utility can be written as a non-negative linear combination of provider utilities. We show that Strong Market Alignment is an instantiation of Weak Market Alignment, where $F_{j,i} = u^{U}_i(a_i,y)$ (Proposition~\ref{prop:implies}). In fact, Weak Market Alignment is strictly weaker (Proposition~\ref{prop:strict}).

%% file: files_arxiv/personalized.tex
We begin by studying the personalized game, where each provider can tailor his conversation rule to the identity of the user he is interacting with. We show that under Weak Market Alignment, every user achieves near-benchmark in \textit{every} Nash equilibrium of the personalized game. In particular, she gets utility close to what she would get by interacting with a perfectly aligned provider restricted to using common information. The intuition is simple. 
With personalization, providers can decouple their choices across users: a provider can strategically influence the action of a single user by changing what she observes, without committing to the same behavior for every other user. This property, together with the alignment condition, implies every user must be getting high utility: if any user were doing poorly, some provider could profit by offering her a better personalized rule, contradicting equilibrium.
\ifmalgai 
\else 

\fi
Importantly, our result does not require \textit{every} provider to satisfy the alignment condition --- an assumption that can be unrealistic in markets of many provider and many users. Instead, we only assume that there exists \textit{some} subset of providers that does. This substantially weaker guarantee suffices to guarantee that competition among all providers yields aligned outcomes for users. In particular, our results will continue to hold if additional providers who do not satisfy the alignment condition enter the market.

\begin{theorem}\label{thm:private}
    Suppose there is a set of providers $T\subseteq K$ that satisfy $(T,\eps_U,\eps_P)$-Weak Market Alignment with weights $(\lambda_{j,i})_{i\in N}$ for each provider $j\in T$ and $(w_{j,i})_{j\in T}$ for each user $i\in N$. In the personalized game among all providers $K$, for each user $i$, her expected utility in any Nash equilibrium is at least $u^{U}_i(C_{S,T}^*(i)) - 2\eps_U - 2\eps_P\mu_i$, where $\mu_i = \sum_{j\in T:w_{j,i}>0} \frac{w_{j,i}}{\lambda_{j,i}}$.
\end{theorem}
\begin{proof}
    Let $C_{\Vec{P}}$ be any Nash equilibrium strategy profile of the personalized game, and let $\cI_{NE} = \cI(C_{\Vec{P}})$ be the induced distribution.
    
    Now, assume for contradiction that there exists a user $i \in N$ such that her utility is strictly lower than $u^{U}_i(C_{S,T}^*(i)) - 2\eps_U - 2\eps_P\mu_i$. Consider a deviation by any provider $j \in T$ in his choice of conversation rule for user $i$ from $C_{P_j}(i)$ to $C_{S,T}^*(i)$, the optimal $T$-shared conversation rule for user $i$. Note that by definition, any provider in $T$ could implement $C^*_{S,T}(i)$.
    
    Let $\cI_{dev} = \cI(C_{S,T}^*(i), C_{P_j}(-i), C_{P_{-j}})$ be the distribution induced by the deviation of provider $j$ in his chosen conversation rule for user $i$ (the conversation rules chosen by provider $j$ for all other users, $C_{P_j}(-i)$, remain unchanged). As user $i$'s utility was previously worse than $u^{U}_i(C_{S,T}^*(i))$, she will now choose to interact with the provider who implements $C^*_{S,T}(i)$, leading to the outcome under $\cI_{dev}$.  

    Now, by the Nash equilibrium condition, no provider $j$ has an incentive to deviate. Thus, for all providers $j\in T$:
    $$\mathbb{E}_{\mathcal{I}_{dev}}[u^{P}_j(a_{1:n},y)] \le \mathbb{E}_{\mathcal{I}_{NE}}[u^{P}_j(a_{1:n},y)]. $$
    Since every user observes only the conversation rules chosen for her, this deviation can only change the induced distribution for user $i$ (in particular the distribution over her actions). 
    Thus, applying Weak Market Alignment, we have that for any provider $j \in T$:
    \begin{align*}
        &\mathbb{E}_{\mathcal{I}_{dev}}\left[\sum_{i\in N} \lambda_{j,i} F_{j,i}(a_i, y) + c_j - \eps_P \right] \le \mathbb{E}_{\mathcal{I}_{NE}}\left[\sum_{i\in N} \lambda_{j,i} F_{j,i}(a_i, y) + c_j + \eps_P \right] \\
        \iff & \mathbb{E}_{\mathcal{I}_{dev}}\left[\lambda_{j,i} F_{j,i}(a_i, y) + \sum_{i'\neq i} \lambda_{j, i'} F_{j, i'}(a_{i'}, y) + c_j - \eps_P\right] \\ &\le \mathbb{E}_{\mathcal{I}_{NE}}\left[\lambda_{j,i} F_{j,i}(a_i, y) + \sum_{i'\neq i} \lambda_{j, i'} F_{j, i'}(a_{i'}, y) + c_j +\eps_P \right] \\
        \iff & \mathbb{E}_{\mathcal{I}_{dev}}\left[\lambda_{j,i} F_{j,i}(a_i, y)\right] + \mathbb{E}_{\mathcal{I}_{NE}}\left[\sum_{i'\neq i} \lambda_{j, i'} F_{j, i'}(a_{i'}, y)\right] +c_j-\eps_P \\ &\le \mathbb{E}_{\mathcal{I}_{NE}}\left[\lambda_{j,i} F_{j,i}(a_i, y)\right] + \mathbb{E}_{\mathcal{I}_{NE}}\left[\sum_{i'\neq i} \lambda_{j, i'} F_{j, i'}(a_{i'}, y)\right] +c_j +\eps_P \\
        \iff & \lambda_{j,i} \mathbb{E}_{\mathcal{I}_{dev}}\left[ F_{j,i}(a_i, y)\right] \leq \lambda_{j,i} \mathbb{E}_{\mathcal{I}_{NE}}\left[ F_{j,i}(a_i, y)\right] + 2\eps_P 
    \end{align*}
    where the third step follows from the fact that only the induced distribution over user $i$'s actions changes under the deviation. For any $j$ with $\lambda_{j,i}>0$, we thus have:
    \[
    \mathbb{E}_{\mathcal{I}_{dev}}\left[ F_{j,i}(a_i, y)\right] \leq  \mathbb{E}_{\mathcal{I}_{NE}}\left[ F_{j,i}(a_i, y)\right]  + 2\frac{\eps_P}{\lambda_{j,i}}
    \]

    Taking a weighted sum over all providers using the non-negative weights $(w_{j,i})_{j\in T}$ from the Weak Market Alignment condition:
    $$\sum_{j\in T} w_{j,i} \mathbb{E}_{\substack{\mathcal{I}_{dev}}}[F_{j,i}(a_i, y)] \le \sum_{j\in T} w_{j,i} \mathbb{E}_{\mathcal{I}_{NE}}[F_{j,i}(a_i, y)] + 2\eps_P\sum_{j\in T:w_{j,i}>0} \frac{w_{j,i}}{\lambda_{j,i}}. $$
    Note that by assumption, if $w_{j,i}>0$, then $\lambda_{j,i}>0$, and so we can apply the inequality here. 
    By linearity of expectation, this is equivalent to:
    $$ \mathbb{E}_{\substack{\mathcal{I}_{dev}}}\left[\sum_{j\in T} w_{j,i} F_{j,i}(a_i, y)\right] \le \mathbb{E}_{ \mathcal{I}_{NE}}\left[\sum_{j\in T} w_{j,i} F_{j,i}(a_i, y)\right] + 2\eps_P\sum_{j\in T:w_{j,i}>0} \frac{w_{j,i}}{\lambda_{j,i}}. $$
    Now we use the market alignment assumption, which states that $\sum_j w_{j,i} F_{j,i}(a_i, y)$ can be approximated by $u^{U}_i(a_i, y) + c_i$. We have that:
    $$ \mathbb{E}_{\substack{\mathcal{I}_{dev}}}\left[\sum_{j\in T} w_{j,i} F_{j,i}(a_i, y)\right] \ge \mathbb{E}_{\mathcal{I}_{dev}}[u^{U}_i(a_i,y) + c_i - \eps_U] = \mathbb{E}_{\mathcal{I}_{dev}}[u^{U}_i(a_i,y)] + c_i -\eps_U \ge u^{U}_i(C_{S,T}^*(i)) + c_i - \eps_U . $$
    Additionally:
    $$ \mathbb{E}_{\mathcal{I}_{NE}}\left[\sum_{j\in T} w_{j,i} F_{j,i}(a_i, y)\right] \le \mathbb{E}_{\mathcal{I}_{NE}}[u^{U}_i(a_i,y) + c_i +\eps_U]  = \mathbb{E}_{\mathcal{I}_{NE}}[u^{U}_i(a_i,y)] + c_i +\eps_U. $$
    Combining these inequalities, we get:
    $$ u^{U}_i(C_{S,T}^*(i)) + c_i -\eps_U \le \mathbb{E}_{\mathcal{I}_{NE}}[u^{U}_i(a_i,y)] + c_i +\eps_U + 2\eps_B\sum_{j\in T:w_{j,i}>0} \frac{w_{j,i}}{\lambda_{j,i}}. $$
    The constant offset $c_i$ cancels, and we are left with:
    $$ u^{U}_i(C_{S,T}^*(i)) -\eps_U \le \mathbb{E}_{\mathcal{I}_{NE}}[u^{U}_i(a_i,y)] +\eps_U + 2\eps_P\sum_{j\in T:w_{j,i}>0} \frac{w_{j,i}}{\lambda_{j,i}}. $$
    Thus:
    $$ \mathbb{E}_{\mathcal{I}_{NE}}[u^{U}_i(a_i,y)] \ge u^{U}_i(C_{S,T}^*(i)) - 2\eps_U - 2\eps_P\sum_{j\in T:w_{j,i}>0} \frac{w_{j,i}}{\lambda_{j,i}}. $$
which completes the proof.
\end{proof}

%% file: files_arxiv/anonymous.tex
In the previous section we established that every user obtains her benchmark utility in the personalized setting when Weak Market Alignment holds. Here, we turn attention to the anonymous version of the game and ask: what can we guarantee users when personalization is not available? 
\ifmalgai
\else 

\fi
Our first result shows a separation: without personalization, users can suffer arbitrarily poor outcomes in Nash equilibrium---even when the Weak Market Alignment condition holds. The intuition is that since providers must commit to a single rule across users, strategic choices are no longer decomposable among users. A disclosure to one user is shown to everyone, and so providers must trade off how different users will respond. We show that this tension can push the market toward policies that reveal no useful information to any user. In particular, we construct a class of anonymous game instances where providers have conflicting preferences over which user should be informed. The users are symmetric (they share the same utility and private information), so they will best respond in the same way. If each provider desire to hide information from one user is stronger than his desire to share information with the other, then each provider will prefer to withhold his information.  

\begin{theorem}\label{thm:public-example}
    For any $\eps>0$ and $c \in (\eps, 1]$, there exists an anonymous game exactly satisfying the Weak Market Alignment condition that has a Nash equilibrium where every user $i$'s expected utility is at most $\eps$, while the optimal shared conversation rule benchmark $u^{U}_i(C_{S,K}^*(i))$ is at least $c$. 
\end{theorem}
\begin{proof}
    Given $\eps$ and $c$, we construct an instance as follows. Consider a game with two providers and two users. Choose $M\in\mathbb{N}$  with $M>\frac{c}{\eps}$. The state $y$ is drawn uniformly from $\{1,...,M\}$. Each provider observes the state, i.e. $x_1^P = x_2^P = y$, while users observe nothing. Users have the action space $\cA_1=\cA_2=\{1,...,M,\bot\}$. Suppose $R=1$ (one round of conversation), and provider conversation rules map the observed state to a (randomized) signal in the state space. 
    
    In this game, the providers' utilities are:
    \begin{align*}
        u^{P}_1(a_1,a_2,y) &= \frac{1}{D+1} \1[a_1=y] + \frac{D}{D+1}\1[a_2=\bot] \\
        u^{P}_2(a_1,a_2,y) &=  \frac{D}{D+1}\1[a_1=\bot] + \frac{1}{D+1}\1[a_2=y]
    \end{align*}
    where $D>1$ is any fixed constant. 
    The users' utilities are:
    \begin{align*}
        u^{U}_1(a_1,y) &= c \cdot \1[a_1=y] + \eps \cdot \1[a_1=\bot] \\
        u^{U}_2(a_2,y) &= c \cdot \1[a_2=y] + \eps \cdot \1[a_2=\bot]
    \end{align*}
    Note that the users' and providers' utilities exactly satisfy the Weak Market Alignment assumption. Furthermore, all utilities are between $0$ and $1$.

    We show that if both providers choose conversation rules that reveal no information (e.g. the rule is: always signal 1, regardless of the observed state), this forms a Nash equilibrium in the public commitment game. Notice that if both providers choose to reveal no information, each user receives utility $\eps$ in expectation by choosing the action $\bot$ and utility $\frac{c}{M} <\eps$ in expectation by choosing any other action. Thus
    each user's utility-maximizing action is $\bot$, and each provider receives utility $\frac{D}{D+1}$. 
    
    Consider any deviation by provider 1 to another conversation rule, resulting in a vector of provider conversation rules $C_{\vec{P}}$. Since both users have the same utility and same private information, $\mathcal{I}_{1}(C_{\vec{P}}) = \mathcal{I}_{2}(C_{\vec{P}})$. Let $p=\Pr_{\mathcal{I}_{1}(C_{\vec{P}})}[a_1\neq \bot] = \Pr_{\mathcal{I}_{2}(C_{\vec{P}})}[a_2\neq \bot]$. Then, provider 1's expected utility is: $\frac{1}{D+1}\Pr_{{\mathcal{I}_{1}(C_{\vec{P}})}}[a_1=y] + \frac{D}{D+1} \Pr_{{\mathcal{I}_{2}(C_{\vec{P}})}}[a_2=\bot] \leq \frac{p}{D+1}+\frac{D(1-p)}{D+1} = \frac{p-Dp}{D+1}+\frac{D}{D+1}\leq \frac{D}{D+1}$ (since $D>1$). 
    Thus, no deviation can give any provider utility strictly higher than $\frac{D}{D+1}$, and so the no disclosure signal profile is a Nash equilibrium. 
    
    Note as both providers have the same private information, in this case the optimal shared conversation rule is the rule which reveals $y$. This induces the action $a_i=y$ and gives both users expected utility $c$. This completes our proof. 
\end{proof}

Motivated by this separation, we then ask when strong user guarantees can be recovered in the anonymous setting. We give a simple setting where Strong Market Alignment suffices to guarantee that every user obtains her benchmark utility --- namely, the setting where a \textit{user-dominant} conversation rule (one that is simultaneously optimal for \textit{every} user) is available.

\begin{definition}[User-Dominant Conversation Rule]
    Given a provider $j$, we say that a conversation rule $C^*_{P_j}$ is user-dominant if for all users $i$ and all alternative conversation rules $C_{P_j}$ provider $j$ could deploy, $u^{U}_i(C^*_{P_j};j) \geq u^{U}_i(C_{P_j};j)$.
\end{definition}

Our next result shows that whenever a user-dominant rule exists\footnote{Observe that a user-dominant conversation rule exists whenever the conversation is permissive enough for a provider to reveal his private information --- this ``full-revelation" rule is user-dominant.}, every user obtains utility close to her utility if she were to interact with this user-dominant rule, under the Strong Market Alignment condition. Intuitively, since provider utilities are non-negative linear combinations of user utilities, and user utilities are maximized by the user-dominant rule, providers will themselves prefer to deploy the user-dominant rule. As before, our results here do not require every provider in the marketplace to satisfy the alignment condition; it will suffice that there exist providers that do.

\begin{theorem}\label{thm:public-natalie}
    Suppose there exists a set of providers $T\subseteq K$ such that $(T,\eps)$-Strong Market Alignment is satisfied with weights $(\lambda_{j,i})_{i\in N}$ for every provider $j\in T$. Furthermore, suppose for every provider $j\in T$, there exists a user-dominant conversation rule $C^*_{P_j}$ for all users $i\in N$. Then, for each user $i\in N$, her expected utility in any Nash equilibrium is at least $u^{U}_i(C^*_{S,T}(i)) - \frac{2\eps}{\lambda_i^*}$, where $\lambda_i^* = \max_{j\in T} \lambda_{j,i}$. 
\end{theorem}
\begin{proof}
    Fix any Nash equilibrium strategy profile $C_{\vec P}$ and let $\cI_{NE}=\cI(C_{\vec P})$. Suppose there is a user $i$ such that her utility is strictly less than $u^{U}_i(C_{S,T}^*(i)) - \frac{2\eps}{\lambda_i^*}$. We show that this implies the existence of some provider who has a profitable deviation, thus contradicting the premise that $C_{\vec P}$ is a Nash equilibrium. 
    
    Let $\tilde{j} = \argmax_{j\in T} \lambda_{j,i}$. Thus $\lambda_i^* = \lambda_{\tilde{j},i}$. Consider a deviation by provider $\tilde{j}$ to the user-dominant conversation rule $C^*_{P_{\tilde{j}}}$. Let $\cI_{dev}=\cI(C^*_{P_{\tilde{j}}}, C_{P_{-\tilde{j}}})$ be the distribution induced by this deviation.
    By definition of $C^*_{P_{\tilde{j}}}$, $u^{U}_i(C^*_{P_{\tilde{j}}};\tilde{j}) \geq u^{U}_i(C^*_{S,T}(i))$, so user $i$'s utility in equilibrium must be strictly less than $u^{U}_i(C^*_{P_{\tilde{j}}}; \tilde{j}) - \frac{2\eps}{\lambda_{\tilde{j},i}}$. Thus, under this deviation, user $i$ will now choose to interact with provider $\tilde{j}$ and receive utility $u^{U}_i(C^*_{P_{\tilde{j}}}; \tilde{j})$, increasing her utility by $> \frac{2\eps}{\lambda_{\tilde{j},i}}$.

    Now, by the user-dominance of $C^*_{P_{\tilde{j}}}$, we have that for all other users $i'$, $u^{U}_{i'}(C^*_{P_{\tilde{j}}};\tilde{j}) \geq u^{U}_{i'}(C_{P_{\tilde{j}}}; \tilde{j})$. That is, users' utilities can only (weakly) increase under this deviation. Then, applying approximate Strong Market Alignment, we have:
    \begin{align*}
        \E_{\cI_{dev}}[u^{P}_{\tilde j}(a_{1:n}, y)] &\geq \E_{\cI_{dev}}\left[\sum_{i\in N} \lambda_{\tilde{j}, i} u^{U}_i(a_i,y) + c_{\tilde{j}} - \eps \right] \\
        &= \sum_{i\in N} \lambda_{\tilde{j}, i} \E_{\cI_{dev}}\left[  u^{U}_i(a_i,y)  \right] + c_{\tilde{j}} - \eps \\
        &> \lambda_{\tilde{j}, i} \left( \E_{\cI_{NE}}\left[  u^{U}_i(a_i,y)  \right] + \frac{2\eps}{\lambda_{\tilde{j},i}}\right) + \sum_{i' \neq i} \lambda_{\tilde{j}, i'} \E_{\cI_{NE}}\left[  u^{U}_{i'}(a_{i'},y)  \right] + c_{\tilde{j}} - \eps \\
        &= \sum_{i\in N} \lambda_{\tilde{j}, i} \E_{\cI_{NE}}\left[  u^{U}_{i}(a_{i},y)  \right] + c_{\tilde{j}} + \eps \\
        &= \E_{\cI_{NE}}\left[ \sum_{i\in N} \lambda_{\tilde{j}, i} u^{U}_{i}(a_{i},y)  + c_{\tilde{j}} \right] + \eps \\
        &\geq \E_{\cI_{NE}}\left[ u^{P}_{\tilde j}(a_{1:n}, y) -\eps \right] + \eps\\
        &= \E_{\cI_{NE}}\left[ u^{P}_{\tilde j}(a_{1:n}, y) \right]
    \end{align*}
    where the first and final inequalities follow from approximate Strong Market Alignment, and the third inequality follows from the users' increases in utility under the deviation. Therefore, provider $\tilde{j}$ can strictly increase his utility by deviating to $C^*_{P_{\tilde{j}}}$, and so $C_{\vec P}$ cannot be a Nash equilibrium. This completes the proof.

\end{proof}

We now extend positive results in the anonymous game with Strong Market Alignment beyond settings where a user-dominant conversation rule is available. The key observation is: any provider can always use the first round to elicit a user's identity (which is feasible whenever the message space has size at least $N$), and use the remaining $R-1$ rounds to implement that user's optimal shared $(R-1)$-conversation rule. If Strong Market Alignment holds, and sacrificing one round of conversation is not too harmful, then providers have incentive to adopt this rule. In effect, this strategy recovers personalization --- with the caveat that users need not report their identities truthfully. We show that under this rule, users have no incentive to lie. 

To prove this result, then, we need to quantify the marginal value of a single additional round for a user. For any provider $j$, let $u^{U}_{i}(C^{*}_{R}(i);j) = \max_{C \in \mathcal{C}_{R}}(u^{U}_{i}(C; j))$ denote user $i$'s utility under her favorite $R$-round conversation rule offered by provider $j$. Next, define $\Delta_R(j) = \max_{i\in N} \left(u^{U}_i(C_R^*(i);j) - u^{U}_i(C^*_{S,T,R-1}(i))\right)$ to be the largest gain (across all users) from interacting with the best $R$-round conversation rule deployed by provider $j$ versus the best $(R-1)$-round $T$-shared conversation rule. Note that if $R$ is large, then the marginal value of a single round --- and hence $\Delta_R(j)$ --- is often small. 

\begin{theorem}\label{thm:public-natalie-general}
   Suppose there exists a set of providers $T\subseteq K$ such that $(T,\eps)$-Strong Market Alignment is satisfied with weights $(\lambda_{j,i})_{i\in N}$ for every provider $j\in T$. Assume $|\mathcal{M}| \geq N$. Then, for each user $i$, her expected utility in any Nash equilibrium is at least $u^{U}_i(C^*_{S,T,R-1}(i))- \delta_{i}$, where $\delta_{i} = \min_{j:\\\lambda_{j,i}>0}\left(\Delta_R(j)\frac{\sum_{i'\neq i} \lambda_{j,i'}}{\lambda_{j,i}} + \frac{2\eps}{\lambda_{j,i}}\right)$.
\end{theorem}
\begin{proof}
    Consider any Nash equilibrium strategy profile $C_{\vec{P}}$. Suppose that there is a user $i^*$ such that her utility is strictly less than $u^{U}_{i^*}(C^*_{S,T,R-1}(i^*))-\delta_{i^*}$. We show that this implies the existence of some provider who has a profitable deviation, thus contradicting the premise that $C_{\vec{P}}$ is a Nash equilibrium. 

    Recall that under the Strong Market Alignment condition, for every provider $j$'s utility:
    \[
    \left| u^{P}_j\big(a_{1:n},y\big)\ - \left(\sum_{i\in N}\lambda_{j,i}\,u^{U}_{i}(a_i,y) + c_j\right)  \right| \leq \eps
    \]
    where $\lambda_{j,i}\ge 0$ and $c_j\in\R$ for all $i$. Furthermore, for all users $i$, there exists at least one provider $j$ with $\lambda_{j,i}>0$. Let $j^* = \argmin_{j: \lambda_{j,i^*}>0} \left(\Delta_R(j)\frac{\sum_{i'\neq i^*} \lambda_{j,i'}}{\lambda_{j,i^*}} + \frac{2\eps}{\lambda_{j,i^*}}\right)$. 

    Suppose provider $j^*$ deviates to the following conversation rule $C_{P_{j^*}}'$: on the first round, every user $i$ self-selects into one of $|N|$ ``branches," using the message space of size $\geq |N|$. In the remaining $R-1$ rounds, provider $j^*$ implements $C_{S,T,R-1}^*(i)$ in each branch $i$. Note that this is a valid $R$-round conversation rule. Let $C_{\vec P}^* = (C'_{P_{j^*}}, C_{P_{-j^*}})$ be the profile of provider strategies under this unilateral deviation.
    
    As all users are strategic, any user $i$ who does interact with provider $j^*$ will either declare her own type or declare another type. If she declares her own type, she gets utility $u^{U}_{i}(C^{*}_{S,T,R-1}(i))$. If she doesn't, since $C^{*}_{S,T,R-1}(i)$ is defined as the best $R-1$-round conversation rule for her, it must be that she is also getting the same utility $u^{U}_{i}(C^{*}_{S,T,R-1}(i))$. That is, declaring her own type weakly maximizes her utility.
    
    Observe that if user $i^*$ interacts with provider $j^*$ under $C^{*}_{\vec{P}}$ and self-selects into branch $i^*$, her expected utility will be $u^{U}_{i^*}(C^{*}_{S,T,R-1}(i^*))$, and so it will have increased by at least $\delta_{i^*}$.  

    Furthermore, every user $i \neq i^*$ has her utility decreased by at most $\Delta_{R}(j^*)$. To see this, note that there are three cases for each user $i$:
    \begin{itemize}
        \item In $C_{\vec P}$, user $i$ is interacting with provider $j^*$. Then her utility under $C_{\vec P}$ is at most $u^{U}_i(C^{*}_R(i);j)$, and her utility under $C_{\vec P}^*$ is lower by at most $\Delta_{R}(j^*)$.
        \item In $C_{\vec P}$, user $i$ is interacting with some provider $j \neq j^*$, and in $C^{*}_{\vec P}$ she is also interacting with some provider $j \neq j^*$. Then, as the other providers have not changed their conversation rules, user $i$'s utility does not decrease at all.
        \item In $C_{\vec P}$, user $i$ is interacting with some provider $j \neq j^*$, and in $C^{*}_{\vec P}$ she is interacting with provider $j^*$. She could still have gotten the same expected utility from provider $j$ under $C^{*}_{\vec P}$, so her utility must have weakly \emph{increased} under the deviation to $C^{*}_{\vec P}$.
    \end{itemize}
    Thus, every other user $i$'s expected utility can only decrease by at most $\Delta_R(j^*)$ by assumption. So, we can compute the change in provider $j^*$'s utility:
    \begin{align*}
        & \E_{\cI(C'_{P_{j^*}}, C_{P_{-j^*}})}[u^{P}_{j^*}(a_{1:n},y)] - \E_{\cI(C_{\vec p})}[u^{P}_{j^*}(a_{1:n},y)] \\
        &\geq \E_{\cI(C'_{P_{j^*}}, C_{P_{-j^*}})}\left[\lambda_{j^*,i^*} u^{U}_{i^*}(a_{i^*},y) + \sum_{i \in N \neq i^*} \lambda_{j^*,i} u^{U}_i(a_i,y) +c_j - \eps \right] \\
        &\ \ \ \ \  - \E_{\cI(C_{\vec P})}\left[\lambda_{j^*,i^*} u^{U}_{i^*}(a_{i^*},y) + \sum_{i \in N \neq i^*} \lambda_{j^*,i} u^{U}_i(a_i,y) +c_j + \eps \right] \\
         &= \E_{\cI(C'_{B_{j^*}}, C_{P_{-j^*}})}\left[\lambda_{j^*,i^*} u^{U}_{i^*}(a_{i^*},y) + \sum_{i \in N \neq i^*} \lambda_{j^*,i} u^{U}_i(a_i,y) +c_j  \right] \\
        & \ \ \ \ \ - \E_{\cI(C_{\vec P})}\left[\lambda_{j^*,i^*} u^{U}_{i^*}(a_{i^*},y) + \sum_{i \in N \neq i^*} \lambda_{j^*,i} u^{U}_i(a_i,y) +c_j \right] - 2\eps \\
        &= \lambda_{j^*,i^*} \left(\E_{\cI(C'_{P_{j^*}}, C_{P_{-j^*}})}\left[ u^{U}_{i^*}(a_{i^*},y) \right] - \E_{\cI(C_{\vec P})}\left[ u^{U}_{i^*}(a_{i^*},y) \right] \right) \\
        & \ \ \ \ \ + \sum_{i \in N \neq i^*} \lambda_{j^*,i} \left( \E_{\cI(C'_{P_{j^*}}, C_{P_{-j^*}})}\left[ u^{U}_i(a_i,y) \right] -  \E_{\cI(C_{\vec B})}\left[ u^{U}_i(a_i,y) \right] \right) - 2\eps \\
        &> \lambda_{j^*,i^*} \delta - \Delta{_R}(j^*) \sum_{i\neq i^*} \lambda_{j^*,i} - 2\eps \\
        &= \lambda_{j^*,i^*} \left(\Delta{_R}(j^*) \frac{\sum_{i\neq i^*} \lambda_{j^*,i}}{\lambda_{j^*,i^*}} + \frac{2\eps}{\lambda_{j^{*},i^{*}}} \right) - \Delta{_R}(j^*) \sum_{i\neq i^*} \lambda_{j^*,i} - 2\eps\\
        &=0
    \end{align*}

    Therefore, provider $j^*$ can increase his expected utility by deviating to $C'_{P_j^*}$, and the strategy profile $C_{\vec P}$ cannot be a Nash equilibrium. This completes the proof.    
\end{proof}

%% file: files_arxiv/robustness.tex
So far, we have analyzed when pluralistically aligned outcomes emerge for a fixed population of users. In reality, new users --- with different preferences and information --- may regularly enter the ecosystem, possibly breaking the conditions that guarantee alignment. What happens to the guarantees enjoyed by pre-existing users when this occurs? 

We study this question in the two regimes where we have identified sufficient conditions for aligned outcomes: the personalized setting under Weak Market Alignment, and the anonymous setting under Strong Market Alignment. In each regime, we analyze the effect of a single user entering the market. 
More specifically, given an original game instance $G$, we define the \textit{augmented game} $G^+$ by adding user $n+1$ with utility $u^{U}_{n+1}:\cA_{n+1}\times \cY \to [0,1]$ and feature distribution $\phi(x^U_{n+1})$ that is the extension of $\phi$ to a joint distribution given the addition of user $n+1$'s private features $x^U_{n+1}$. We allow each provider $j$'s utility to change via the addition of an arbitrary function of the new user's action:
$$
u^{P}_j(a_{1:n+1},y) = (1-\beta_j)\cdot u^{P}_j(a_{1:n},y) + \beta_j \cdot f(a_{n+1}, y)
$$
where $\beta_j\in (0,1)$ and $f:\cA_{n+1}\times \cY \to [0,1]$ is an arbitrary function. We assume the original users' utilities stay the same. 

Importantly, this new user may violate our alignment conditions, and we do not hope to prove utility guarantees for her; instead, we study the effect of her entrance on the utility guarantees of the other users. We show that adding a user to the personalized setting with (previously satisfied) Weak Market Alignment does \emph{not} hurt the other users; they still attain the guarantees they had before her entrance. By contrast, adding a user to the anonymous setting with (previously satisfied) Strong Market Alignment can not only result in poor outcomes for the new user, but also break the utility guarantees that the other users previously had.

\begin{theorem}\label{thm:private-adding-users}
    Let $T\subseteq K$ be a set of providers that satisfy $(T, \eps_U,\eps_P)$-approximate Weak Market Alignment with respect to users $N$. In every Nash equilibrium of an augmented personalized game given by $G^+$,
    every original user $i\in N$ obtains expected utility at least $u^{U}_i(C_{S,T}^*(i)) - 2\eps_U - 2\eps_P\mu_i$, where $\mu_i = \sum_{j\in T:w_{j,i}>0} \frac{w_{j,i}}{\lambda_{j,i}}$.
\end{theorem}
Observe that here, for every original user $i\in N$, her optimal shared conversation rule $C^*_S(i)$ is unchanged by the addition of a new user; the joint distribution of $(x^{U}_i, y)$ is unaffected, and her utility $u^{U}_{i}$ depends only on her interaction (and not other users' actions). Thus, each original user recovers \textit{exactly the same guarantee} as before. 
\begin{proof}
    Let $C_{\Vec{P}}$ be any Nash equilibrium strategy profile of the augmented personalized game, and let $\cI_{NE} = \cI(C_{\Vec{P}})$ be the induced distribution.
    Now, assume for contradiction that there exists a user $i \in N$ such that her utility is strictly lower than $u^{U}_i(C_{S,T}^*(i)) - 2\eps_U - 2\eps_P\mu_i$. Consider a deviation by any provider $j \in T$ in his choice of conversation rule for user $i$ from $C_{P_j}(i)$ to $C_{S,T}^*(i)$, the optimal shared conversation rule for user $i$. Note that by definition, any provider could implement $C^*_S(i)$.
    
    Let $\cI_{dev} = \cI(C_{S,T}^*(i), C_{P_j}(-i), C_{P_{-j}})$ be the distribution induced by the deviation of provider $j$ in his chosen conversation rule for user $i$ (the conversation rules chosen by provider $j$ for all other users, $C_{P_j}(-i)$, remain unchanged). As user $i$'s utility was previously worse than $u^{U}_i(C_S^*(i))$, she will now choose to interact with the provider who implements $C^*_{S,T}(i)$, leading to the outcome under $\cI_{dev}$.  

    Now, by the Nash equilibrium condition, no provider $j$ has an incentive to deviate. Thus, for all providers $j\in T$:
    $$ \mathbb{E}_{\mathcal{I}_{dev}}[u^{P}_j(a_{1:n+1},y)] \le \mathbb{E}_{\mathcal{I}_{NE}}[u^{P}_j(a_{1:n+1},y)]. $$
    Since every user observes only the conversation rules chosen for them, this deviation can only change the induced distribution for user $i$ (in particular the distribution over her actions). 
    Thus, using our assumption on how provider utilities change under an additional user and applying approximate Weak Market Alignment, we have that for any provider $j \in T$:
    \begin{align*}
        & \mathbb{E}_{\mathcal{I}_{dev}}[(1-\beta_j)\cdot u^{P}_j(a_{1:n},y) + \beta_j \cdot f(a_{n+1}, y)] \le \mathbb{E}_{\mathcal{I}_{NE}}[(1-\beta_j)\cdot u^{P}_j(a_{1:n},y) + \beta_j \cdot f(a_{n+1}, y)] \\
        \iff& \mathbb{E}_{\mathcal{I}_{dev}}\left[(1-\beta_j)\cdot \left(\sum_{i\in N} \lambda_{j,i} F_{j,i}(a_i, y) + c_j - \eps_P \right) + \beta_j \cdot f(a_{n+1}, y) \right] \\ &\le \mathbb{E}_{\mathcal{I}_{NE}}\left[(1-\beta_j)\cdot \left(\sum_{i\in N} \lambda_{j,i} F_{j,i}(a_i, y) + c_j + \eps_P \right) + \beta_j \cdot f(a_{n+1}, y) \right] \\
        \iff & \mathbb{E}_{\mathcal{I}_{dev}}\left[(1-\beta_j)\cdot \left( \lambda_{j,i} F_{j,i}(a_i, y) + \sum_{i'\in N: i'\neq i} \lambda_{j, i'} F_{j, i'}(a_{i'}, y) + c_j - \eps_P \right) + \beta_j \cdot f(a_{n+1}, y) \right] \\ &\le \mathbb{E}_{\mathcal{I}_{NE}}\left[(1-\beta_j)\cdot \left( \lambda_{j,i} F_{j,i}(a_i, y) + \sum_{i'\in N: i'\neq i} \lambda_{j, i'} F_{j, i'}(a_{i'}, y) + c_j + \eps_P \right) + \beta_j \cdot f(a_{n+1}, y) \right] \\
        \iff & (1-\beta_j)\cdot \left( \mathbb{E}_{\mathcal{I}_{dev}}\left[\lambda_{j,i} F_{j,i}(a_i, y)\right] +  \mathbb{E}_{\mathcal{I}_{NE}}\left[\sum_{i'\in N: i'\neq i} \lambda_{j, i'} F_{j, i'}(a_{i'}, y)\right] +  c_j- \eps_P\right) + \beta_j \cdot \E_{\cI_{NE}}[f(a_{n+1}, y)] \\ &\le (1-\beta_j)\cdot \left( \mathbb{E}_{\mathcal{I}_{NE}}\left[\lambda_{j,i} F_{j,i}(a_i, y)\right] +  \mathbb{E}_{\mathcal{I}_{NE}}\left[\sum_{i'\in N: i'\neq i} \lambda_{j, i'} F_{j, i'}(a_{i'}, y)\right] +  c_j + \eps_P\right) + \beta_j \cdot \E_{\cI_{NE}}[f(a_{n+1}, y)] \\
        \iff & (1-\beta_j)\cdot\lambda_{j,i} \mathbb{E}_{\mathcal{I}_{dev}}\left[ F_{j,i}(a_i, y)\right] \leq (1-\beta_j)\cdot\lambda_{j,i} \mathbb{E}_{\mathcal{I}_{NE}}\left[ F_{j,i}(a_i, y)\right] + 2(1-\beta_j)\cdot \eps_P 
    \end{align*}
    where the third step follows from the fact that only the induced distribution over user $i$'s actions changes under the deviation. Since $\beta_j\in (0,1)$, for any $j$ with $\lambda_{j,i}>0$, we thus have:
    \[
    \mathbb{E}_{\mathcal{I}_{dev}}\left[ F_{j,i}(a_i, y)\right] \leq  \mathbb{E}_{\mathcal{I}_{NE}}\left[ F_{j,i}(a_i, y)\right]  + 2\frac{\eps_P}{\lambda_{j,i}}
    \]

    Taking a weighted sum over all providers using the non-negative weights $(w_{j,i})_{j\in T}$ from the Weak Market Alignment condition:
    $$ \sum_{j\in T} w_{j,i} \mathbb{E}_{\substack{\mathcal{I}_{dev}}}[F_{j,i}(a_i, y)] \le \sum_{j\in T} w_{j,i} \mathbb{E}_{\mathcal{I}_{NE}}[F_{j,i}(a_i, y)] + 2\eps_P\sum_{j\in T:w_{j,i}>0} \frac{w_{j,i}}{\lambda_{j,i}}. $$
    Note that by assumption, if $w_{j,i}>0$, then $\lambda_{j,i}>0$, and so we can apply the inequality here. 
    By linearity of expectation, this is equivalent to:
    $$ \mathbb{E}_{\substack{\mathcal{I}_{dev}}}\left[\sum_{j\in T} w_{j,i} F_{j,i}(a_i, y)\right] \le \mathbb{E}_{ \mathcal{I}_{NE}}\left[\sum_{j\in T} w_{j,i} F_{j,i}(a_i, y)\right] + 2\eps_P\sum_{j\in T:w_{j,i}>0} \frac{w_{j,i}}{\lambda_{j,i}}. $$
    Now we use the market alignment assumption, which states that $\sum_j w_{j,i} F_{j,i}(a_i, y)$ can be approximated by $u^{U}_i(a_i, y) + c_i$. We have that:
    $$ \mathbb{E}_{\substack{\mathcal{I}_{dev}}}\left[\sum_{j\in T} w_{j,i} F_{j,i}(a_i, y)\right] \ge \mathbb{E}_{\mathcal{I}_{dev}}[u^{U}_i(a_i,y) + c_i - \eps_U] = \mathbb{E}_{\mathcal{I}_{dev}}[u^{U}_i(a_i,y)] + c_i -\eps_U \ge u^{U}_i(C_{S,T}^*(i)) + c_i - \eps_U . $$
    Additionally:
    $$ \mathbb{E}_{\mathcal{I}_{NE}}\left[\sum_{j\in T} w_{j,i} F_{j,i}(a_i, y)\right] \le \mathbb{E}_{\mathcal{I}_{NE}}[u^{U}_i(a_i,y) + c_i +\eps_U]  = \mathbb{E}_{\mathcal{I}_{NE}}[u^{U}_i(a_i,y)] + c_i +\eps_U. $$
    Combining these inequalities, we get:
    $$ u^{U}_i(C_{S,T}^*(i)) + c_i -\eps_U \le \mathbb{E}_{\mathcal{I}_{NE}}[u^{U}_i(a_i,y)] + c_i +\eps_U + 2\eps_P\sum_{j\in T:w_{j,i}>0} \frac{w_{j,i}}{\lambda_{j,i}}. $$
    The constant offset $c_i$ cancels, and we are left with:
    $$ u^{U}_i(C_{S,T}^*(i)) -\eps_U \le \mathbb{E}_{\mathcal{I}_{NE}}[u^{U}_i(a_i,y)] +\eps_U + 2\eps_P\sum_{j\in T:w_{j,i}>0} \frac{w_{j,i}}{\lambda_{j,i}}. $$
    Thus:
    $$ \mathbb{E}_{\mathcal{I}_{NE}}[u^{U}_i(a_i,y)] \ge u^{U}_i(C_{S,T}^*(i)) - 2\eps_U - 2\eps_P\sum_{j\in T:w_{j,i}>0} \frac{w_{j,i}}{\lambda_{j,i}}. $$
which completes the proof.
\end{proof}

We now prove that in the anonymous game with Strong Market Alignment, adding a new user can break guarantees for preexisting users.

\begin{theorem}\label{thm:public-adding-users}
    Consider a game instance $G$ where the $(K, 0)$-approximate Strong Market Alignment condition is satisfied and full revelation is possible. Then, there exists $(f_{1:K}, \beta_{1:K}, u^{U}_{n+1}, P(x^U_{n+1}))$ such that, in the induced augmented anonymous game $G^+$, there is a Nash equilibrium 
    such that \emph{all} users $i\in N \cup \{n+1\}$ get utility $ < u^{U}_{i}(C^{*}_{S,K}(i))$.
\end{theorem}
\begin{proof}
    Consider a game with $R=1$ rounds with one user and one provider. Note that as we have only a single provider, computing an equilibrium between providers reduces to computing a best action for that provider. Let $x^{U}_{1} = \emptyset$, and let $x^{P}_{1} = y$. Furthermore, assume that the message space is large enough to encode $y$. Thus, full disclosure is possible. Let there be two states $y_{1}$ and $y_{2}$ which each occur with probability $\frac{1}{2}$, and let the user have action set $\{y_{1},y_{2},\bot\}$. Her utility is
    \[u^{U}_{1} = \mathbb{I}[a_1 = y] + \frac{2}{3}\mathbb{I}[a_{1} = \bot]\]
    Meanwhile, the provider's utility is 
     \[u^{P}_{1} = \frac{1}{2}\mathbb{I}[a_1 = y] + \frac{1}{3}\mathbb{I}[a_{1} = \bot]\]
     Note that this instance satisfies exact Strong Market Alignment: the provider will disclose $y$ in his conversation rule, and the user will always match the state and get expected utility $u^{U}_{i}(C_{S,K}^{*}(i)) = 1$. Now, consider the instance $A(G,f_{1},\frac{1}{2}, \{1,2,\bot\},u^{U}_{2},\emptyset, P(\emptyset))$. Here $f_{1} = \mathbb{I}[a_{2} = \bot]$, and thus the provider's utility $\hat{u}_{1}^{P}$ in this augmented game is $$\hat{u}^{P}_{1} = \frac{1}{4}\mathbb{I}[a_1 = y] + \frac{1}{6}\mathbb{I}[a_{1} = \bot] +\frac{1}{2}\mathbb{I}[a_{2} = \bot]$$ 
    Additionally, $u^{U}_{2} = \mathbb{I}[a_{2}=y] + \frac{2}{3}\mathbb{I}[a_{2}=\bot]$.

    Now, we will show that under any best action from the provider, the user $1$ has expected utility $\frac{2}{3} < u^{U}_{i}(C^{*}_{S,K}(i))$. To see this, note that if the provider does not provide any information to the users, both users will select action $\bot$. In this case, the provider's utility is $\frac{1}{6} + \frac{1}{2} = \frac{2}{3}$. Furthermore, note that as both users have the same hidden information and utilities, they will have the same distribution over actions against any provider conversation rule. Now, consider the case where the provider commits to a conversation rule where, if the users best-respond, each user has probability $p > 0$ of playing something other than $\bot$. If both users plays something other than $\bot$, the provider's utility is at most $\frac{1}{4}$. So his expected utility is at most 
    $$\frac{1}{4} p + (1-p)\frac{2}{3} < \frac{2}{3}$$ This is strictly less than that of providing no information, so the best response to the provider's conversation rule for user $1$ will always be to play $\bot$. Thus, her expected utility is at most $\frac{2}{3}$, completing the proof. 
\end{proof}

%% file: files_arxiv/experiments.tex
\subsection{Experimental Setup}
\paragraph{Dataset.} Following \cite{collina2025emergentalignmentcompetition}, we empirically test our alignment conditions using the OpinionQA dataset \citep{pmlr-v202-santurkar23a}, which provides answer distributions for both demographic groups and LLMs on survey questions.
This setting maps to our framework: demographic groups play the role of users (with heterogeneous preferences), LLMs play the role of providers, and survey questions play the role of states. While the LLMs are not strategically competing over these users, this setup lets us test whether the structural alignment conditions hold in realistic preference data.
We test on demographic partitions (Political Ideology, Political Party, Education, Religion, Income, Age, Race, Census Region, Sex) across two survey waves from the Pew American Trends Panel: W29 (2017; gender and health topics) and W82 (2021; economy and democracy topics). We present results for Political Ideology here and refer the reader to the \Cref{app:experiments} for results on other partitions.

\paragraph{Utilities.} We use $p_i(\cdot \mid y)$ for the empirical answer distribution of demographic group $i$, and $q_j(\cdot \mid y)$ for the predicted answer distribution of model $j$.
For each demographic group $i$ and question $y$, we define user utility as $u_i^U(a,y) = p_i(a \mid y)$, the probability that group $i$ selects answer $a$. Intuitively, this says a user is better served by actions that her demographic group would more frequently endorse. Similarly, for each model $j$, we define the provider per-user component $v_j(a,y) = q_j(a \mid y)$, the model's predicted probability, and provider utility as $u_j^P(a_{1:n}, y) = \frac{1}{n}\sum_{i=1}^n v_j(a_i,y)$, which treats each demographic group equally. 
This formulation corresponds to a linear scoring rule, which we use as the default because it makes provider utility exactly additively separable ($\eps_P=0$), isolating the user alignment error $\epsilon_U$ as the sole source of approximation for \Cref{def:separability}. We show robustness to alternative scores (log, Brier) in \Cref{app:experiments}.

\begin{figure}
\centering
\includegraphics[width=0.85\textwidth]{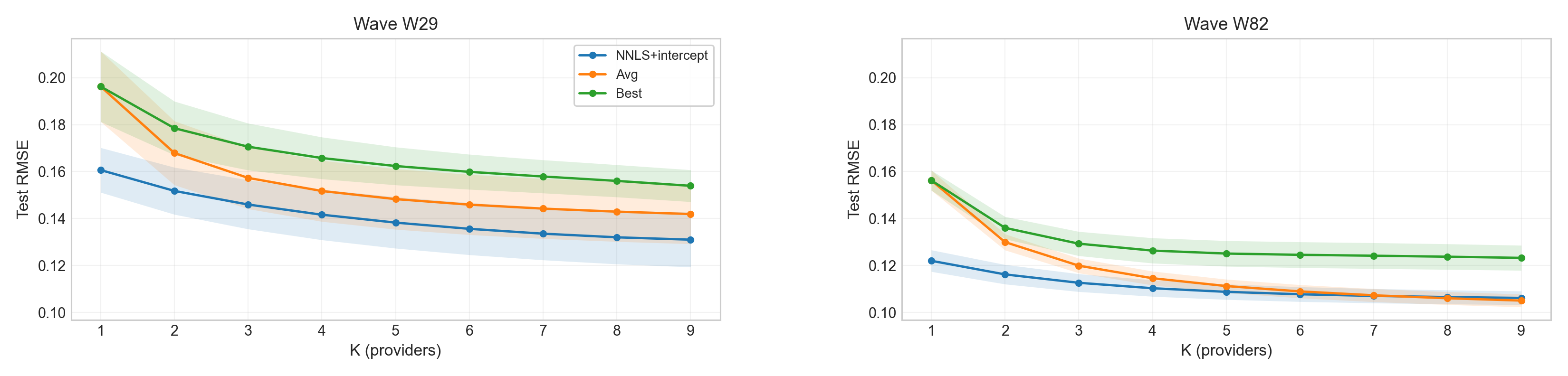}
\vspace{-0.5em}

\includegraphics[width=0.85\textwidth]{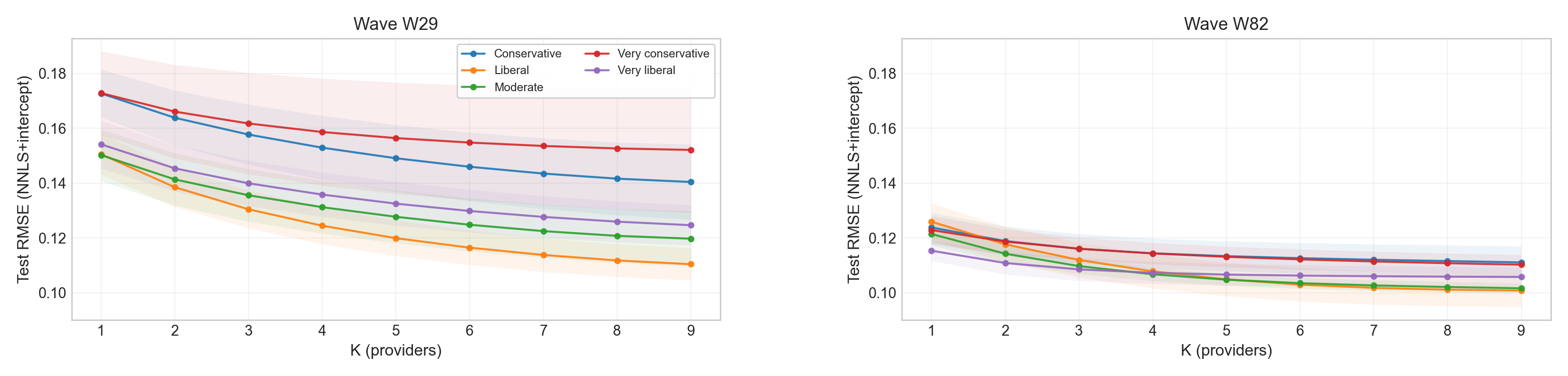}
\caption{(Top) \textbf{Weak Market Alignment improves with more providers.} Mean root mean squared error (RMSE, across groups) vs.\ number of available providers $K$. Nonnegative least squares (NNLS)+intercept outperforms baselines; error decreases with $K$. (Bottom) \textbf{Some groups are harder to align.} Per-group RMSE vs.\ $K$. Within each partition, alignment quality varies across demographic groups. Shaded bands: CV standard error.}
\label{fig:weak-market}
\end{figure}
\paragraph{Alignment Assumptions.}
\textit{Weak Market Alignment.} Because $\eps_P = 0$, we fit $u_i^U \approx \sum_{j \in K} w_{j,i} v_j + c_i$, with $w_{j,i} \ge 0$.
\textit{Strong Market Alignment.} We fit $u_j^P \approx \sum_i \lambda_{j,i} \, u_i^U + c_j$ with $\lambda_{j,i} \ge 0$ via NNLS on sampled action profiles.
For each question $y$, we generate profiles by sampling each user's action independently from their empirical marginal, $a_{1:n}^{(m)} \sim \prod_i p_i(\cdot \mid y)$, and evaluate on held-out questions.
Beyond approximation error $\eps$, the key quantity is the \emph{transfer factor} $1/\lambda_i^*$, where $\lambda_i^* = \max_j \lambda_{j,i}$.
This measures how well some provider covers user $i$: small means strong coverage (good) and large means weak coverage (bad).
It directly scales the anonymous-game slack term $2\eps/\lambda_i^*$ in Theorem~\ref{thm:public-natalie}, so large transfer factors make the bound worse.

\paragraph{Measurement.} Our theorems rely on worst-case error bounds over all action profiles and states. Since computing these exactly is prohibitive, we use average-case Root Mean Squared Error (RMSE) and transfer factors (via 5-fold cross-validation over questions) as tractable proxies. Average-case RMSE lower-bounds the worst-case error, so small RMSE is necessary but not sufficient for the theoretical conditions to hold at the level the RMSE suggests. Nevertheless, when the action and state spaces are modest in size as they are here, it can serve as a good proxy.

\subsection{Findings}
\begin{figure}
\centering
\includegraphics[width=0.85\textwidth]{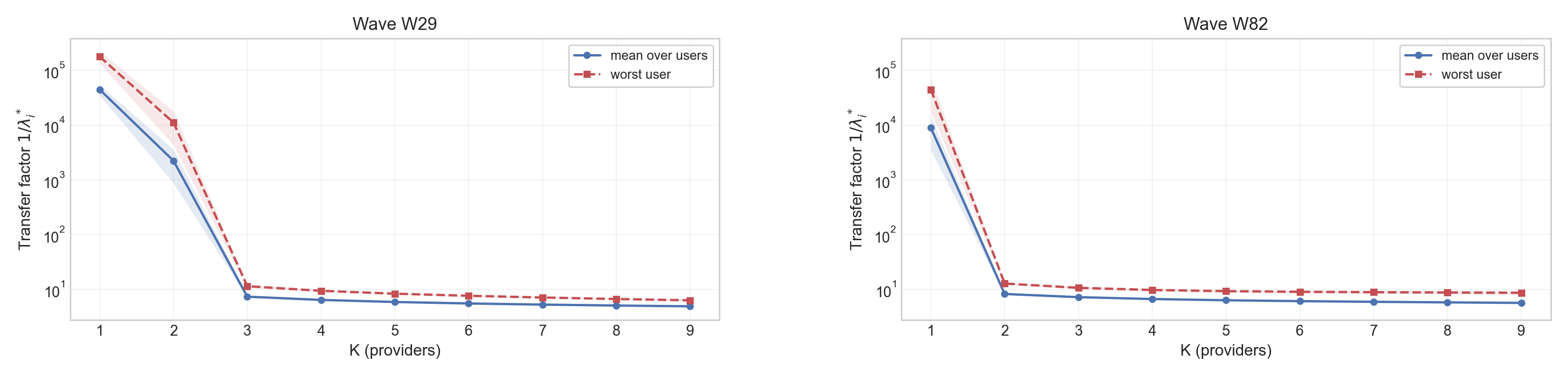}
\vspace{-0.5em}

\includegraphics[width=0.85\textwidth]{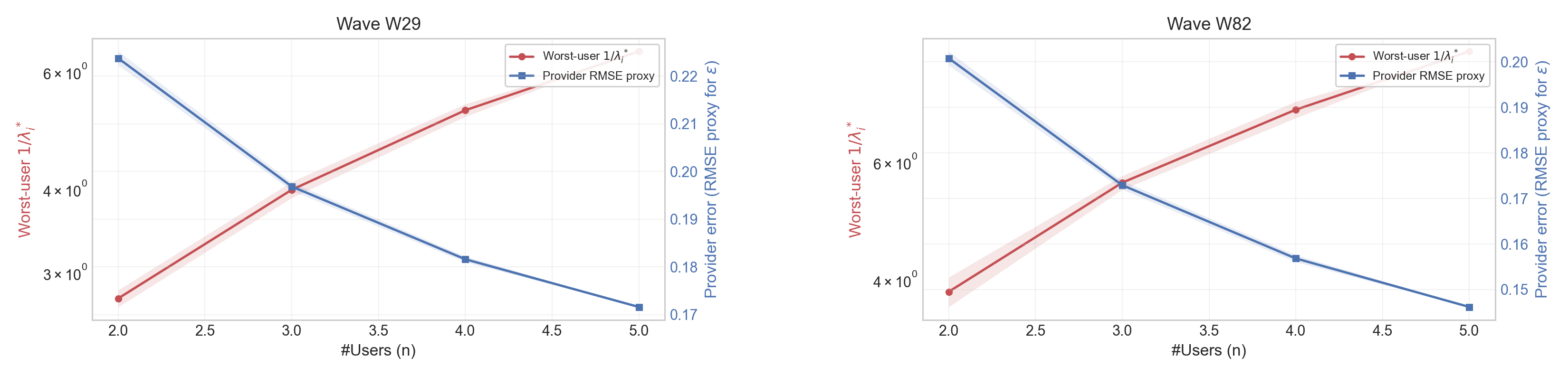}
\caption{(Top) \textbf{More providers reduce transfer factors.} Transfer factor $1/\lambda_i^*$ vs.\ number of providers $K$: mean over users (solid) and worst user (dashed). The dramatic drop shows how provider diversity improves coverage. (Bottom) \textbf{More users: lower error but harder coverage.} Worst-user transfer factor $1/\lambda_i^*$ (red, left axis) vs.\ provider fitting error $\eps$ (blue, right axis) as number of users increases.}
\label{fig:transfer-vs-k}
\end{figure}
\begin{figure}
\centering
\includegraphics[width=0.72\textwidth]{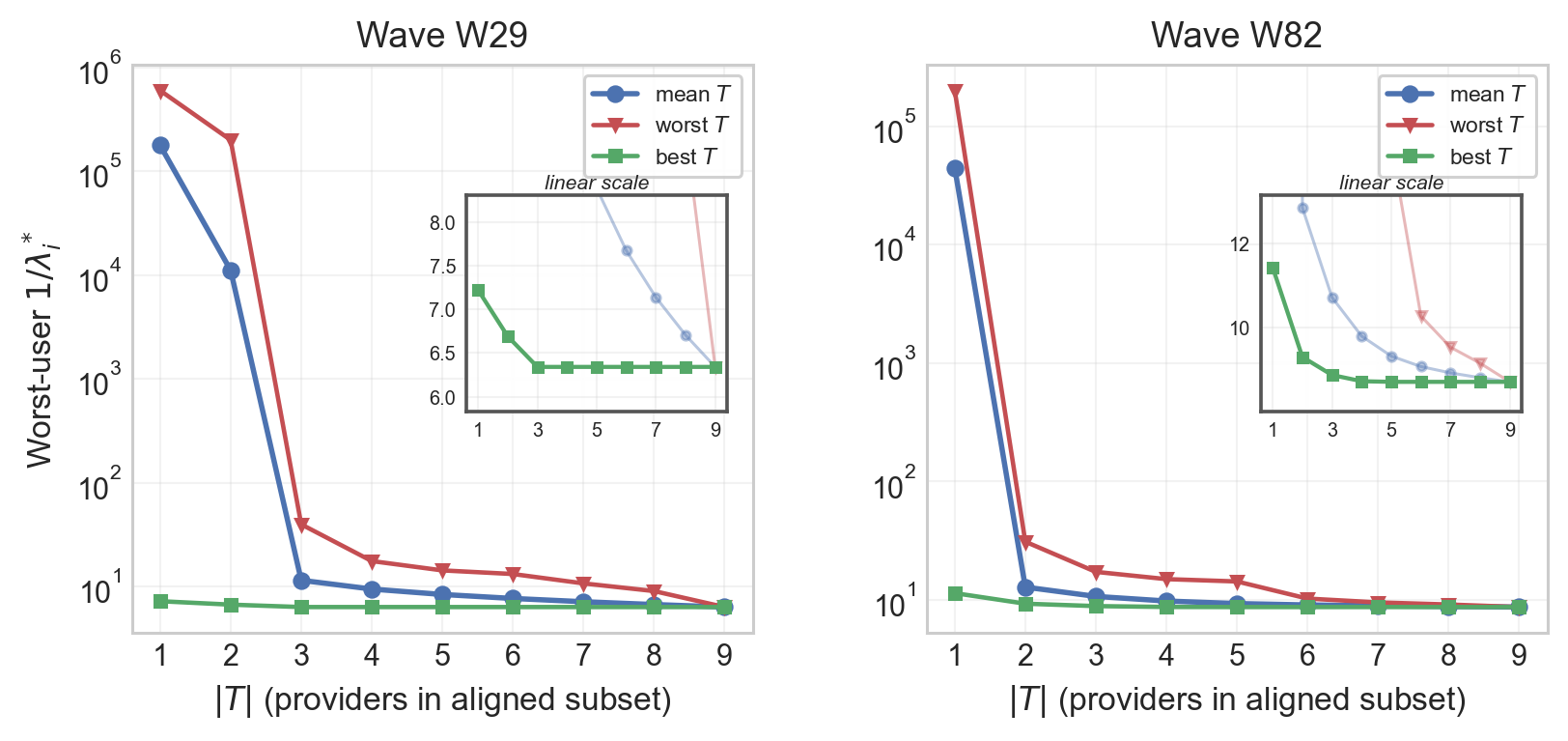}
\caption{\textbf{Not all providers need to be aligned.} Worst-user transfer factor $1/\lambda_i^*$ vs.\ size of aligned subset $|T|$ (log scale). Mean $T$ (blue) and worst $T$ (red) show that a random small subset often fails; the best subset (green) achieves low transfer even at $|T|=1$. Insets zoom to linear scale, revealing the gap between best and mean/worst at large $|T|$.}
\label{fig:subset}
\end{figure}
\paragraph{More providers improve Weak Alignment Fit.} Nonnegative combinations of model utilities approximate group preferences substantially better than either the best single model or a simple average, and fit improves as $K$ grows, though the gap decreases as well (Figure~\ref{fig:weak-market}, top).

\paragraph{Heterogeneous difficulty exists across groups.} Within a demographic partition, some groups are consistently harder to approximate than others (Figure~\ref{fig:weak-market}, bottom).

\paragraph{Strong alignment is harder than Weak.} The Strong alignment fitting error $\eps$ is approximately 30–40\% larger than the Weak alignment error $\eps_U$ at comparable $K$, though these measure different regression targets (cf.\ Figures~\ref{fig:weak-market} and \ref{fig:transfer-vs-k}). The corresponding theorem slacks are $\eps_U$ for the personalized game ($\eps_P = 0$) versus $2\eps/\lambda_i^*$. Combined with transfer amplification at small $K$, the effective anonymous-game bound can be orders of magnitude looser than the personalized bound.

\paragraph{More providers dramatically improve coverage for Strong Alignment.} At $K=1$, worst-user transfer factors range from $\sim7$ (the best individual model) to $>10^4$ (the worst).
Adding just one or two more providers ($K=2,3$) drops transfer factors by orders of magnitude (Figure~\ref{fig:transfer-vs-k}, top). The appendix visualizes the learned weights $\lambda_{j,i}$ and $w_{j,i}$ (Figure~\ref{fig:weights-appendix} in \Cref{app:experiments}), showing that different models specialize in different groups.

\paragraph{Using a subset of providers improves Strong Alignment.} Our theorems require only a \emph{subset} $T\subseteq K$ to satisfy alignment. This makes the Stronger Alignment condition easier to satisfy.
Figure~\ref{fig:subset} varies $|T|$, enumerating all subsets and tracking best, mean, and worst-case transfer factors.
A well-chosen single provider already achieves low transfer ($\sim7-10$), while the average single provider has transfer factors orders of magnitude larger.
The gap reflects specialization: many individual models place near-zero mass on at least one group, making $\lambda_i^*$ tiny for that group and inflating $1/\lambda_i^*$.
By $|T|=3$, even the mean subset reaches reasonable transfer.
This confirms that provider diversity, not universality, drives the guarantees.

\paragraph{More users create a coverage-error trade-off.} Increasing the number of demographic groups (users) within a partition has opposing effects (Figure~\ref{fig:transfer-vs-k}, bottom):
fitting error $\eps$ \emph{decreases} (more groups means more signal), but the worst-case transfer factor $1/\lambda_i^*$ from Theorem~\ref{thm:public-natalie} \emph{increases} (harder to cover everyone).
Over the range $n \in \{2,3,4,5\}$, transfer-factor growth outpaces the reduction in fitting error, suggesting that, at least at this scale, finer user segmentation does not improve the effective bound without additional provider diversity.

%% file: files_arxiv/appendix_proofs.tex

\section{Details about Alignment Conditions}
\begin{proposition}[Strong Market Alignment Implies Weak Market Alignment]\label{prop:implies}
    Fix a set $T\subset K$ of providers. Suppose the provider utilities satisfy the $(T,\eps)$-approximate Strong Market Alignment condition. Then, the provider utilities satisfy the $(T, \eps, 0)$-approximate Weak Market Alignment condition. When $\eps=0$, Strong Market Alignment implies Weak Market Alignment. 
\end{proposition}
\begin{proof}
    Take $F_{j,i}(a_i, y) = u^{U}_i(a_i, y)$ for all $j\in T$ and $i\in N$. Then the Weak Market Alignment condition is equivalent to the Strong Market Alignment condition:
    \begin{equation*}\label{eq:separableU}
    \left| u^{P}_j\big(a_{1:n},y\big)\ - \left(\sum_{i\in N}\lambda_{j,i}\,F_{j,i}(a_i,y) + c_j\right)\right|  = \left| u^{P}_j\big(a_{1:n},y\big)\ - \left(\sum_{i\in N}\lambda_{j,i}\,u_{i}(a_i,y) + c_j\right)\right| \leq \eps
    \end{equation*}
    
    By assumption there exists at least one provider $j$ with $\lambda_{j,i}>0$. Let $n^*$ be $|\{j: \lambda_{j,i}>0\}|$. 
    Then, the Weak Market Alignment condition is exactly satisfied with $\eps_U=0$ by taking weights $w_{j,i} = 1/n^*$ for all $j$ with $\lambda_{j,i}>0$, $w_{j,i} = 0$ for all other $j$, and $c_i=0$.  
\end{proof}

\begin{proposition}[Weak Market Alignment is Strictly Weaker]\label{prop:strict}
    There exist provider and user utilities that exactly satisfy the Weak Market Alignment condition, but for every $\eps<1/2$, the $(T,\eps)$-Strong Market Alignment condition fails.
\end{proposition}
    
\begin{proof}
Consider an instance with $n=k=2$ (i.e. 2 providers and 2 users). Let $\cY = \{(y_1,y_2)\in \mathbb{R}^2 : y_1,y_2\geq 0,\ y_1+y_2=1\}$ and $\cA_1=\cA_2=[0,1]$. Define:
    \begin{align*}
        &F_{1,1}(a_1,y) = \left \langle \begin{bmatrix}
           a_1/2 \\ 0 \end{bmatrix}, y \right \rangle , \ \ 
           F_{1,2}(a_2,y) = \left \langle \begin{bmatrix}
           a_2/2 \\ 0 \end{bmatrix}, y \right \rangle \\
        &F_{2,1}(a_1,y) = \left \langle \begin{bmatrix}
           0 \\ a_1/2 \end{bmatrix}, y \right \rangle , \ \ 
           F_{2,2}(a_2,y) = \left \langle \begin{bmatrix}
           0 \\ a_2/2 \end{bmatrix}, y \right \rangle 
    \end{align*}
Suppose $u^{P}_j(a_1,a_2,y)=F_{j,1}(a_1,y) + F_{j,2}(a_2,y)$ for each provider $j$ and 
    \[
        u^{U}_i(a_i,y) = \frac{1}{2}F_{1,i}(a_i,y)+\frac{1}{2}F_{2,i}(a_i,y) 
        = \left \langle \begin{bmatrix} \frac{a_i}{4} \\ \frac{a_i}{4} \end{bmatrix}, y \right \rangle
    \]
for each user $i$. Since $a_i\in[0,1]$ and $y_1,y_2\in[0,1]$ with $y_1+y_2=1$, we have $F_{j,i}(a_i,y)\in[0,1/2]$, hence $u^{P}_j(a_1,a_2,y)\in[0,1]$ and $u^{U}_i(a_i,y)\in[0,1]$. Thus Weak Market Alignment (Definition \ref{def:separability}) is satisfied.

Now, for provider $1$, we have that 
    \[
        u^{P}_1(a_1,a_2,y)=\left \langle \begin{bmatrix} \frac{a_1+a_2}{2} \\ 0 \end{bmatrix}, y \right \rangle.
    \]
We show that $(T,\eps)$-Strong Market Alignment fails for every $\eps<\frac{1}{2}$. Fix $\eps<\frac{1}{2}$ and suppose for contradiction that there exist $\lambda_1,\lambda_2\ge 0$ and $c\in\mathbb{R}$ such that for all $a_1,a_2\in[0,1]$ and $y\in\cY$,
    \[
        \left|u^{P}_1(a_1,a_2,y) - \left(\lambda_1 u^{U}_1(a_1,y) + \lambda_2 u^{U}_2(a_2,y) + c\right)\right|\le \eps.
    \]
Take $a_1=a_2=1$, and let $y^{(0)}=(0,1)$ and $y^{(1)}=(1,0)$. Then $u^{U}_1(1,y^{(0)})=u^{U}_1(1,y^{(1)})=\frac{1}{4}$ and $u^{U}_2(1,y^{(0)})=u^{U}_2(1,y^{(1)})=\frac{1}{4}$, while $u^{P}_1(1,1,y^{(0)})=0$ and $u^{P}_1(1,1,y^{(1)})=1$. Writing $v=\frac{\lambda_1}{4}+\frac{\lambda_2}{4}+c$, the inequality implies both $|v|\le \eps$ and $|1-v|\le \eps$, hence $\eps\ge \frac{1}{2}$, a contradiction.

The same argument applies for $u^{P}_2$ (swap the roles of the two coordinates of $y$), so no provider satisfies $(T,\eps)$-Strong Market Alignment for any $\eps<\frac{1}{2}$. Therefore Strong Market Alignment fails.
\end{proof}

%% file: files_arxiv/appendix_experiments.tex
\paragraph{Scoring rules.}
The main text defines both user and provider utilities using the \emph{linear} score, i.e., the probability assigned to the realized action.
Keeping the user utility as $u_i^U(a,y) = p_i(a \mid y)$, we change the provider component $v_j(a,y)$:
\begin{itemize}
  \item \textbf{Linear:} $v_j(a,y) = q_j(a \mid y)$. Additive across users, giving $\eps_P = 0$.
  \item \textbf{Log:} $v_j(a,y) \propto \log q_j(a \mid y)$ (normalized to $[0,1]$). Rewards assigning high probability mass to the realized option.
  \item \textbf{Brier:} $v_j(a,y) = \frac{1}{2} + q_j(a \mid y) - \frac{1}{2}\|q_j(\cdot \mid y)\|^2$ the expected Brier score contribution of action $a$, affine rescaled to $[0,1]$). Penalizes miscalibration and overconfidence.
\end{itemize}
Log and Brier scores are still separable since $u_j^P = \frac{1}{n}\sum_i v_j(a_i,y)$. However, we need to fit the user utility with these new basis functions $v_j$ across different scores. Nevertheless, across all three scores (Figure~\ref{fig:scoring-rules-appendix}): increasing $K$ improves Weak Alignment root mean squared error (RMSE) and rapidly reduces transfer factors for Strong Alignment, though the log score consistently yields higher converged transfer factors than linear or Brier.

\begin{figure}[t]
\centering
\begin{subfigure}[b]{0.48\textwidth}
  \centering
  \includegraphics[width=\textwidth]{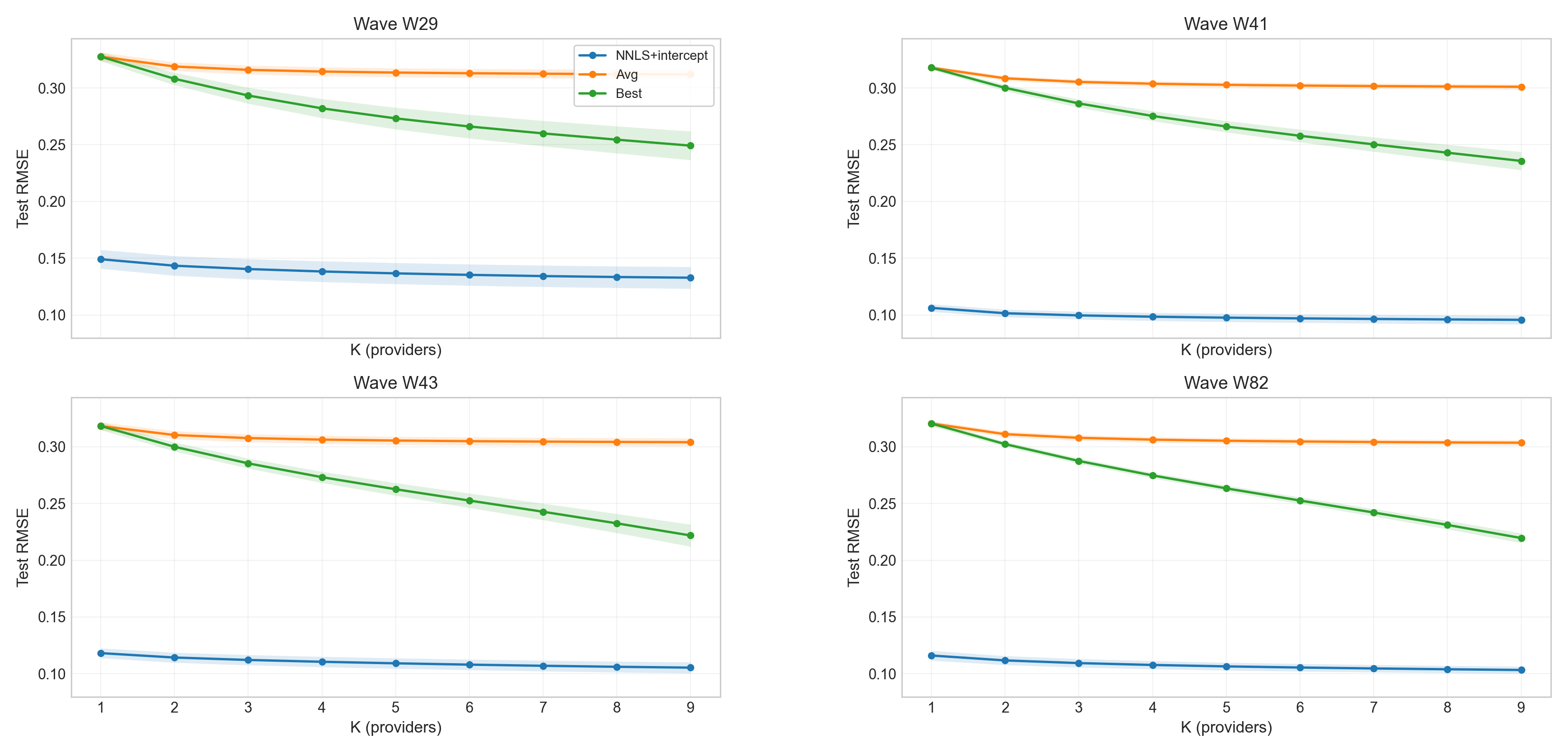}
  \caption{Weak alignment (log score)}
\end{subfigure}
\hfill
\begin{subfigure}[b]{0.48\textwidth}
  \centering
  \includegraphics[width=\textwidth]{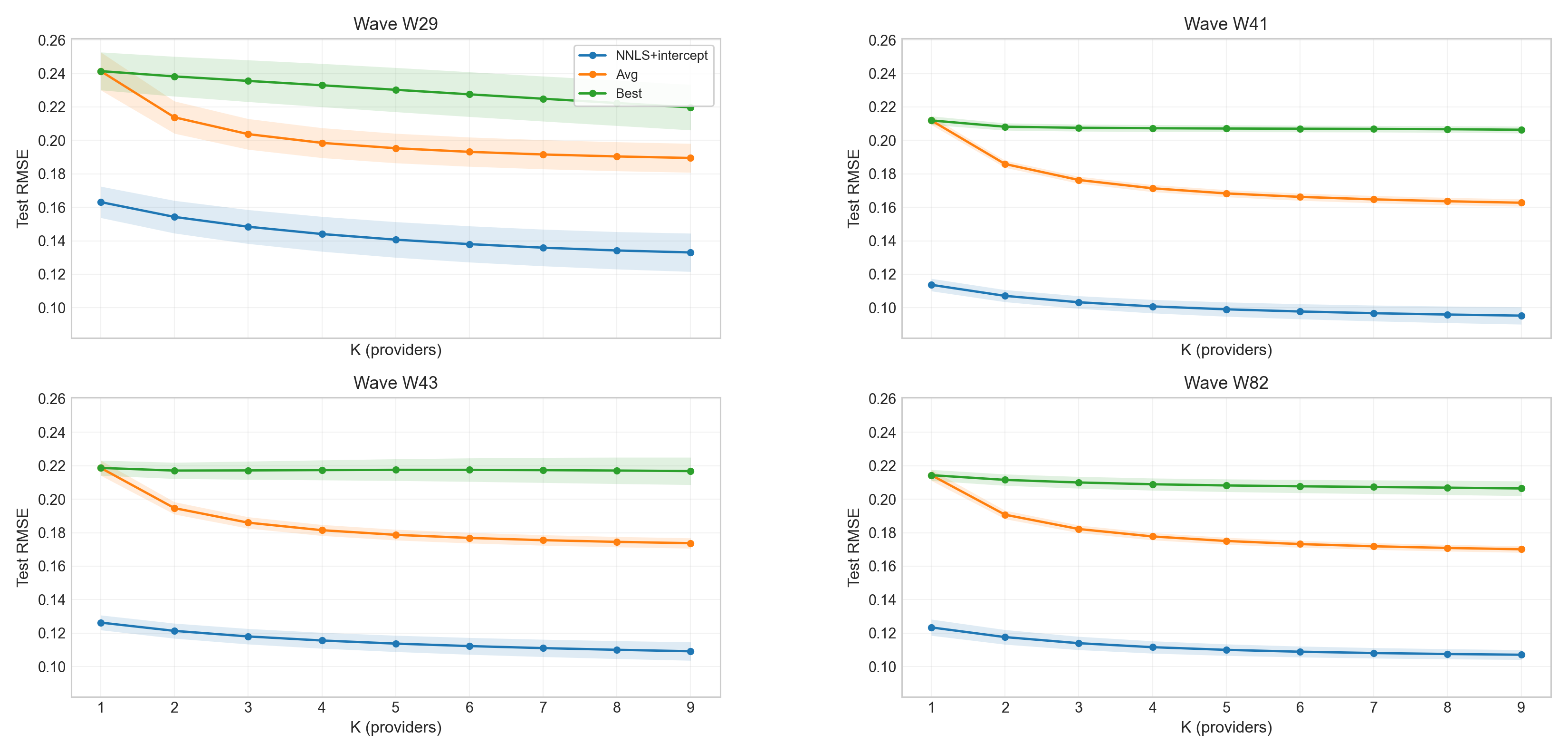}
  \caption{Weak alignment (Brier score)}
\end{subfigure}

\vspace{0.4em}
\begin{subfigure}[b]{0.48\textwidth}
  \centering
  \includegraphics[width=\textwidth]{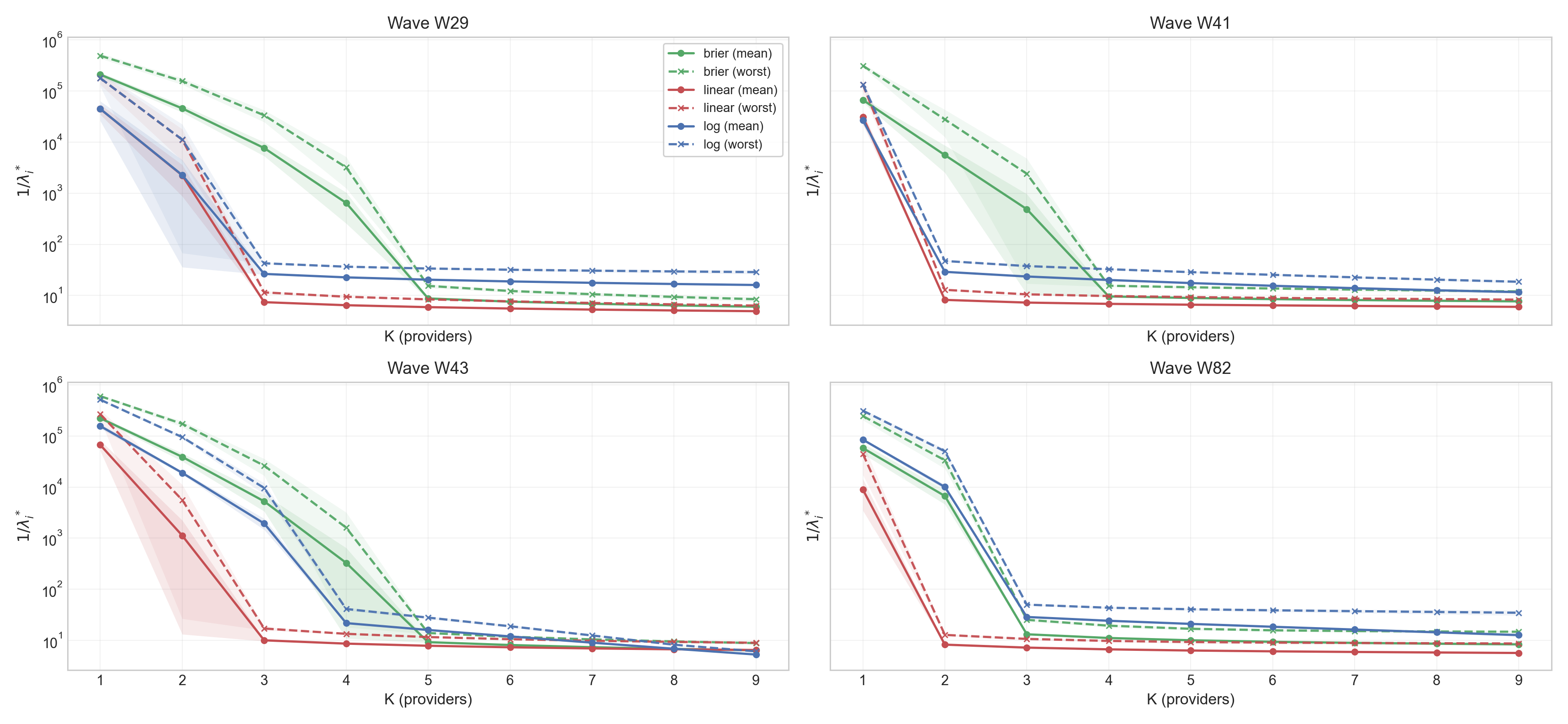}
  \caption{Transfer factor (all scores, Pol.\ Ideology)}
\end{subfigure}
\hfill
\begin{subfigure}[b]{0.48\textwidth}
  \centering
  \includegraphics[width=\textwidth]{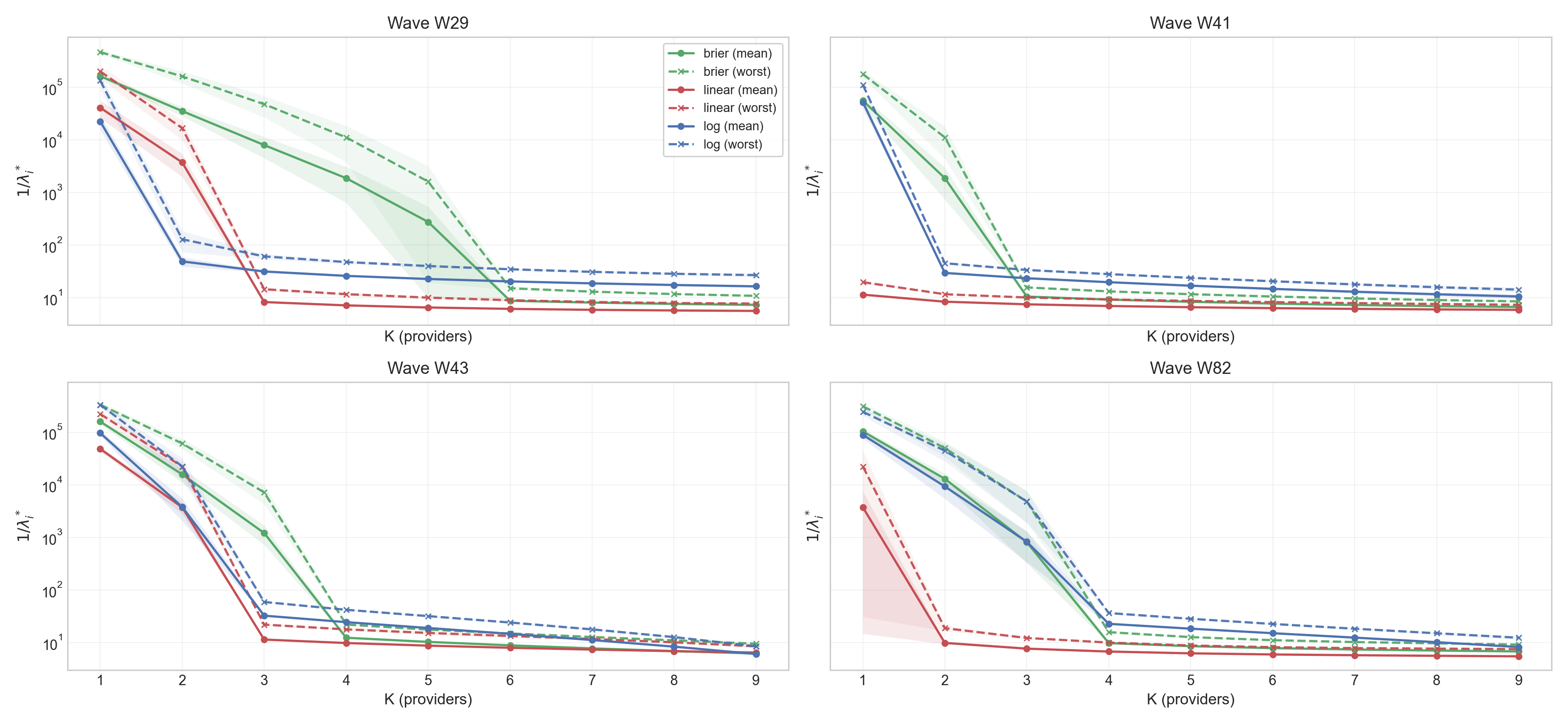}
  \caption{Transfer factor (all scores, Education)}
\end{subfigure}
\caption{\textbf{Scoring rule robustness.} (a,b) Weak alignment RMSE under log and Brier scores (Political Ideology). (c,d) Transfer factors across all three scores. Patterns are stable.}
\label{fig:scoring-rules-appendix}
\end{figure}

\paragraph{Weak alignment across demographics.}
Figure~\ref{fig:weak-all-demos} shows weak alignment RMSE vs.\ $K$ for eight additional demographic partitions.
Nonnegative least squares (NNLS) consistently outperforms baselines.
For Race and Sex in W82, equal-weight averaging approaches the NNLS solution at large $K$, suggesting that models' predicted distributions are more uniformly useful across these groups; in W29 a modest gap persists.
For more polarized partitions (Religion, Political Party), optimized weighting provides larger improvements, consistent with greater heterogeneity between different groups' preferences.

\begin{figure}[t]
\centering
\begin{subfigure}[b]{0.48\textwidth}
  \centering
  \includegraphics[width=\textwidth]{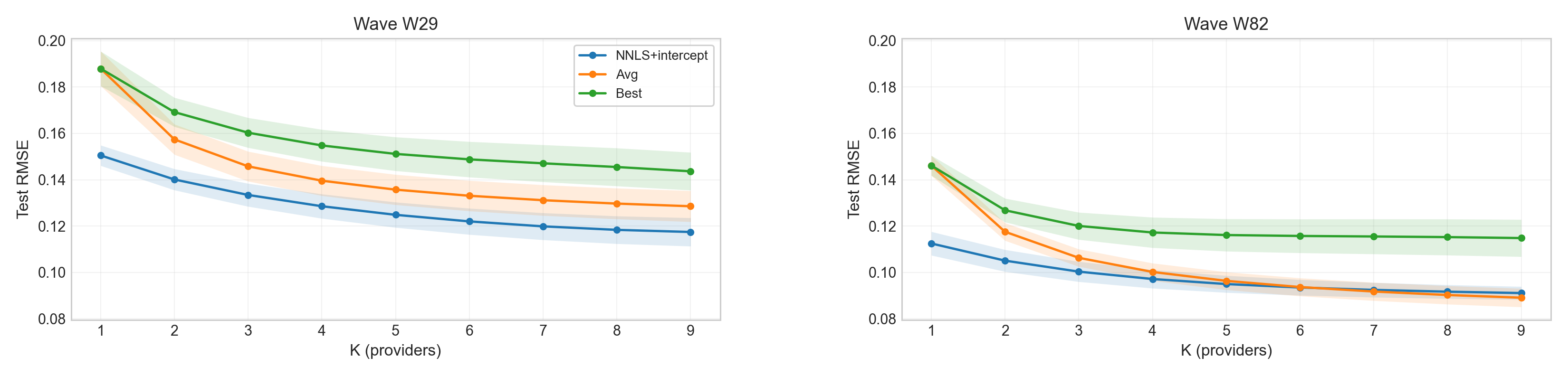}
  \caption{Education (6 groups)}
\end{subfigure}
\hfill
\begin{subfigure}[b]{0.48\textwidth}
  \centering
  \includegraphics[width=\textwidth]{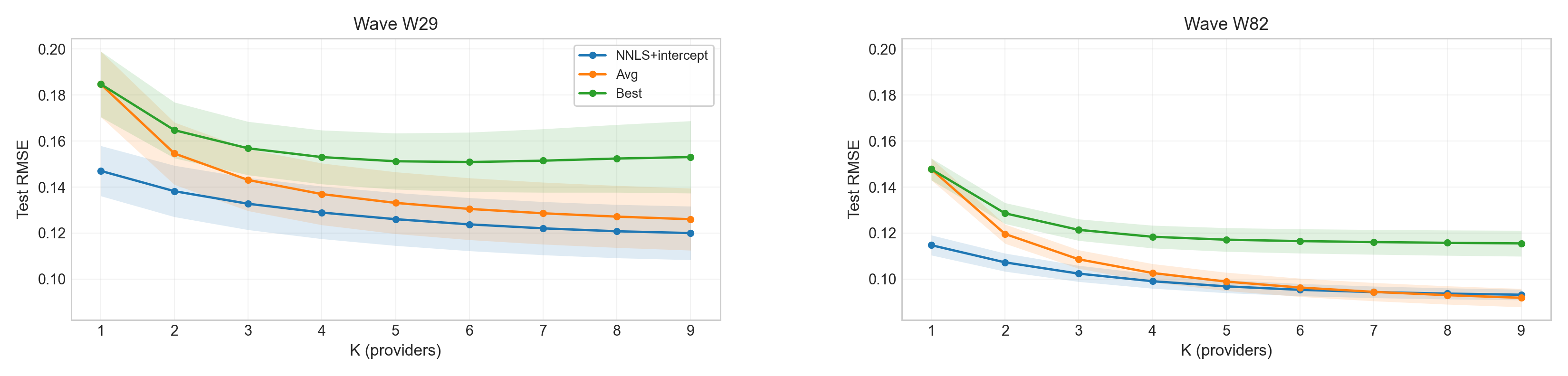}
  \caption{Race (5 groups)}
\end{subfigure}

\vspace{0.4em}
\begin{subfigure}[b]{0.48\textwidth}
  \centering
  \includegraphics[width=\textwidth]{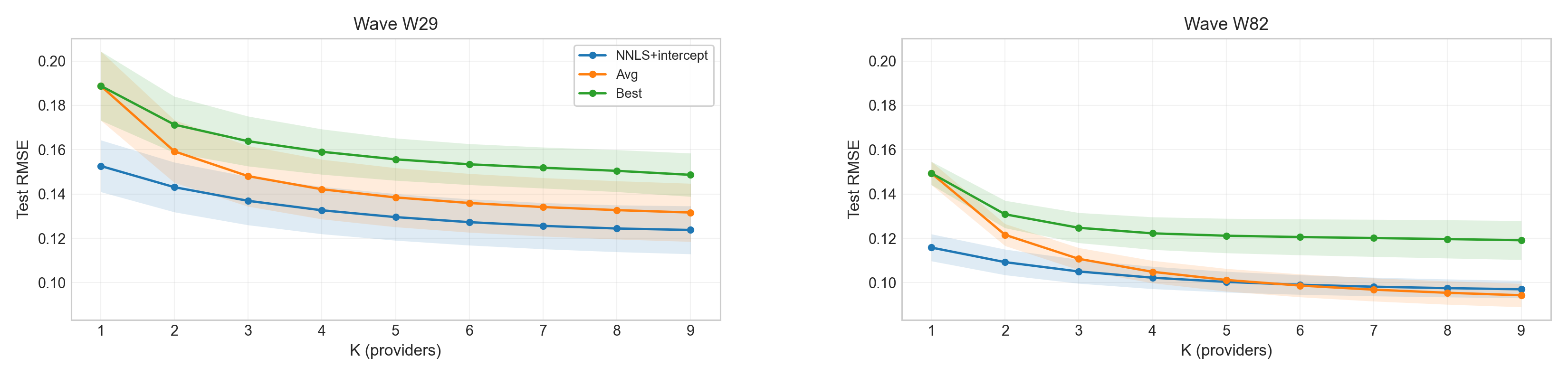}
  \caption{Religion (7--8 groups)}
\end{subfigure}
\hfill
\begin{subfigure}[b]{0.48\textwidth}
  \centering
  \includegraphics[width=\textwidth]{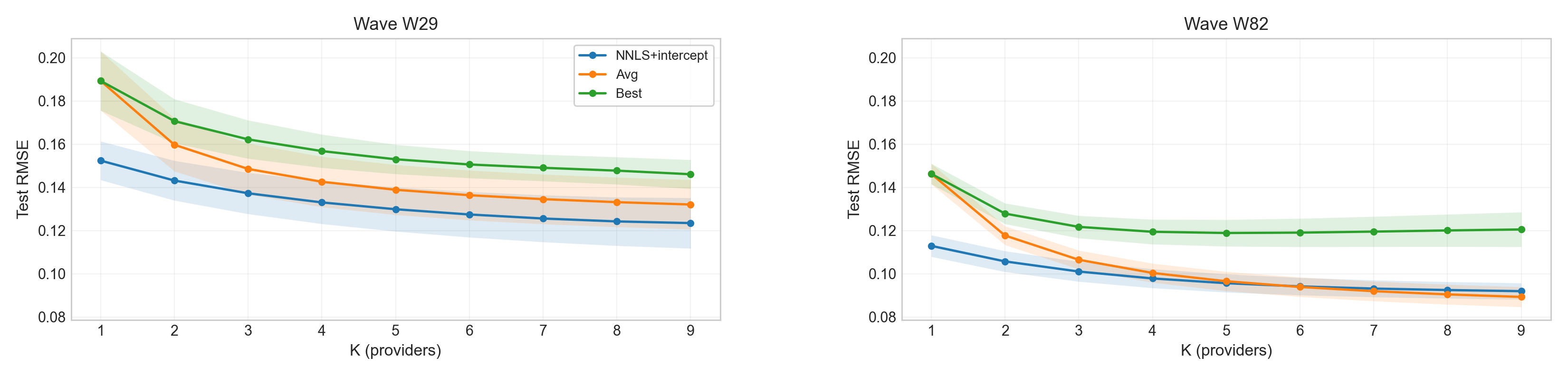}
  \caption{Income (5 groups)}
\end{subfigure}

\vspace{0.4em}
\begin{subfigure}[b]{0.48\textwidth}
  \centering
  \includegraphics[width=\textwidth]{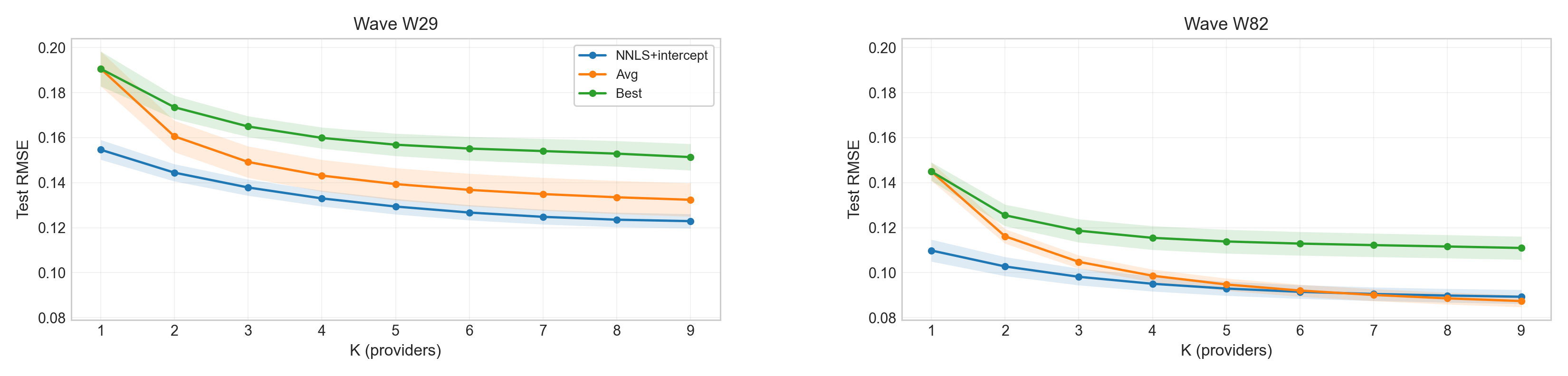}
  \caption{Age (4 groups)}
\end{subfigure}
\hfill
\begin{subfigure}[b]{0.48\textwidth}
  \centering
  \includegraphics[width=\textwidth]{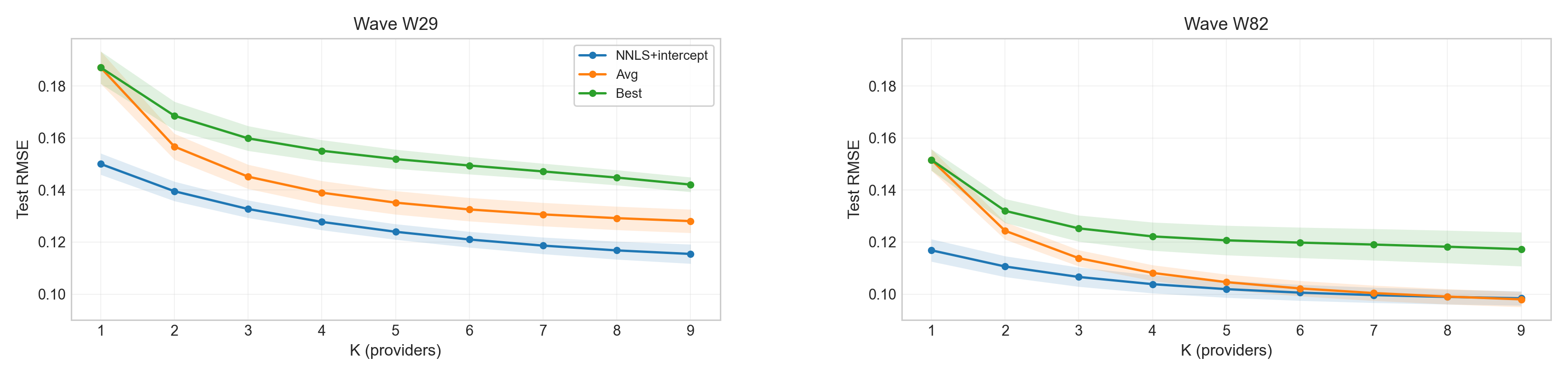}
  \caption{Pol.\ Party (4 groups)}
\end{subfigure}

\vspace{0.4em}
\begin{subfigure}[b]{0.48\textwidth}
  \centering
  \includegraphics[width=\textwidth]{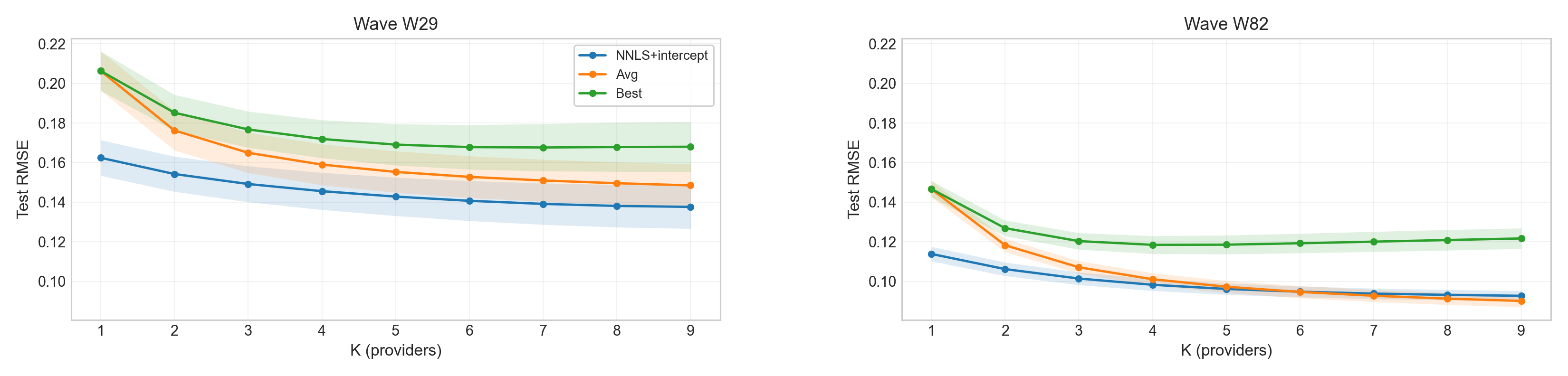}
  \caption{Sex (2 groups)}
\end{subfigure}
\hfill
\begin{subfigure}[b]{0.48\textwidth}
  \centering
  \includegraphics[width=\textwidth]{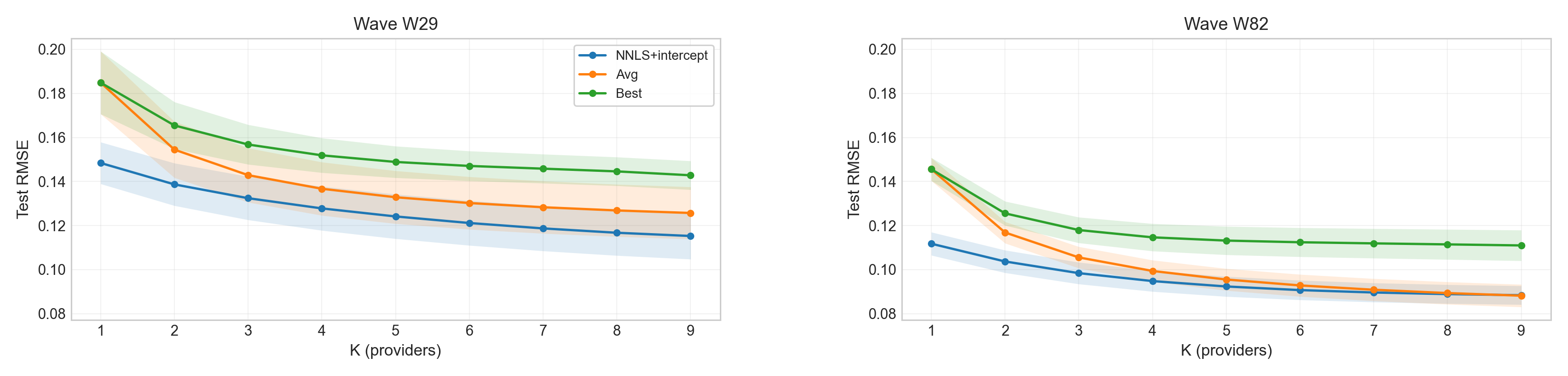}
  \caption{Census Region (4 groups)}
\end{subfigure}
\caption{\textbf{Weak alignment across demographics.} Mean RMSE vs.\ $K$ for eight additional partitions. NNLS outperforms baselines in all cases.}
\label{fig:weak-all-demos}
\end{figure}

\paragraph{Per-group RMSE.}
Figures~\ref{fig:pergroup-all} and \ref{fig:pergroup-all-2} show per-group weak alignment RMSE for all partitions (excluding Sex).
Within each partition, alignment quality varies across groups: 
Very Conservative is consistently hardest in Political Ideology; in Religion, the hardest group varies by wave (Protestant in W29, Agnostic in W82); in Income, the highest bracket (\$100,000+) is hardest in W82.
The hardest groups have RMSE 5–20\% above the partition mean, which would dominate the worst-case bound in our theorems.

\begin{figure}[t]
\centering
\begin{subfigure}[b]{0.48\textwidth}
  \centering
  \includegraphics[width=\textwidth]{figs/opinionqa/by_demo/POLIDEOLOGY/alice_alignment_mse_vs_k_by_group_agg_linear_demo_POLIDEOLOGY_panels_W29_W82.png}
  \caption{Political Ideology (5 groups)}
\end{subfigure}
\hfill
\begin{subfigure}[b]{0.48\textwidth}
  \centering
  \includegraphics[width=\textwidth]{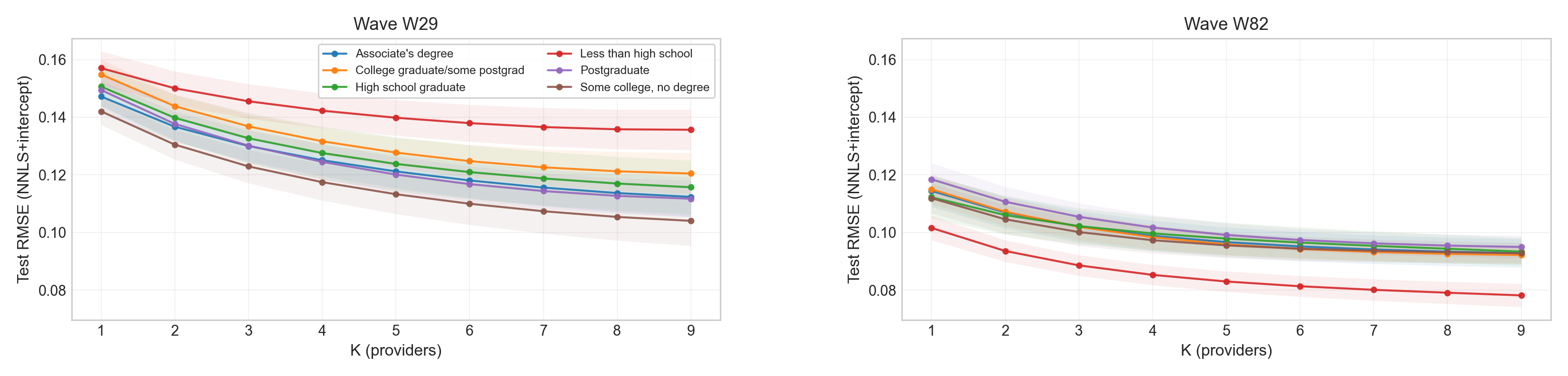}
  \caption{Education (6 groups)}
\end{subfigure}

\vspace{0.4em}
\begin{subfigure}[b]{0.48\textwidth}
  \centering
  \includegraphics[width=\textwidth]{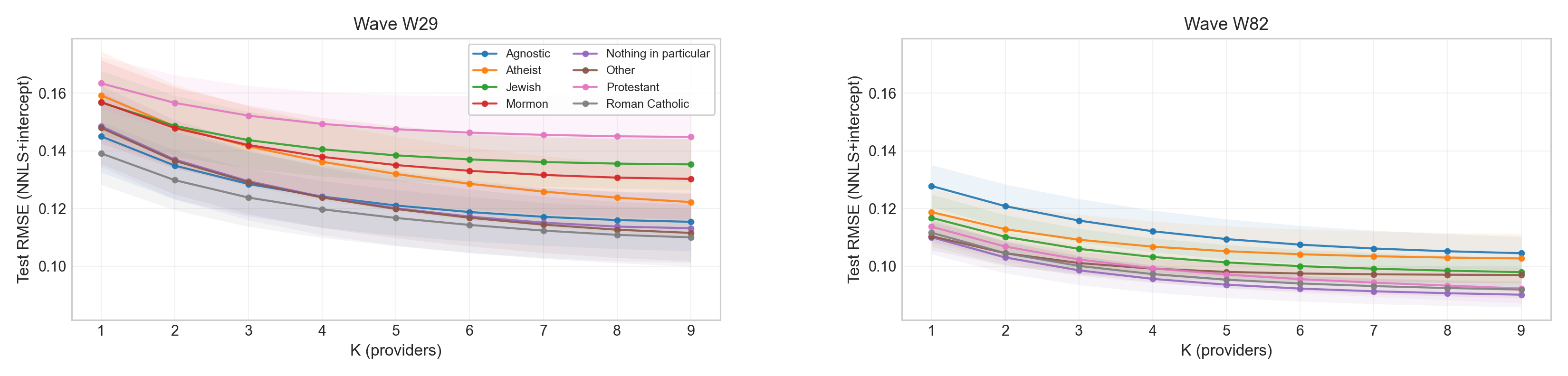}
  \caption{Religion (7--8 groups)}
\end{subfigure}
\hfill
\begin{subfigure}[b]{0.48\textwidth}
  \centering
  \includegraphics[width=\textwidth]{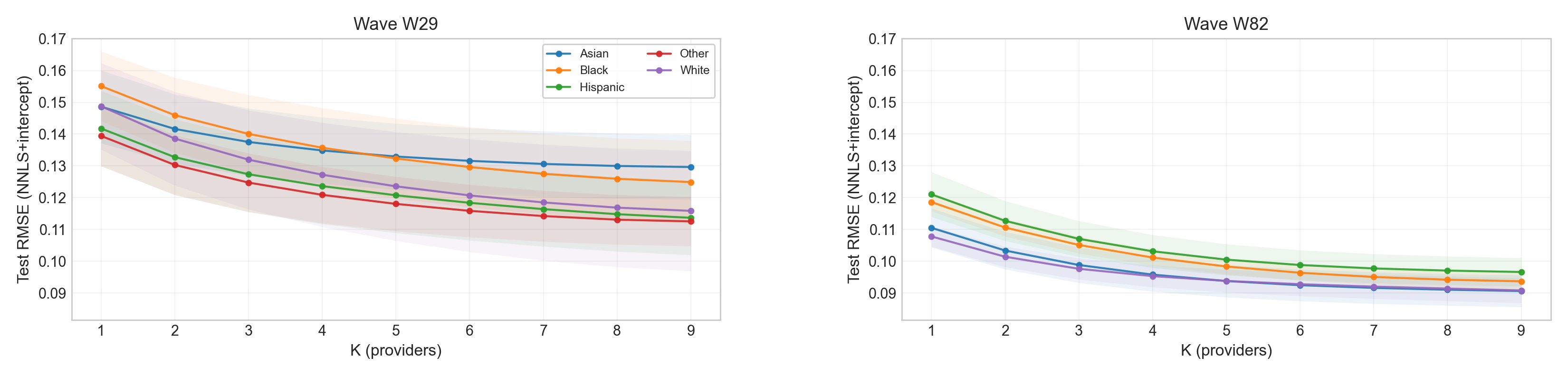}
  \caption{Race (5 groups)}
\end{subfigure}
\caption{\textbf{Per-group RMSE (part 1).} Some groups are consistently harder to approximate. Left panels: W29; right panels: W82.}
\label{fig:pergroup-all}
\end{figure}

\begin{figure}[t]
\centering
\begin{subfigure}[b]{0.48\textwidth}
  \centering
  \includegraphics[width=\textwidth]{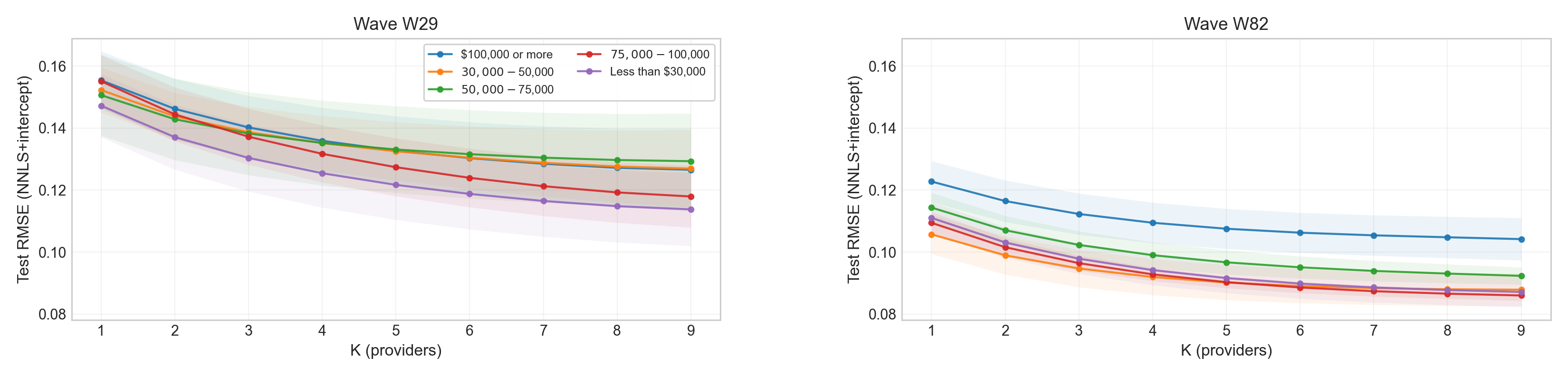}
  \caption{Income (5 groups)}
\end{subfigure}
\hfill
\begin{subfigure}[b]{0.48\textwidth}
  \centering
  \includegraphics[width=\textwidth]{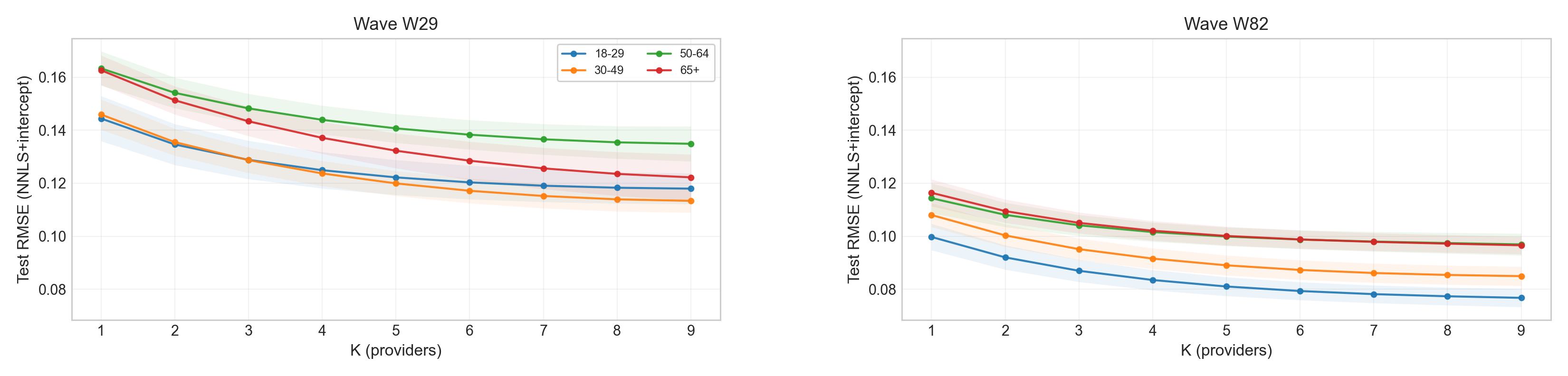}
  \caption{Age (4 groups)}
\end{subfigure}

\vspace{0.4em}
\begin{subfigure}[b]{0.48\textwidth}
  \centering
  \includegraphics[width=\textwidth]{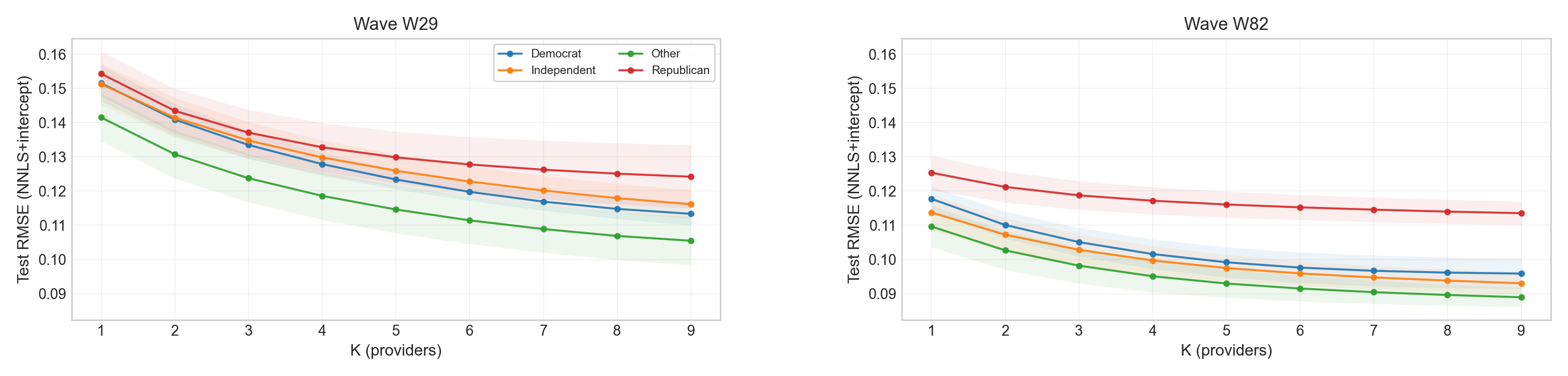}
  \caption{Pol.\ Party (4 groups)}
\end{subfigure}
\hfill
\begin{subfigure}[b]{0.48\textwidth}
  \centering
  \includegraphics[width=\textwidth]{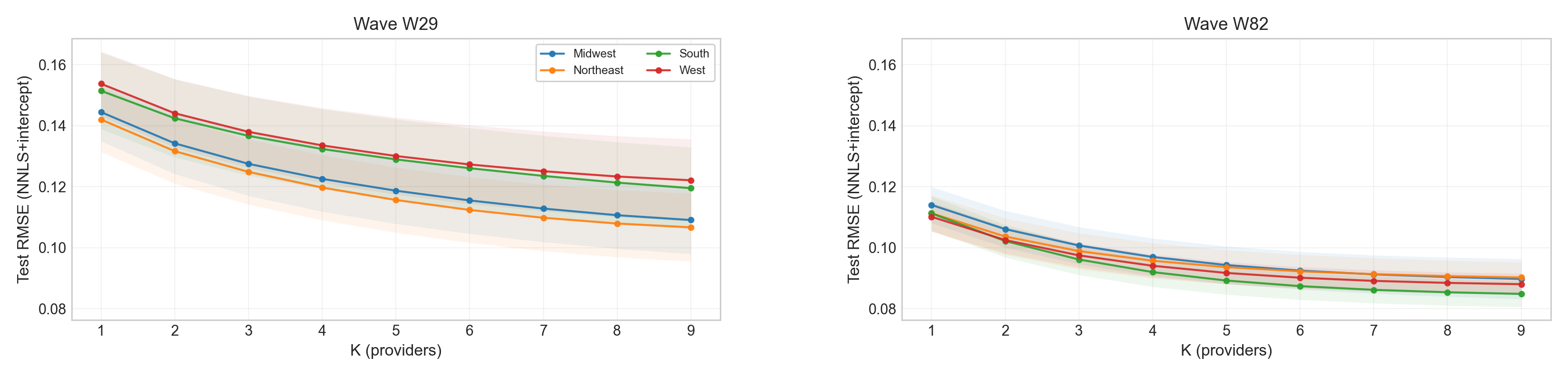}
  \caption{Census Region (4 groups)}
\end{subfigure}
\caption{\textbf{Per-group RMSE (part 2).} Continued from Figure~\ref{fig:pergroup-all}.}
\label{fig:pergroup-all-2}
\end{figure}

\paragraph{Transfer factors across demographics.}
Figure~\ref{fig:transfer-all-demos} shows the transfer factor $1/\lambda_i^*$ vs.\ $K$ for seven partitions (excluding Sex, where 2 groups make coverage comparatively easy).
All show the characteristic sharp drop at small $K$.
Religion requires more providers ($K=4-5$) to stabilize than other partitions ($K=2-3$), consistent with its larger number of groups. Because the anonymous-game slack scales as $2\eps/\lambda_i^*$, this cross-demographic variation in transfer factors directly translates into weaker bounds for some demographic partitions.

\begin{figure}[t]
\centering
\begin{subfigure}[b]{0.48\textwidth}
  \centering
  \includegraphics[width=\textwidth]{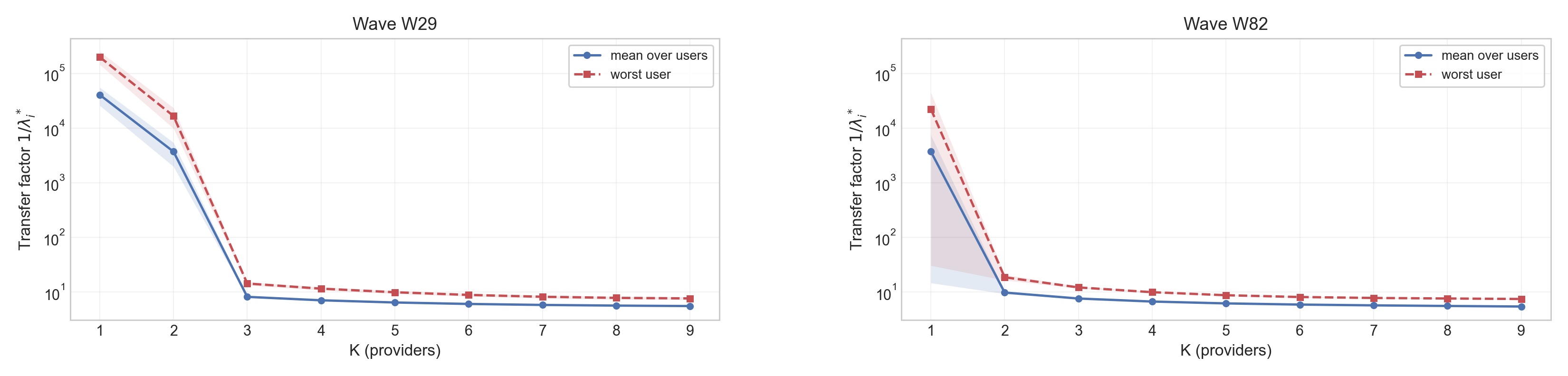}
  \caption{Education (6 groups)}
\end{subfigure}
\hfill
\begin{subfigure}[b]{0.48\textwidth}
  \centering
  \includegraphics[width=\textwidth]{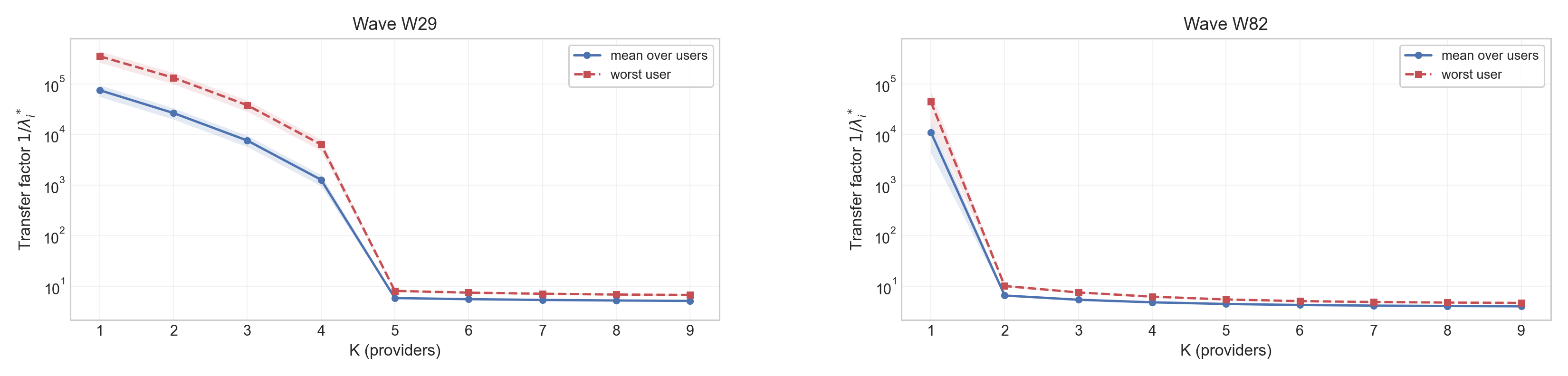}
  \caption{Race (5 groups)}
\end{subfigure}

\vspace{0.4em}
\begin{subfigure}[b]{0.48\textwidth}
  \centering
  \includegraphics[width=\textwidth]{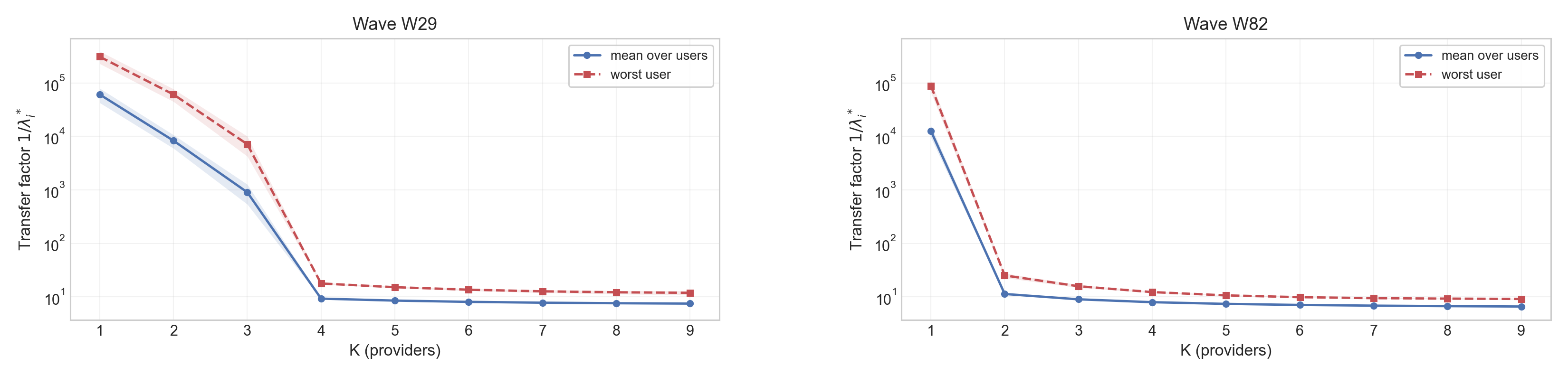}
  \caption{Religion (7--8 groups)}
\end{subfigure}
\hfill
\begin{subfigure}[b]{0.48\textwidth}
  \centering
  \includegraphics[width=\textwidth]{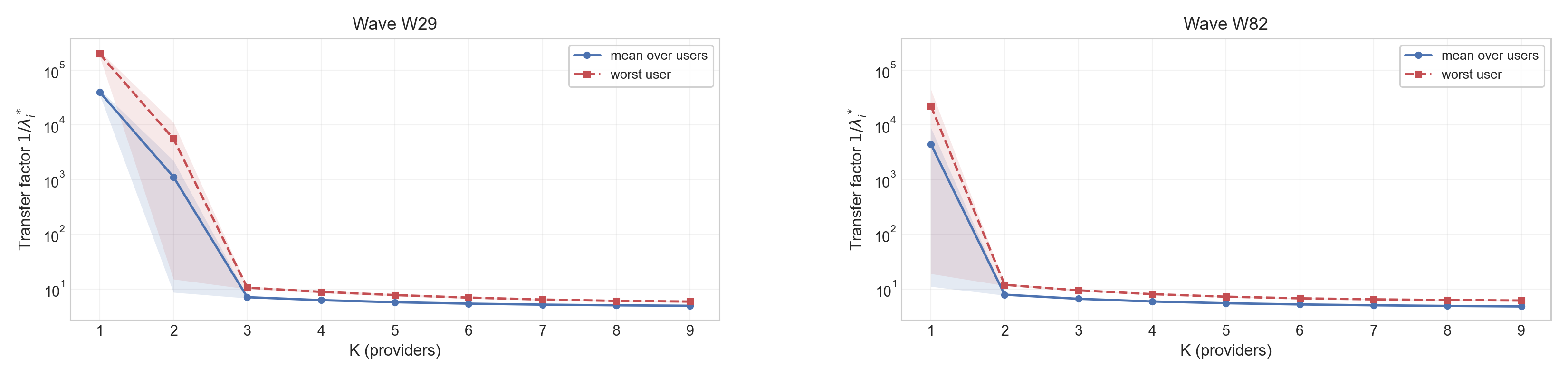}
  \caption{Income (5 groups)}
\end{subfigure}

\vspace{0.4em}
\begin{subfigure}[b]{0.48\textwidth}
  \centering
  \includegraphics[width=\textwidth]{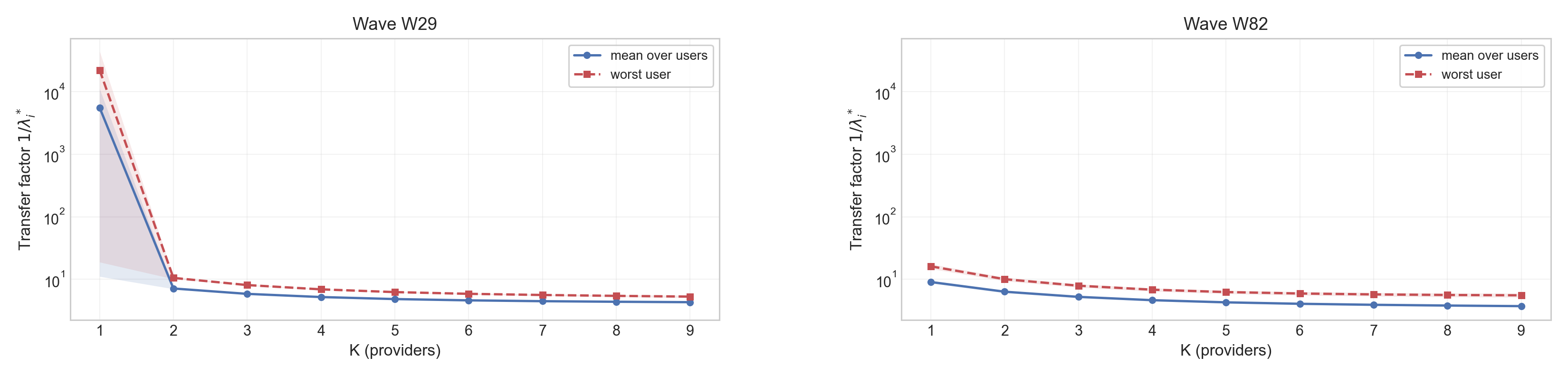}
  \caption{Age (4 groups)}
\end{subfigure}
\hfill
\begin{subfigure}[b]{0.48\textwidth}
  \centering
  \includegraphics[width=\textwidth]{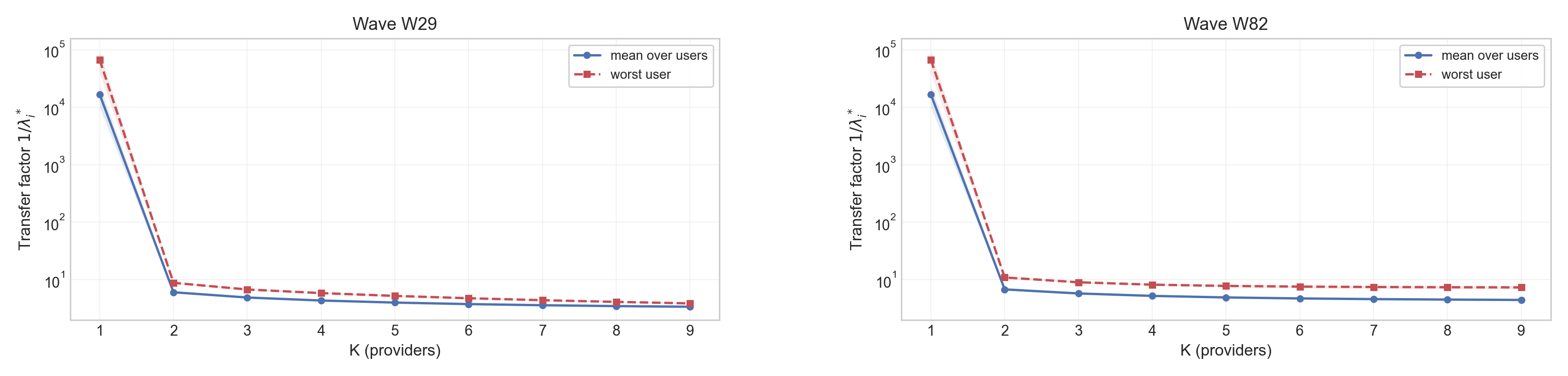}
  \caption{Pol.\ Party (4 groups)}
\end{subfigure}

\vspace{0.4em}
\begin{subfigure}[b]{0.48\textwidth}
  \centering
  \includegraphics[width=\textwidth]{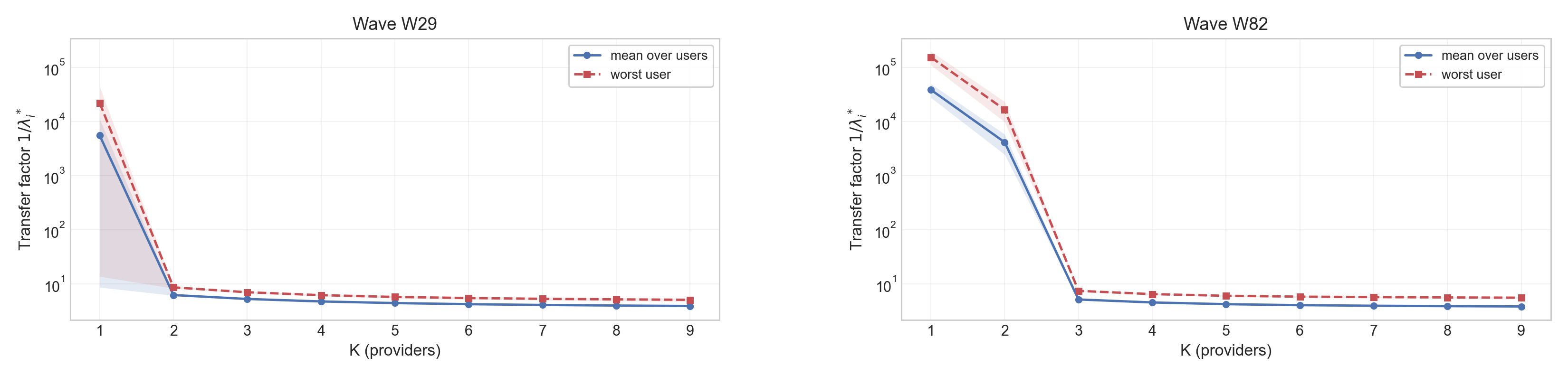}
  \caption{Census Region (4 groups)}
\end{subfigure}
\caption{\textbf{Transfer factors across demographics.} $1/\lambda_i^*$ vs.\ $K$: mean (solid) and worst user (dashed). All partitions show the sharp drop at small $K$. Religion requires $K=4{-}5$ to stabilize.}
\label{fig:transfer-all-demos}
\end{figure}

\paragraph{Subset analyses.}
Figure~\ref{fig:subset-all} shows the subset analysis for all seven partitions. Religion is the hardest: the best single provider has transfer $\sim20-30$ (vs.\ $\sim7-10$ for Political Ideology in the main text), and the mean subset needs $|T|=4-5$ to stabilize.
Partitions with fewer groups (Age, Census Region, Political Party) stabilize at $|T|=2-3$.
These curves are consistent with the theoretical prediction that guarantees depend on provider diversity rather than universality: a small well-chosen subset can suffice, while a random subset of the same size often leaves large coverage gaps.
\begin{figure}[t]
\centering
\begin{subfigure}[b]{0.48\textwidth}
  \centering
  \includegraphics[width=\textwidth]{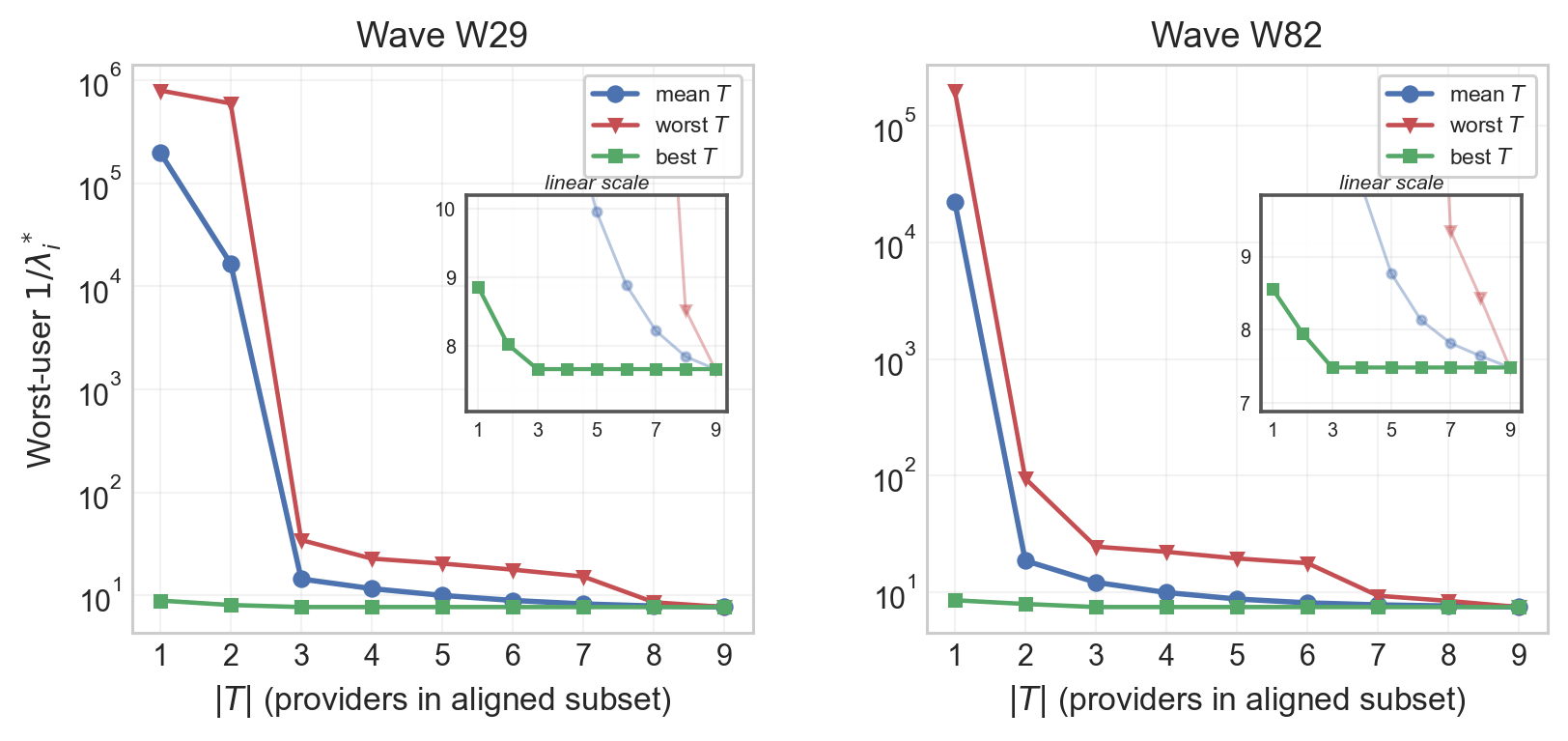}
  \caption{Education (6 groups)}
\end{subfigure}
\hfill
\begin{subfigure}[b]{0.48\textwidth}
  \centering
  \includegraphics[width=\textwidth]{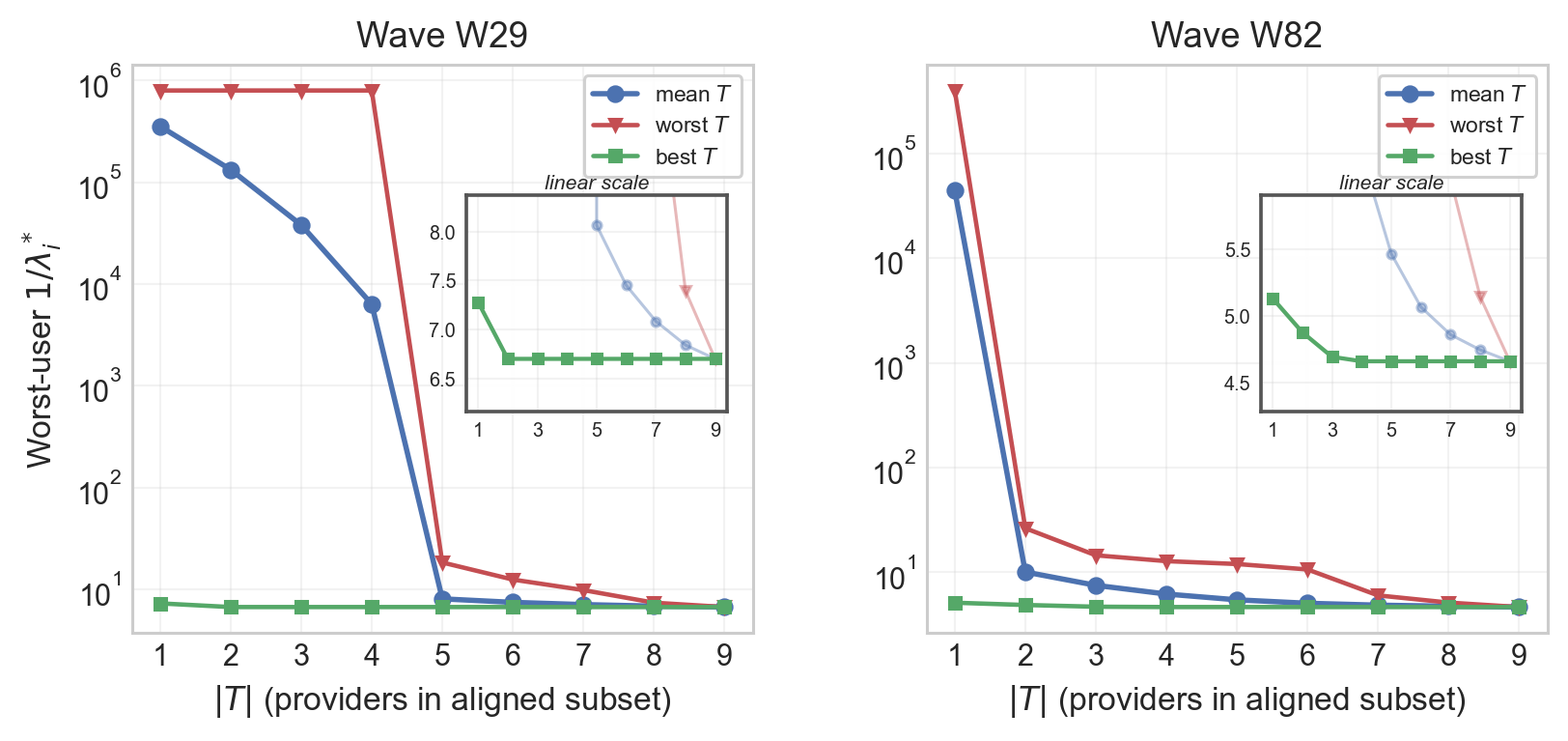}
  \caption{Race (5 groups)}
\end{subfigure}

\vspace{0.4em}
\begin{subfigure}[b]{0.48\textwidth}
  \centering
  \includegraphics[width=\textwidth]{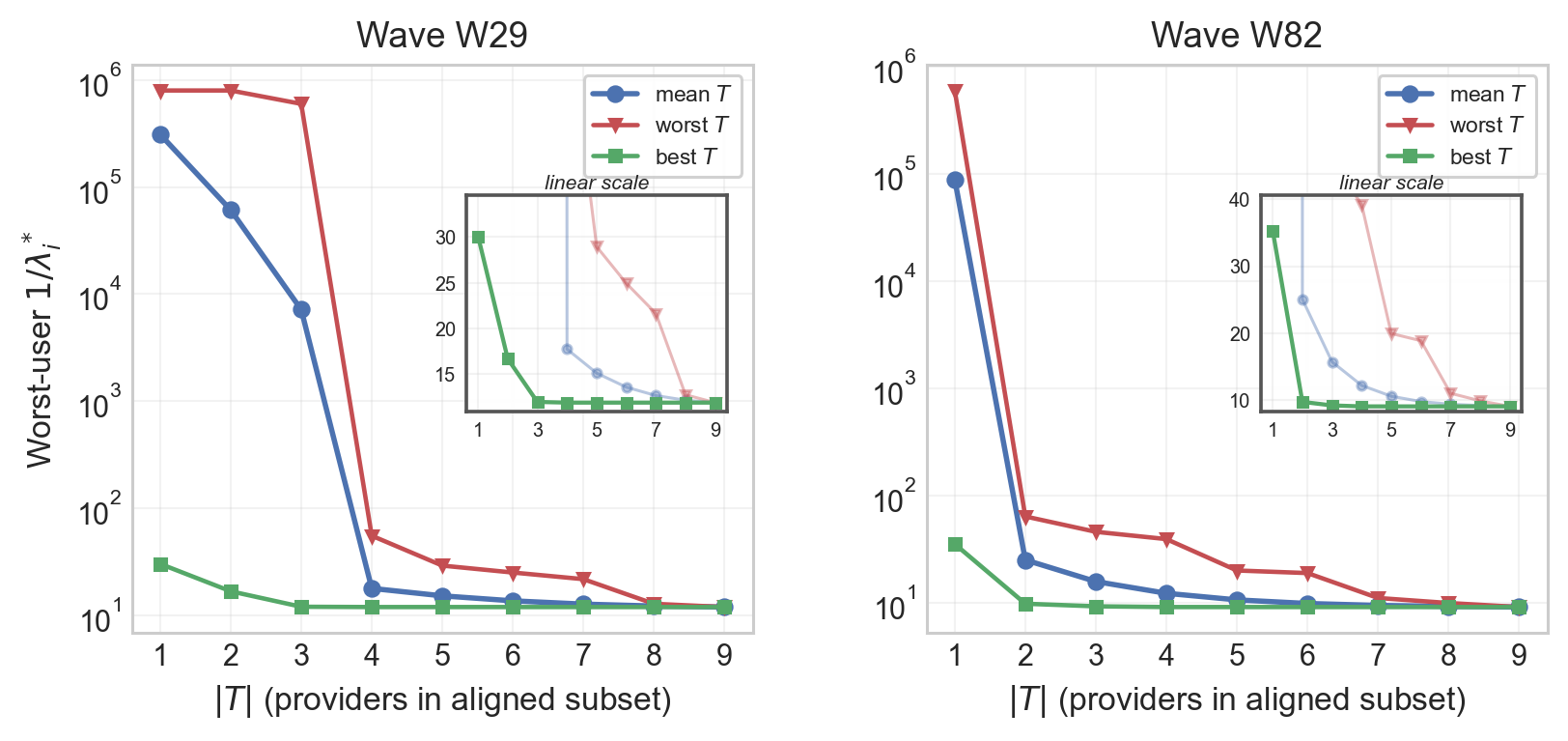}
  \caption{Religion (7--8 groups)}
\end{subfigure}
\hfill
\begin{subfigure}[b]{0.48\textwidth}
  \centering
  \includegraphics[width=\textwidth]{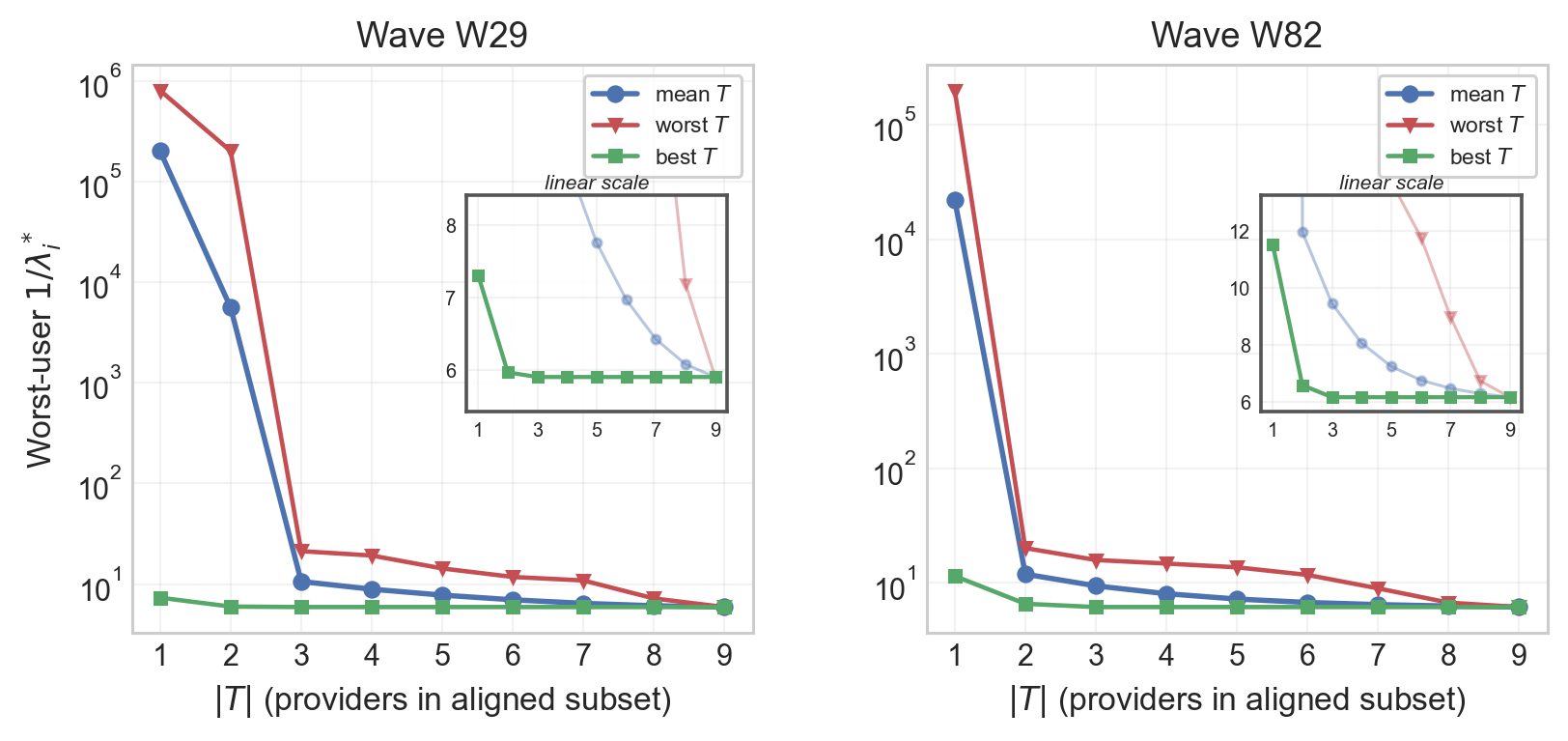}
  \caption{Income (5 groups)}
\end{subfigure}

\vspace{0.4em}
\begin{subfigure}[b]{0.48\textwidth}
  \centering
  \includegraphics[width=\textwidth]{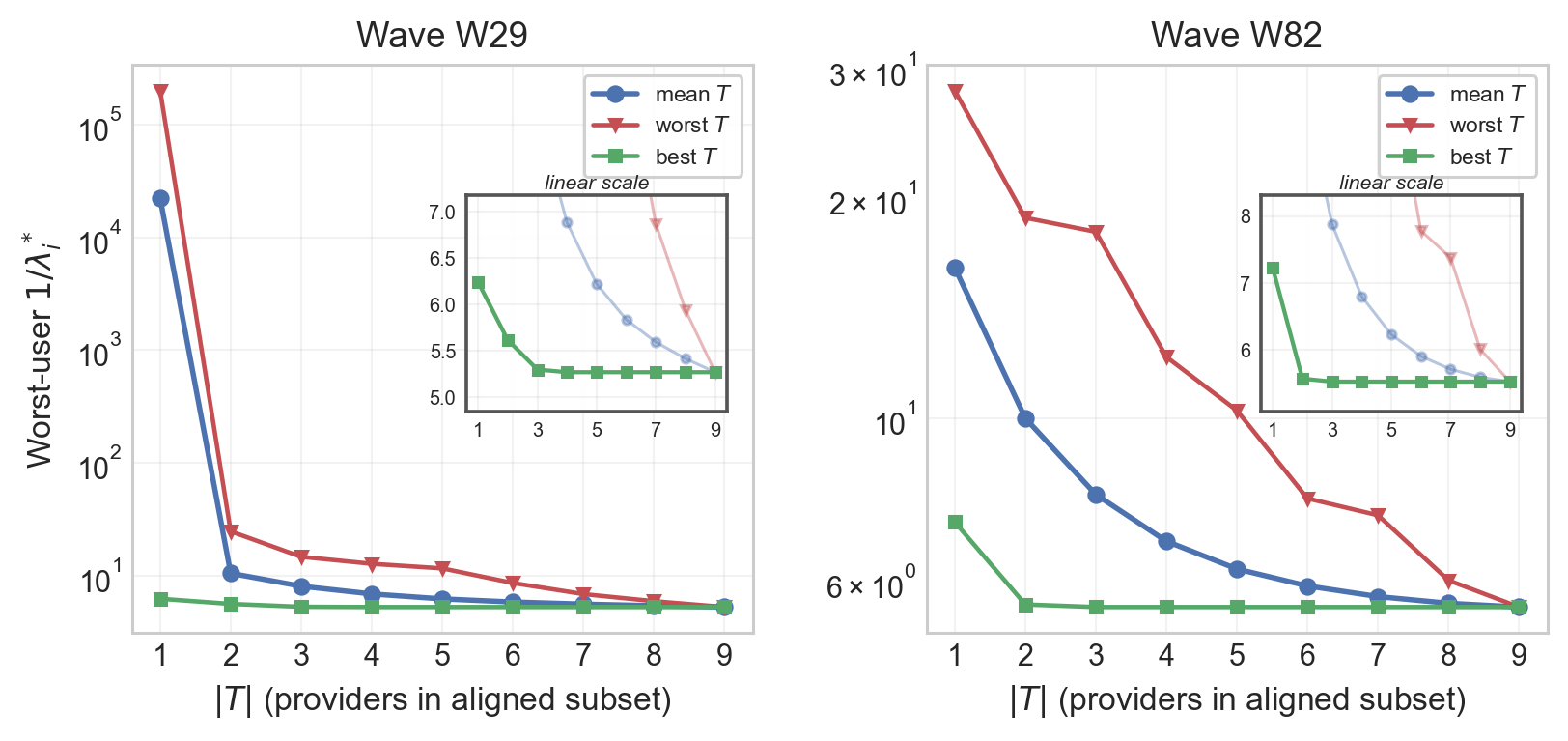}
  \caption{Age (4 groups)}
\end{subfigure}
\hfill
\begin{subfigure}[b]{0.48\textwidth}
  \centering
  \includegraphics[width=\textwidth]{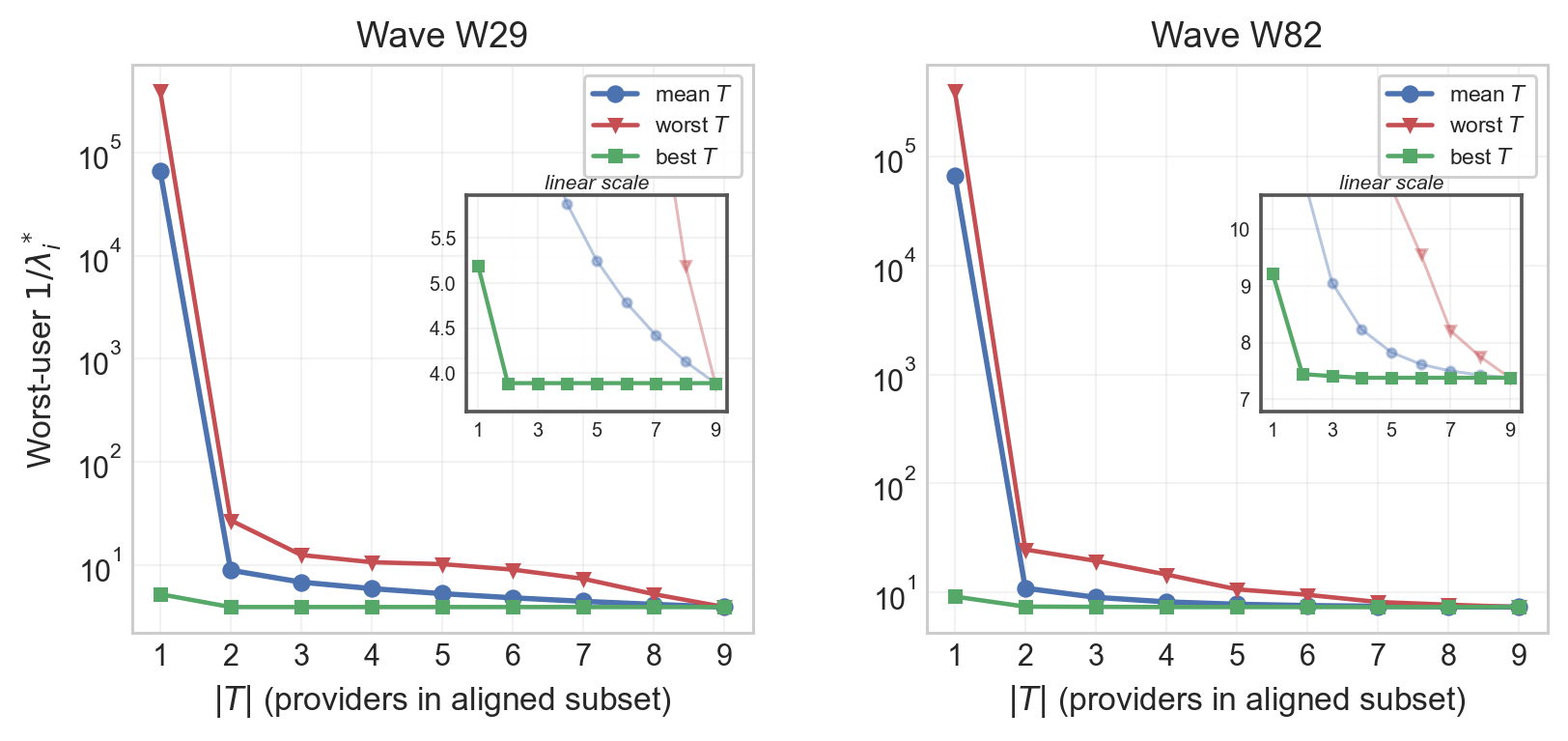}
  \caption{Pol.\ Party (4 groups)}
\end{subfigure}

\vspace{0.4em}
\begin{subfigure}[b]{0.48\textwidth}
  \centering
  \includegraphics[width=\textwidth]{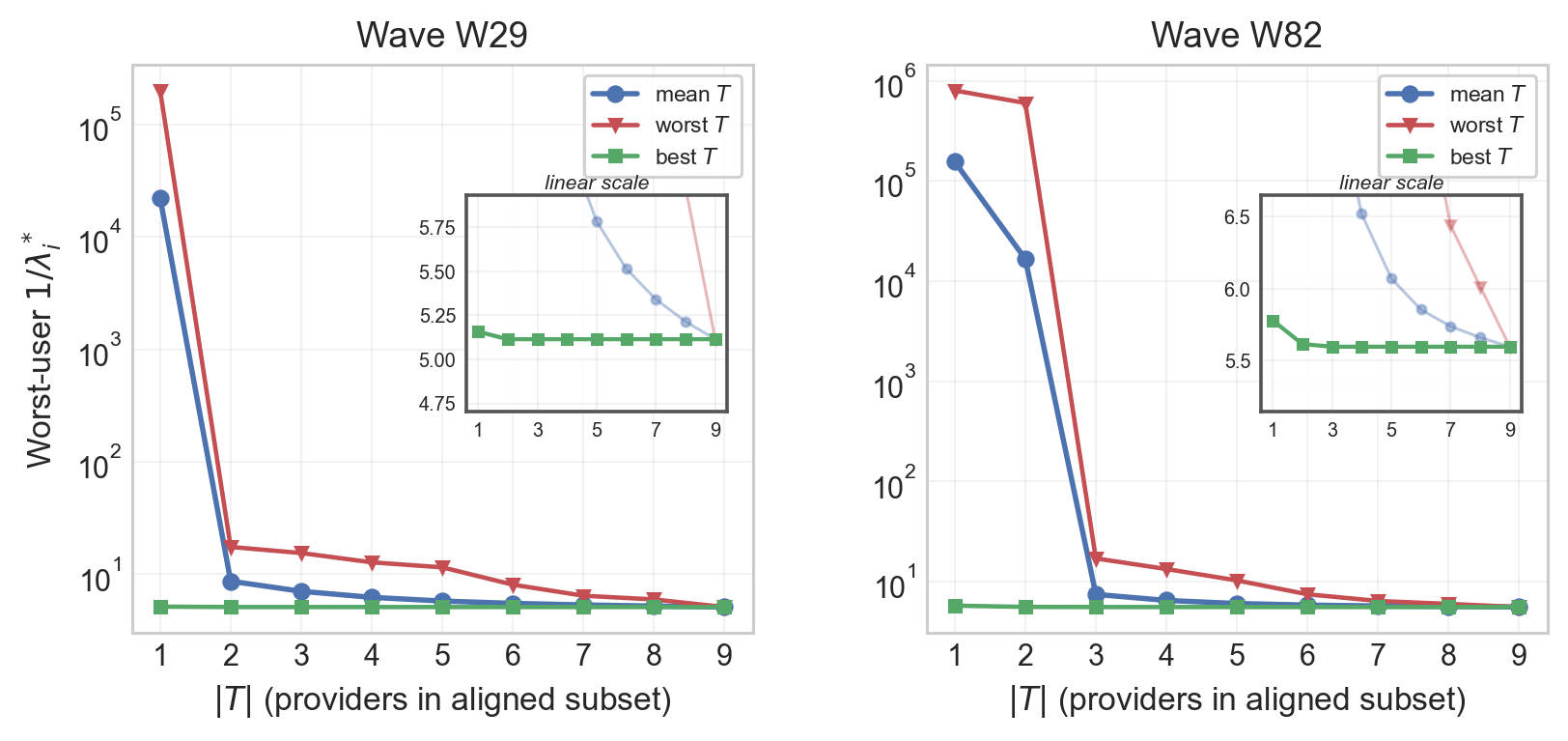}
  \caption{Census Region (4 groups)}
\end{subfigure}
\caption{\textbf{Subset analyses across demographics.} Worst-user $1/\lambda_i^*$ vs.\ $|T|$. Insets show linear-scale zoom. Coverage difficulty scales with number of groups.}
\label{fig:subset-all}
\end{figure}

\paragraph{Learned alignment weights.}
We visualize the learned NNLS weights to understand model specialization (Figure~\ref{fig:weights-appendix}).
For weak alignment, $w_{j,i}$ (rows = users, columns = models) shows which providers contribute to approximating each group.
For strong alignment, $\lambda_{j,i}$ (rows = models, columns = users) shows which groups each model covers.

In Political Ideology, W29 shows sparse specialization: several models have zero weight for some groups, while W82 has more distributed weights.
Religion reveals that minority groups (Mormon, Jewish) rely heavily on specific models, explaining their higher transfer factors.

\begin{figure}[t]
\centering
\begin{subfigure}[b]{0.48\textwidth}
  \centering
  \includegraphics[width=\textwidth]{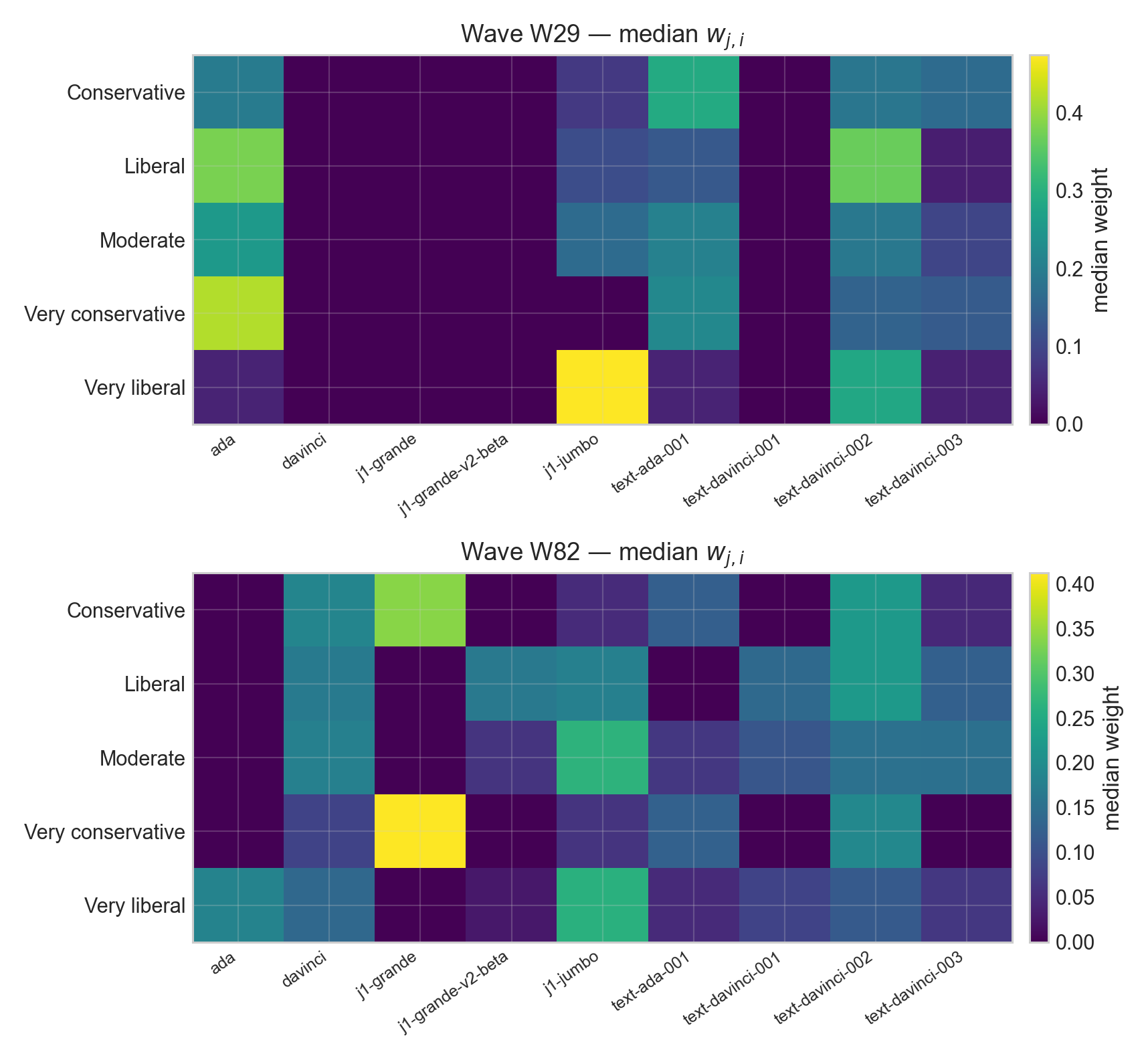}
  \caption{$w_{j,i}$ (Pol.\ Ideology)}
\end{subfigure}
\hfill
\begin{subfigure}[b]{0.48\textwidth}
  \centering
  \includegraphics[width=\textwidth]{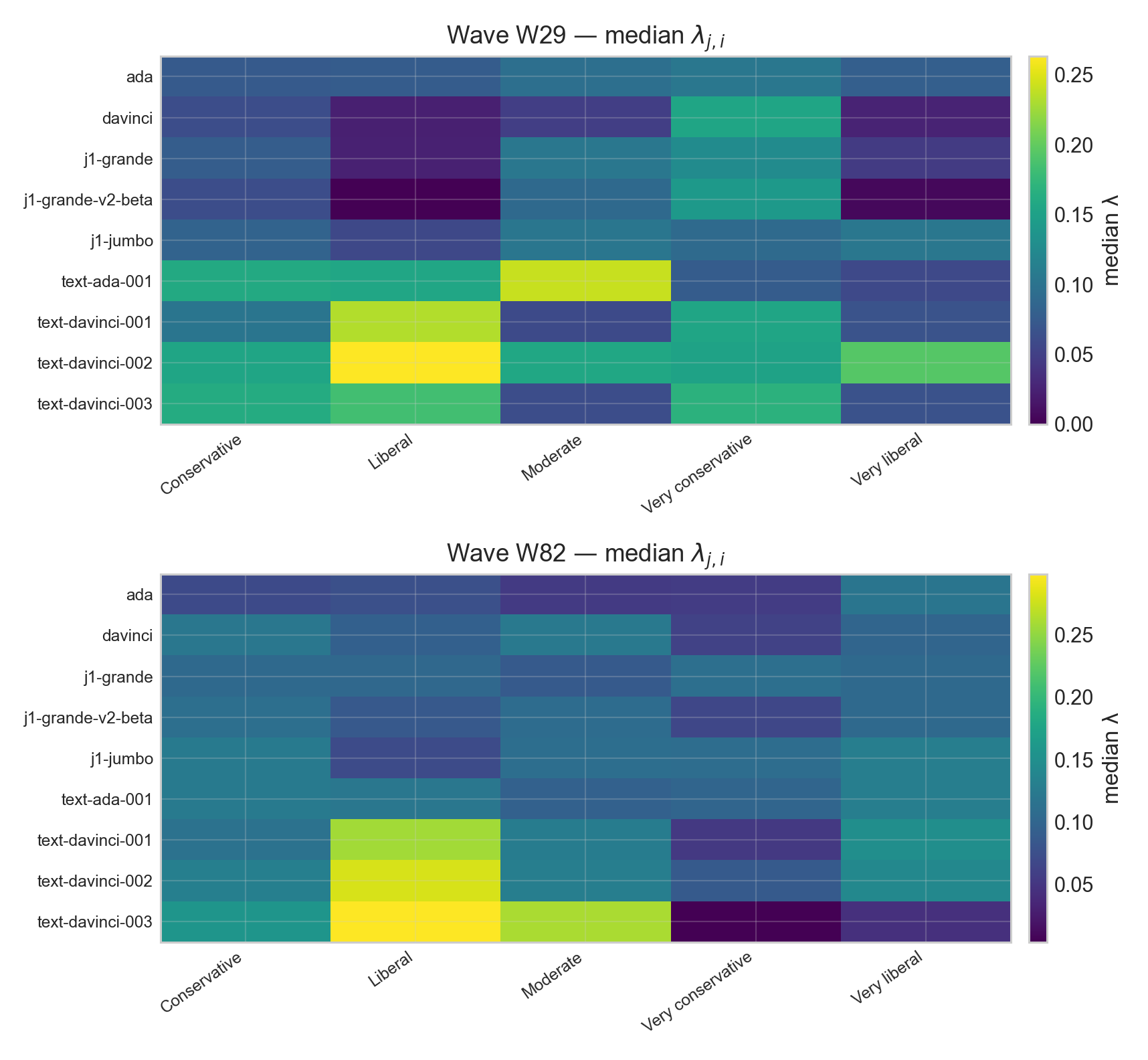}
  \caption{$\lambda_{j,i}$ (Pol.\ Ideology)}
\end{subfigure}

\vspace{0.4em}
\begin{subfigure}[b]{0.48\textwidth}
  \centering
  \includegraphics[width=\textwidth]{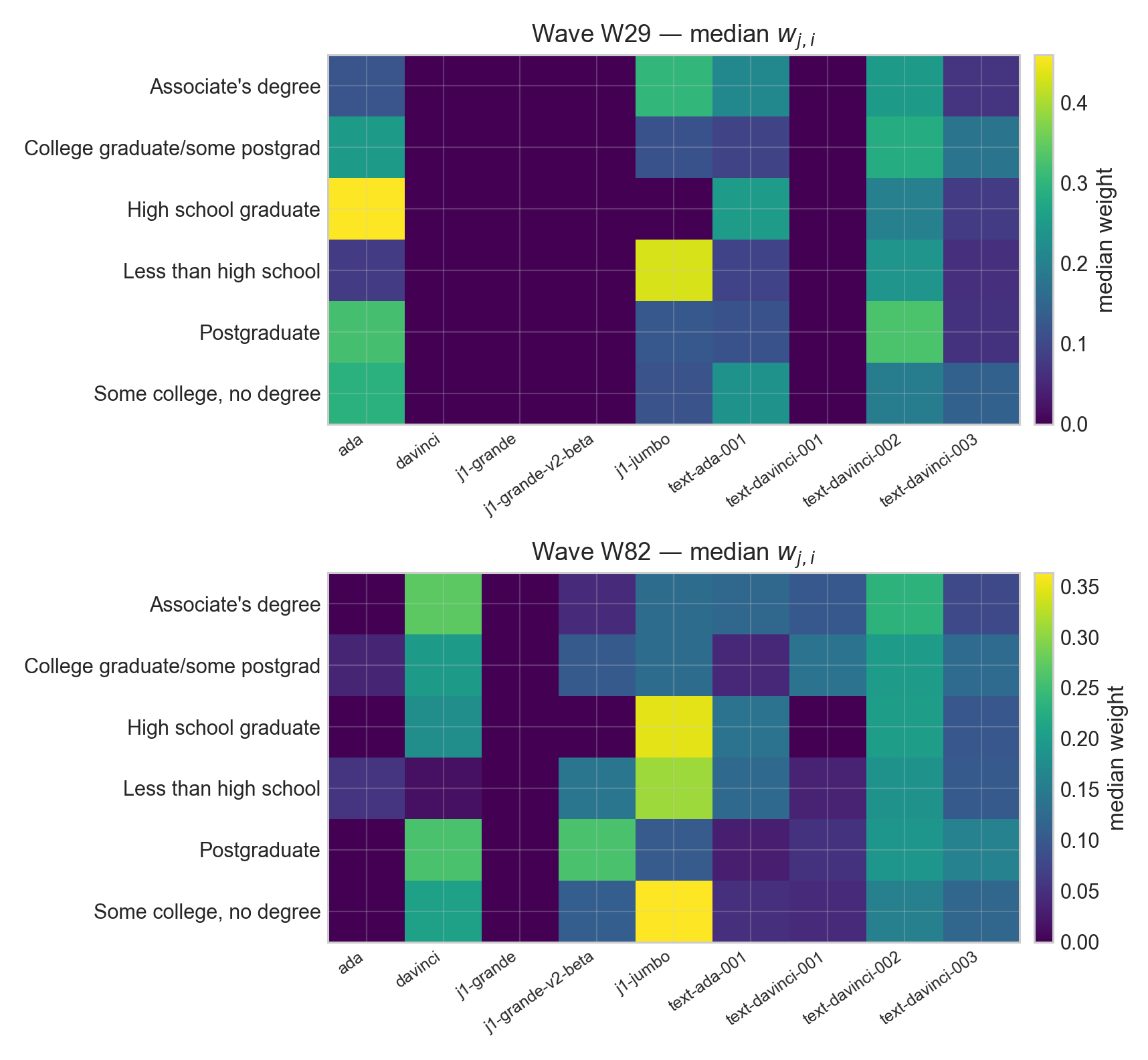}
  \caption{$w_{j,i}$ (Education)}
\end{subfigure}
\hfill
\begin{subfigure}[b]{0.48\textwidth}
  \centering
  \includegraphics[width=\textwidth]{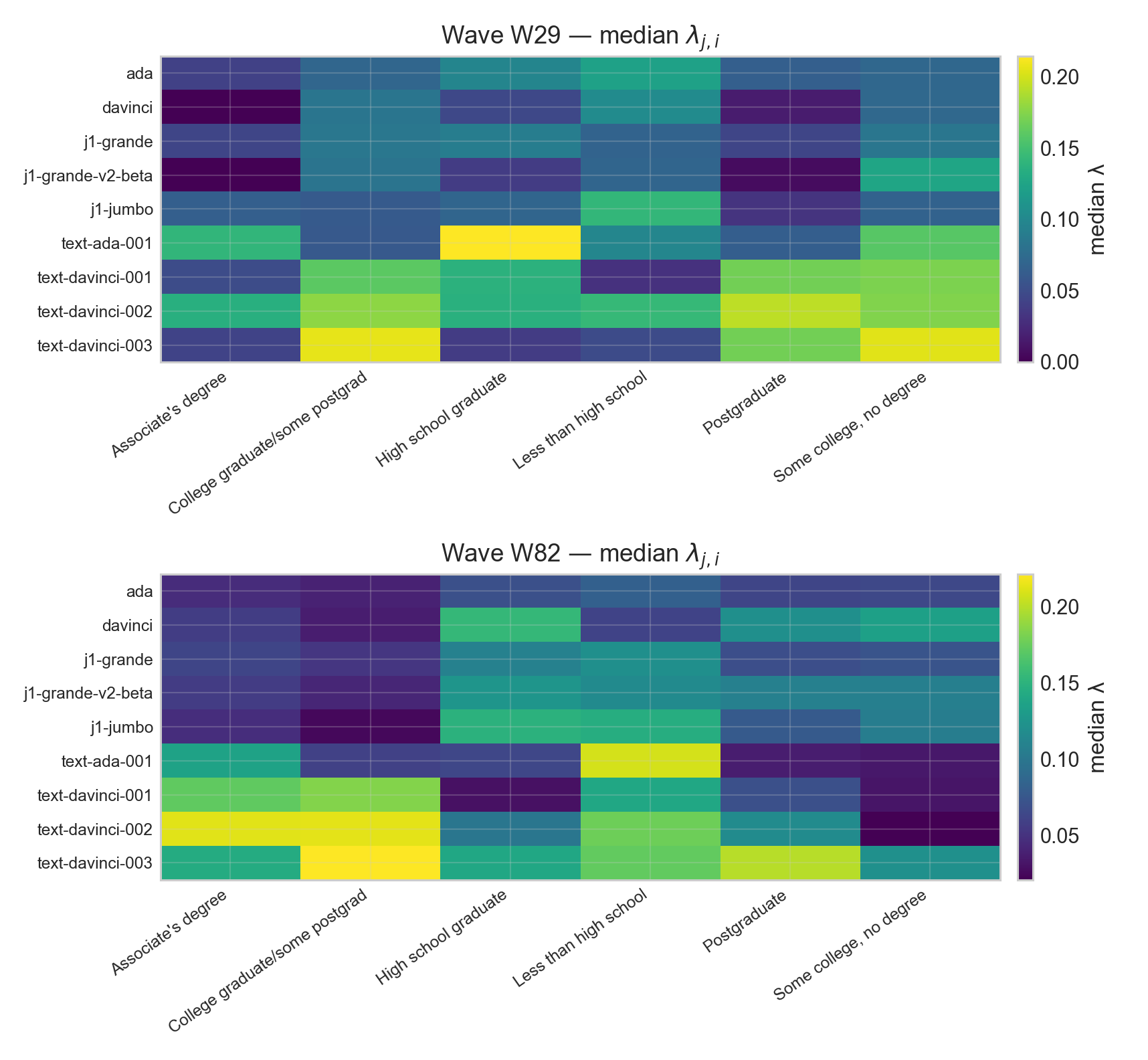}
  \caption{$\lambda_{j,i}$ (Education)}
\end{subfigure}
\caption{\textbf{Learned alignment weights.} Left: weak $w_{j,i}$ (rows = users, columns = models). Right: strong $\lambda_{j,i}$ (rows = models, columns = users). Each heatmap shows W29 (top) and W82 (bottom). Different models contribute to different groups; W29 shows sparser specialization than W82.}
\label{fig:weights-appendix}
\end{figure}